\documentclass[11pt,a4paper,reqno]{amsart}%
\usepackage{amsthm,amsmath,amsfonts,amssymb,appendix,bookmark,dsfont}
\usepackage{color}
\usepackage{bbm}

\theoremstyle{plain}
\newtheorem{theorem}{Theorem}
\newtheorem{lemma}[theorem]{Lemma}

\newtheorem{proposition}[theorem]{Proposition}

\newtheorem{remark}[theorem]{Remark}

\theoremstyle{definition}

\newcounter{step} 

\usepackage{etoolbox}
\AtBeginEnvironment{proof}{\setcounter{step}{0}}

\numberwithin{equation}{section}
\numberwithin{theorem}{section}

\DeclareMathOperator{\supp}{supp}
\DeclareMathOperator{\Tr}{Tr}

\def\bq{\begin{eqnarray}}
\def\eq{\end{eqnarray}}
\def\bqq{\begin{align*}}
\def\eqq{\end{align*}}

\def\nn{\nonumber}

\def\eps{\varepsilon}
\def\wto{\rightharpoonup}

\newcommand{\ud}{\mathrm{d}}
\renewcommand{\Re}{\operatorname{Re}}

\newcommand\1{{\ensuremath {\mathds 1} }}

\def\bra#1{{\langle#1|}}
\def\ket#1{{|#1\rangle}}

\def\bU{\mathbb{U}}

\def\cF {\mathcal{F}}
\def\cB {\mathcal{B}}
\def\cH{\mathcal{H}}
\def\cV {\mathcal{V}}
\def\R {\mathbb{R}}

\def\cN {\mathcal{N}}

\def\cD{\mathcal{D}}

\def\cG {\mathcal{G}}

\def\cP {\mathcal{P}}
\def\cE {\mathcal{E}}
\def\cA{\mathcal{A}}

\def\cK{\mathcal{K}}

\def\hess{\operatorname{Hess}\mathcal{E}^{\operatorname{H}}[u_0,v_0]}

\def\fh{\mathfrak{h}}

\def\fh {\mathfrak{h}}
\def\bN{\mathbb{N}}
\def\cV {\mathcal{V}}
\def\R {\mathbb{R}}

\def\J {\mathcal{J}}

\def\cA{\mathcal{A}}

\def\d{{\, \rm d}}
\def\cH{{\, \mathcal{H}}}

\newcommand{\bH}{\mathbb{H}}

\title[] {Ground state energy of mixture of Bose gases}

\author[A. Michelangeli]{Alessandro Michelangeli}

\address{SISSA, International School for Advanced Studies, Via Bonomea 265, 34136 Trieste, Italy}

\email{michel@math.lmu.de, alemiche@sissa.it}

\author[P.T. Nam]{Phan Th\`anh Nam}
\address{Department of Mathematics, LMU Munich, Theresienstrasse 39, 80333 Munich, Germany} 
\email{nam@math.lmu.de}

\author[A. Olgiati]{Alessandro Olgiati}

\address{SISSA, International School for Advanced Studies, Via Bonomea 265, 34136 Trieste, Italy}

\email{aolgiati@sissa.it}

\begin{document}
	
	\begin{abstract}   We consider the asymptotic behavior of a system of multi-component trapped bosons, when the total particle number $N$ becomes large. In the dilute regime, when the interaction potentials have the length scale of order $O(N^{-1})$, we show that the leading order of the ground state energy is captured correctly by the Gross-Pitaevskii energy functional and that the many-body ground state fully condensates on the Gross-Pitaevskii minimizers. In the mean-field regime, when the interaction length scale is $O(1)$, we are able to verify Bogoliubov's approximation and obtain the second order expansion of the ground state energy. While such asymptotic results have several precursors in the literature on one-component condensates, the adaption to the multi-component setting is non-trivial in various respects and the analysis will be presented in details. 
	\end{abstract}
	
	\date{\today}

	\maketitle
	
	\setcounter{tocdepth}{2}
	\tableofcontents

	\section{Introduction}\label{sec:intro}

	Experimental and theoretical investigations on mixtures of Bose gases displaying condensation have made important progress in the past years. 
		The physical systems of interest consist of a gas formed by different species of interacting bosons, each of which is brought to condensation, thus with a macroscopic occupation of a one-body orbital for each species. They can be prepared as atomic gases of the same element, typically $^{87}\mathrm{Rb}$, which occupy two hyperfine states with no interconversion between particles of different hyperfine states \cite{MBGCW-1997,Matthews_HJEWC_DMStringari_PRL1998,HMEWC-1998,Hall-Matthews-Wieman-Cornell_PRL81-1543}, or also as heteronuclear mixtures such as $^{41}\mathrm{K}$-$^{87}\mathrm{Rb}$ \cite{Modugno-Ferrari-Inguscio-etal-Science2001_multicompBEC}, $^{41}\mathrm{K}$-$^{85}\mathrm{Rb}$ \cite{Modugno-PRL-2002}, $^{39}\mathrm{K}$-$^{85}\mathrm{Rb}$ \cite{MTCBM-PRL2004_BEC_heteronuclear}, and $^{85}\mathrm{K}$-$^{87}\mathrm{Rb}$ \cite{Papp-Wieman_PRL2006_heteronuclear_RbRb}. 
	For a comprehensive review of the related physical properties we refer to \cite[Chapter 21]{pita-stringa-2016}.

	Mathematically, the natural Hilbert space for two populations, say, of $N_1$ and $N_2$ identical bosons of different type in three dimensions  is
	\begin{equation} \label{eq:many_body_hs}
	\cH_{N_1,N_2} \;=\; L^2_{\rm sym}(\R^{3N_1}, \d x_1,\dots, \d x_{N_1}) \otimes L^2_{\rm sym}(\R^{3N_2},\d y_1,\dots, \d y_{N_2})\,,
	\end{equation}
	where $L^2_\mathrm{sym}(\mathbb{R}^{3N_j})$ is the space of square-integrable functions in $(\R^3)^{N_j}$ which are symmetric under  permutations of $N_j$ variables. No overall exchange symmetry is present among variables of different type.

	Most of the physically relevant phenomena are modelled by Hamiltonians acting on $\cH_{N_1,N_2}$ which include the customary intra-species and inter-species two-body interaction term, as well as an overall trapping potential for each species -- such confinements can indeed be different for each type of particles. Rigorous statements can only be proved in suitable limits of infinitely many particles, say, when  the population ratios are asymptotically fixed:
	\begin{equation} \label{eq:ratios}
	N\;:=N_1+N_2\;\to\;+\infty\,,\qquad 
	\lim_{N\to +\infty}\frac{N_j}{N}\;=:\;c_j\in(0,1)\,,\qquad j\in\{1,2\}\,.
	\end{equation}
It is not restrictive to assume that the ratios $N_1/N$ and $N_2/N$ are fixed, and so shall we henceforth.


In the present work we intend to investigate the limit of infinitely many particles for two significant scaling regimes: the \emph{mean field}, for interactions of weak magnitude and long range, and the \emph{dilute} scaling, for interactions of strong magnitude and short range. In principle, the former regime is mathematically easier because of the natural emergence of the law of large number, while the latter regime is more relevant to realistic experiments of the Bose-Einstein condensation.

	To be precise, we shall consider the following Hamiltonian (in suitable units) 
	\begin{equation} \label{eq:many_body_hamiltonian}
	\begin{split}
	H_{N} \;&:=\; \sum_{j=1}^{N_1}\big(-\Delta_{x_j}+U^{(1)}_{\operatorname{trap}}(x_j)\big) +  \frac{1}{N} \sum_{1\leqslant j<r \leqslant N_1}  V^{(1)}_N(x_i-x_j) \\
	&\quad +\sum_{k=1}^{N_2}\big(-\Delta_{y_k}+U^{(2)}_{\operatorname{trap}}(y_k)\big) + \frac{1}{N} \sum_{1\leqslant k<\ell \leqslant N_2}  V^{(2)}_N(y_k-y_\ell) \\
	&\quad + \frac{1}{N} \sum_{j=1}^{N_1}\sum_{k=1}^{N_2} \,V^{(12)}_N (x_j-y_k)\,,
	\end{split}
	\end{equation}
	with self-explanatory kinetic, confining, and interaction terms. Under our assumptions, $H_N$ will be bounded from below on the core domain of smooth, compactly supported functions and it is unambiguously realised as a self-adjoint operator by Friedrichs's extension, still denoted by $H_N$.  	
	
	Correspondingly, we shall specialise \eqref{eq:many_body_hamiltonian} in two forms, the mean field Hamiltonian $H_N^{\operatorname{MF}}$ and the Gross-Pitaevskii  Hamiltonian $H_N^{\operatorname{GP}}$, in which the interaction potentials $V^{(\alpha)}_N$ are chosen as \footnote{In the physical literature \cite{M-Olg-2016_2mixtureMF,AO-GPmixture-2016volume,MO-pseudospinors-2017} a different convention is preferred for $\alpha\in \{1,2\}$, namely $V^{(\alpha)}= c_\alpha^{-1}\cV^{(\alpha)}$ for mean-field regime, and $V^{(\alpha)} = c_\alpha^2 \cV^{(\alpha)}(c_\alpha \cdot)$ for Gross-Pitaevskii regime.}
	\begin{equation}\label{eq:scaling_scheme}
	\begin{array}{rl}
	V_N^{(\alpha),\operatorname{MF}}(x)\;:=&\!\!V^{(\alpha)}(x),\\
	V_N^{(\alpha),\operatorname{GP}}(x)\;:=&\!\!N^3V^{(\alpha)}(Nx),
	\end{array}\qquad\qquad \alpha\in\{1,2,12\}
	\end{equation}
	for suitable $N$-independent potentials $V^{(1)}$, $V^{(2)}$, and $V^{(12)}$.  
%
	
	We are interested in the large-$N$ behavior of the ground state energies
	\begin{equation}
	\begin{split}
	E_{N}^{\operatorname{MF}}\;&:=\;\inf\sigma( {H}_{N}^{\operatorname{MF}}), \\
	E_{N}^{\operatorname{GP}}\;&:=\;\inf\sigma( {H}_{N}^{\operatorname{GP}})\,,
	\end{split}
	\end{equation}
	and the corresponding ground states. This subject has been deeply investigated for one-component Bose gases, as we shall address below, but is virtually unexplored for multi-component gases.


%

	It turns out that by analogy with the one-component case, the leading order of the ground-state energies are captured by effective energy functionals which describe the condensation at the one-body level. More precisely, in the mean-field regime we have
the Hartree functional
	\begin{equation}\label{eq:functionals-MF}
	\begin{split}
	\cE^{\operatorname{H}}[u,v]\;&:=\;c_1\int_{\mathbb{R}^3}|\nabla u|^2\,\ud x+c_1\int_{\mathbb{R}^3}U^{(1)}_{\mathrm{trap}}|u|^2\,\ud x+\frac{\,c_1^2}{2}\int_{\mathbb{R}^3}(V^{(1)}\!*\!|u|^2)\,|u|^2\,\ud x \\ 
	&\qquad + c_2\int_{\mathbb{R}^3}|\nabla v|^2\,\ud x+c_2\int_{\mathbb{R}^3}U^{(2)}_{\mathrm{trap}}|u|^2\,\ud x+\frac{\,c_2^2}{2}\int_{\mathbb{R}^3}(V^{(2)}\!*\!|v|^2)|v|^2\,\ud x \\ 
	&\qquad + c_1c_2\int_{\mathbb{R}^3}(V^{(12)}\!*\!|v|^2)|u|^2\,\ud x
	\end{split}
	\end{equation}
	where $c_1$ and $c_2$ are the population ratios defined in \eqref{eq:scaling_scheme}. On the other hand, in the dilute regime,  we have the Gross-Pitaevskii functional
	\begin{equation}\label{eq:functionals-GP}
	\begin{split}
	\cE^{\operatorname{GP}}[u,v]\;&:=\;c_1\int_{\mathbb{R}^3}|\nabla u|^2\,\ud x+c_1\int_{\mathbb{R}^3}U^{(1)}_{\mathrm{trap}}|u|^2\,\ud x+4\pi a_1c_1^2\int_{\mathbb{R}^3}|u|^4\,\ud x \\ 
	&\qquad + c_2\int_{\mathbb{R}^3}|\nabla v|^2\,\ud x+c_2\int_{\mathbb{R}^3}U^{(2)}_{\mathrm{trap}}|u|^2\,\ud x+4\pi a_2c_2^2\int_{\mathbb{R}^3}|v|^4\,\ud x \\ 
	&\qquad + 8\pi a_{12}c_1c_2\int_{\mathbb{R}^3}|v|^2|u|^2\,\ud x.
	\end{split}
	\end{equation}
Here $a_{\alpha}$ is the ($s$-wave) scattering length of the potential $V^{(\alpha)}$, $\alpha\in \{1,2,12\}$, which is defined by the variational problem (see \cite{LSeSY-ober} for a detailed discussion)
\begin{equation} \label{eq:def-a}
4\pi a_\alpha = \inf \Big\{ \int_{\R^3} \Big[ |\nabla f(x)|^2 + \frac{1}{2}V^\alpha(x) |f(x)|^2 \Big] \d x :\,\,\, f:\R^3\to \R, \lim_{|x|\to \infty} f(x)= 1\Big \}.
\end{equation}

		We are going to show that under suitable, physically realistic conditions, the many-body ground state energy  per particle $E^{\mathrm{MF}}_N/N$ or $E^{\mathrm{GP}}_N/N$ converges to the ground state energy of the Hartree/Gross-Pitaevskii functional.  
Moreover, the corresponding many-body ground state $\psi_N$ condensates on the unique Hartree/Gross-Pitaevskii minimisers $(u_0,v_0)$ in terms of the reduced density matrices:
\begin{equation}\label{eq:g1convergesH-GP-intro}
			\lim_{N\to +\infty}\gamma_{\psi_N}^{(k,\ell)}\;=\;\ket{u_0^{\otimes k}\otimes v_0^{\otimes \ell}}\bra{u_0^{\otimes k}\otimes v_0^{\otimes \ell}}, \quad \forall k,\ell=0,1,2,...
			\end{equation}
Here recall that for every $k,\ell\geqslant 0$, the reduced density matrix $\gamma^{(k,\ell)}_{\psi_{N}}$ of $\psi_N$ is a trace class, positive operator on $\cH_{k,\ell}$ with kernel 
	\begin{equation}\label{eq:def_double_partial_trace-KERNEL-kl}
	\begin{split}
	\gamma^{(k,\ell)}_{\psi_{N}}(X,Y; X',Y')\;:=\;&\int_{\mathbb{R}^{3(N_1-k)}} \int_{\mathbb{R}^{3(N_2-\ell)}} \ud x_{k+1}\cdots\ud x_{N_1}\ud y_{\ell+1}\cdots\ud y_{N_2} \\
	& \qquad\times \psi_{N}(X,x_{k+1},\dots,x_{N_1};Y,y_{\ell+1},\dots,y_{N_2}) \\
	& \qquad\times \overline{\psi_{N}}(X',x_{k+1},\dots,x_{N_1};Y',y_{\ell+1},\dots,y_{N_2})
	\end{split}
	\end{equation}
	where $X,X'\in (\R^3)^k$ and $Y,Y'\in (\R^3)^\ell$. 
		
	The convergence \eqref{eq:g1convergesH-GP-intro}  expresses the asymptotic closeness as $N\to\infty$
	\[
	\psi_N\;\sim\; u_0^{\otimes N_1}\otimes v_0^{\otimes N_2}
	\]
	in a weak sense. Heuristically, this closeness emerges naturally in the mean-field regime as the Hartree functional is the energy per particle of the trial state $u^{\otimes N_1}\otimes v^{\otimes N_2}$. In the Gross-Pitaevskii regime, however, the ansazt $u_0^{\otimes N_1}\otimes v_0^{\otimes N_2}$ is not good enough to produce the Gross-Pitaevskii energy, and a non-trivial correction due to the short-range correlation between particles  is needed to recover the actual scattering lengths $a_\alpha$, instead of their first Born approximations  $(8\pi)^{-1}\int V^{(\alpha)}$. 

	In the mean-field regime, to see the macroscopic effect of the particle correlation we have to go to the next-to-leading order of the energy contribution. We will show that there exists a $N$-independent self-adjoint operator $\bH$ on a suitable Fock space such that
			\begin{equation}
			E_N^{\operatorname{MF}}\;=\;Ne_{\operatorname{H}}+\inf\sigma(\bH)+o(1)_{N\to+\infty}.
			\end{equation}
In fact, the operator $\bH$ is the Bogoliubov Hamiltonian (the second quantisation of the Hessian of the Hartree functional at the minimiser) that will be introduced in Section \ref{sect:setting}. Moreover, we will obtain an approximation for the ground state $\psi_N$ of $H_N^{\rm MF}$ in the norm topology of $\cH_{N_1,N_2}$, which is much more precise than the convergence of the reduced density matrices \eqref{eq:g1convergesH-GP-intro}.  
	
Such asymptotic results have several precursors in the literature on one-component condensates; see \cite{LieSeiYng-00,LieSei-02,LieSei-06,NamRouSei-16,BBCS-17} for the leading order in the Gross-Pitaevskii regime and \cite{Seiringer-11,GreSei-13,LewNamSerSol-15,DerNap-14,NamSei-15,Pizzo-15,RouSpe-16} for the second order in the mean-field regime. 


The novelty here is the adaptation, which in various respects is non-trivial, of previous analyses to the multi-component setting. This includes an amount of controls of the emerging `mixed' (i.e., inter-species) terms, among which the convexity argument for the GP minimiser (which singles out the role of the `miscibility' condition), the inter-species short-scale structure for the energy estimates (which is crucial for the ground state energy in the Gross-Pitaevskii regime), and the double-component bounds on the Bogoliubov Hamiltonian (which is the key ingredient for the second order contribution to the
ground state energy in the mean field regime).


%
%

\section{Main results} \label{sect:setting}

	Let us present the explicit set of assumptions for the potentials appeared in \eqref{eq:many_body_hamiltonian}.

	\begin{itemize}
		\item[($A_1$)] For $\alpha \in\{1,2\}$, the confining potentials satisfy $U^{(\alpha)}_{\operatorname{trap}}\in L^{3/2}_{\rm loc}(\R^3,\R)$ and 
		$$U^{(\alpha)}_{\operatorname{trap}}(x) \to +\infty\quad \text{as}\quad |x|\to +\infty.$$
		
		\item[($A_2^{\mathrm{GP}}$)] For $\alpha\in\{1,2,12\}$ the interaction potentials $V^{(\alpha)} \in C_c^\infty(\R^3)$ are nonnegative, spherically symmetric, and such that the respective scattering lengths satisfy the `miscibility' condition
		\begin{equation} \label{eq:inequality_a}
		a_1 a_2\;\geqslant\; a_{12}^2\,.
		\end{equation}
				
		\item[($A_2^{\mathrm{MF}}$)] For $\alpha\in\{1,2,12\}$ the interaction potentials $V^{(\alpha)}:\mathbb{R}^3\to\mathbb{R}$ are measurable, spherically symmetric, and satisfy the operator inequality
		\begin{equation}\label{eq:potentials_bdd_wrt_Laplacian}
		\big(V^{(\alpha)}\big)^{2}\;\leqslant\; C(\mathbbm{1}-\Delta),
		\end{equation}
		the (point-wise) positivity of the Fourier transform 
		\begin{equation} \label{eq:inequality_V}
		\widehat{V^{(1)}}\geqslant 0\,,\qquad \widehat{V^{(2)}}\geqslant 0
		\end{equation}	
		and the  `miscibility' condition 
		\begin{equation} \label{eq:inequality_V1}
		\widehat{V^{(1)}}\widehat{V^{(2)}}\;\geqslant\;\big(\widehat{V^{(12)}}\big)^2.
		\end{equation}	
		
	\end{itemize}
		
	The assumption $(A_1)$ ensures that  the one-body operator $T^{(\alpha)}\;:=\; -\Delta+U^{(\alpha)}_{\operatorname{trap}}$ is bounded from below, and with compact resolvent. It covers trivially the harmonic confinement $U_{\rm trap}^{(\alpha)}(x)=c_\alpha |x|^2$, which is often used in realistic experiments. Our results can be generalised with the Laplacian $-\Delta$ replaced by  the magnetic Laplacian $(i\nabla + A(x))^2$, and in the mean-field case we can also deal with the pseudo-relativistic Laplacian $\sqrt{-\Delta+m^2}-m$.  	
	
	Conditions \eqref{eq:potentials_bdd_wrt_Laplacian}-\eqref{eq:inequality_V} of  $(A_2^{\rm MF})$ are technical assumptions needed to validate Bogoliubov theory (they were already used in the one-component case in \cite{LewNamSerSol-15}). These conditions cover a wide range of interaction potentials, including Coulomb  interactions. 
	
	Condition \eqref{eq:inequality_a} in $(A_2^{\rm GP})$, as well as its mean field counterpart -- namely \eqref{eq:inequality_V1} in  $(A_2^{\rm MF})$, is needed to ensure that the Gross-Pitaevskii  or Hartree minimiser is unique. These conditions explicitly emerged in the physical literature and they were  recognised in the experimental observations as the `miscibility' condition between the two components of the mixtures (that is, the interspecies repulsion does not overcome the repulsion among particles of the same time and the two components are spatially mixed); see, e.g., 
	\cite{Bohn_PRL1997},
	\cite[Section 15.2]{Malomed2008_multicompBECtheory}, 
	\cite[Section 16.2.1]{Hall2008_multicompBEC_experiments}
	and \cite[Section 21.1]{pita-stringa-2016}. Here we will rigorously justify a uniqueness that is expected in theoretical physical arguments. Let us remark that our results below actually hold in a larger generality where \eqref{eq:inequality_a}  or \eqref{eq:inequality_V} is replaced by the condition that the Gross-Pitaevskii  or Hartree functional has a unique minimiser.
	
	Our first main result is the leading order behavior of the many-body ground state energy and the complete condensation of the ground state in the dilute regime.  
	
	\begin{theorem}[Leading order in the Gross-Pitaevskii limit]\label{thm:GP}~ 
	
	Let Assumptions ($A_1$) and ($A_2^{\mathrm{GP}}$) be satisfied. Then the following statements hold true. 
\begin{itemize}
 \item[(i)] There exists a unique minimiser $(u_0,v_0)$ (up to phases) for the variational problem
$$
e_{\operatorname{GP}} := \inf_{\substack{ u,v\,\in\, H^1(\mathbb{R}^3) \\ \|u\|_{L^2}=\|v\|_{L^2}=1}}	\cE^{\operatorname{GP}}[u,v].
$$

\item [(ii)] The ground state energy of $H_{N}^{\rm GP}$ satisfies
 \begin{equation} \label{eq:GP-energy-CV}
\lim_{N\to\infty}\frac{E^{\operatorname{GP}}_{N}}{N}=e_{\operatorname{GP}}.
\end{equation}	
\item[(iii)] If $\psi_{N}$ is an approximate ground state of $H_{N}^{\rm GP}$, in the sense that
\begin{equation*}
\lim_{N\to\infty}\frac{\langle\psi_N,H_N^{\operatorname{GP}}\psi_N\rangle}{N}\;=\;e_{\operatorname{GP}},
\end{equation*}
then it exhibits complete double-component Bose-Einstein condensation: 
\begin{equation} \label{eq:GP-1pdm-CV}
\lim_{N\to +\infty}\gamma_{\psi_N}^{(k,\ell)}\;=\;\ket{u_0^{\otimes k}\otimes v_0^{\otimes \ell}}\bra{u_0^{\otimes k}\otimes v_0^{\otimes \ell}}, \quad \forall k,\ell=0,1,2,...
\end{equation}
in trace class. 
\end{itemize}
\end{theorem}

In Theorem \ref{thm:GP}, we actually do not need the local integrability of potentials; for example, our analysis can be extended to hard sphere potentials. In fact, using Dyson's argument \cite{Dyson-57} we can always replace $V_N^{(\alpha)}$ by a solf potential (see Lemma \ref{lem:Dyson-1bd} below) and the detailed profile of the potential plays no role; only its scattering length matters. We also do not need to assume that the potential is short-range, because as soon as the potential decays fast enough, we can make a space cut-off and obtain the same result by a standard density argument. However, in the following, we will keep the conditions $(A_1)$, $(A_2^{\rm GP})$ to simplify the presentation. 

	The analogue of Theorem \ref{thm:GP} in the one-component case has been first proved by Lieb, Seiringer and Yngvason \cite{LieSeiYng-00} for the convergence of the ground state energy, and by Lieb and Seiringer \cite{LieSei-02,LieSei-06} for the condensation of the ground state. Later, based on quantum de Finetti methods \cite{ChrKonMitRen-07,AmmNie-08,LewNamRou-14} the ground state energy asymptotics and the condensation of the ground state were re-obtained by Nam, Rougerie and Seiringer \cite{NamRouSei-16}. Very recently, Boccato, Brennecke, Cenatiempo, and Schlein \cite{BBCS-17} obtained the optimal convergence rate in the homogeneous case (where the particles are confined in a unit torus, without external potential). In the present paper, we will follow the simplified approach in \cite{NamRouSei-16} when dealing with the multi-component case.

Theorem \ref{thm:GP} (as well as its mean-field counterpart in Theorem \ref{thm:leading-H}) justifies crucial information of the initial states {\em assumed} in various recent works on the dynamical problem for mixtures of condensates \cite{M-Olg-2016_2mixtureMF,AO-GPmixture-2016volume,Anap-Hott-Hundertmark-2017}. There, one proves that the mixture preserves its double-component condensation in the course of time evolution, if it is prepared at time $t=0$ in a state of condensation and, in the GP regime, provided that the energy per particle of the initial state is given by the GP energy functional. In the experiments the preparation of the compound condensate is precisely made by letting the system relax onto a suitable many-body ground state (or low-energy state), then the dynamical experiments starts by perturbing such an initial state, e.g., removing the confinement \cite{Hall2008_multicompBEC_experiments}. Now our result provides the rigorous ground for such initial conditions assumed in the dynamical analysis.

In the mean-field regime, we can go beyond the leading order by a detailed analysis of the fluctuations around the condensate. It is convenient to turn to the Fock space  formalism where the number of particles is not fixed. 
	
Let us introduce the single-component  Fock spaces
\begin{equation}
 \cF^{(j)}\;:=\;\bigoplus_{n=0}^\infty\big( \fh^{(j)}\big)^{\otimes_{\operatorname{sym}}n}\,,\qquad  \mathfrak{h}^{(j)}:=L^2(\mathbb{R}^3), \qquad j\in\{1,2\}
\end{equation}
and the double-component  Fock  space
\begin{equation} \label{eq:double_Fock}
  \cF\;:=\;\cF^{(1)}\otimes\cF^{(2)} \;=\;\bigoplus_{L=0}^\infty\;\Bigg(\bigoplus_{\substack{n,m\in\mathbb{N}_0 \\ n+m=L}}\big( \fh^{(1)}\big)^{\otimes_{\operatorname{sym}}n} \otimes \big(\fh^{(2)}\big)^{\otimes_{\operatorname{sym}}m}\Bigg).
\end{equation}

Let $(u_m)_{m=0}^\infty$ and $(v_n)_{n=0}^\infty$ be orthonormal bases of $\fh^{(1)}$ and $\fh^{(2)}$, respectively, with $(u_0,v_0)$ being the Hartree minimiser. We shall choose once and for all these two bases in such a way that all their elements belong to the domain of self-adjointness, respectively, of the operator $h^{(1)}$ on $\mathfrak{h}^{(1)}$ and of the operator $h^{(2)}$ on $\mathfrak{h}^{(2)}$ that we are going to define in formula \eqref{eq:h's} below. Let 
\begin{equation}\label{eq:def_modes}
a_m\;:=\;a(u_m),\qquad a^*_m\;:=\;a^*(u_m),\qquad b_n\;:=\;b(v_n),\qquad b^*_n\;:=\;b^*(v_n)\,.
\end{equation}
be the usual creation and annihilation operators on $\mathcal{F}^{(1)}$ and $\mathcal{F}^{(2)}$, which are linear operators defined by the actions 
\begin{equation} \label{eq:a's}
\begin{split}
(a_m \Psi_n)(x_1,\dots, x_{n-1})\;&=\;\sqrt{n}\int_{\R^3} \ud x \;\overline {f(x)}\,\Psi_n(x,x_1,\dots,x_{n-1}), \\
(a^*_m\Psi_n)(x_1,\dots, x_{n+1})\;&=\;\frac{1}{\sqrt{n+1}}\sum_{j=1}^{n+1}f(x_j)\Psi_n(x_1,\dots,x_{j-1},x_j,\dots,x_{n+1})
\end{split}
\end{equation}
for all $\Psi_n\in (\fh^{(1)})^{\otimes_{\operatorname{sym}}n}$ and for all $n\geqslant 0$, and similar actions for $b_m,b_m^*$. They satisfy the canonical commutation relations (CCR)
\begin{equation}\label{CCR}
\begin{split}
 [a_m,a_n] = 0 = [a^*_m,a^*_n], \qquad [a_m,a^*_n]= \delta_{m,n} \1_{ \mathcal{F}_+^{(1)}},\\
  [b_m,b_n] = 0 = [b^*_m,b^*_n]\,,\qquad [b_m,b^*_n]= \delta_{m,n} \1_{\mathcal{F}_+^{(2)}}.
\end{split}
\end{equation}
With no risk of confusion we shall keep denoting with $a_m$, $a^*_m$, $b_m$, $b^*_m$ the operators
\begin{equation} \label{eq:extended_operators}
a_m \otimes \mathbbm{1}_{\mathcal{F}^{(2)}},\qquad a^*_m \otimes \mathbbm{1}_{\mathcal{F}^{(2)}},\qquad\mathbbm{1}_{\mathcal{F}^{(1)}}\otimes b_m, \qquad \mathbbm{1}_{\mathcal{F}^{(1)}}\otimes b^*_m
\end{equation}
now acting on  $\cF=\cF^{(1)}\otimes\cF^{(2)}$. Obviously, $a_m,a_m^*$ commute with $b_n,b_n^*$.

In terms of these operators, we can extend the $N$-body Hamiltonian $H_{N}^{\operatorname{MF}}$ on $\mathcal{H}_{N_1,N_2}$ as an operator on the Fock space $\cF$ as
\begin{equation} \label{eq:many_body_Hamiltonian_second_quantized}
\begin{split}
H_{N}^{\operatorname{MF}}\;&:=\;\sum_{m,n\geqslant 0}\Big( \langle u_m, (-\Delta+U_{\rm trap}^{(1)}) u_n \rangle a^*_ma_n+\langle v_m, (-\Delta+U_{\rm trap}^{(2)}) v_n\rangle b^*_mb_n\Big) \\
&\quad +\frac{1}{N}\!\sum_{m,n,p,q}\Big({\textstyle\frac{1}{2}}V^{(1)}_{mnpq}a^*_ma^*_na_pa_q+{\textstyle\frac{1}{2}}V^{(2)}_{mnpq} b^*_mb^*_nb_pb_q + V^{(12)}_{mnpq}a^*_mb^*_na_pb_q  \Big)
\end{split}
\end{equation}
where 
\begin{align} \label{eq:def-Vmnpq}
V^{(1)}_{mnpq} &:=\;\langle u_m, [V^{(1)}*(\overline{u_n}u_q)]u_p\rangle, \nn\\
V^{(2)}_{mnpq} &:=\;\langle v_m, [V^{(2)}*(\overline{v_n}v_q)]v_p\rangle \\
V^{(12)}_{mnpq} &:=\;\langle u_m, [V^{(12)}*(\overline{v_n}v_q)]u_p\rangle \nn.
\end{align}

Bogoliubov's approximation \cite{Bogoliubov-47} suggests to formally replace $a_0,a_0^*$ and $b_0,b_0^*$ by the scalar values $\sqrt{N_1}$ and $\sqrt{N_2}$, respectively. It turns out that the terms of order $\sqrt{N}$ are canceled due to the Euler-Lagrange equations for the Hartree minimiser 
\begin{equation} \label{eq:mf-eq}
h^{(1)}u_0=0, \quad h^{(2)}v_0=0
\end{equation}
where 
\begin{equation}\label{eq:h's}
 \begin{array}{ll}
  h^{(1)}:=-\Delta+U_{\rm trap}^{(1)}+c_1V^{(1)}*|u_0|^2+c_2V^{(12)}*|v_0|^2-\mu^{(1)} ,\\
  h^{(2)}:=-\Delta+U_{\rm trap}^{(2)} +c_2V^{(2)}*|v_0|^2+c_1V^{(12)}*|u_0|^2-\mu^{(2)} 
 \end{array}
\end{equation}
with the chemical potentials 
\begin{equation} \label{eq:chemical_potentials}
\begin{split}
\mu^{(1)}\;&:=\; \langle u_0,(-\Delta+U_{\rm trap}^{(1)})u_0\rangle+c_1\langle u_0, V^{(1)}*|u_0|^2u_0\rangle+c_2 \langle u_0,V^{(12)}*|v_0|^2u_0\rangle,\\
\mu^{(2)}\;&:=\; \langle v_0,(-\Delta+U_{\rm trap}^{(2)})v_0\rangle+c_2 \langle v_0, V^{(2)}*|v_0|^2v_0\rangle+c_1 \langle u_0,V^{(12)}*|v_0|^2u_0\rangle\,.
\end{split}
\end{equation}
All this results in the quadratic Hamiltonian
\begin{equation} \label{eq:Bogoliubov_Hamiltonian}
\begin{split}
\bH := \sum_{m,n\geqslant 1}&\Big[ \langle u_m, h^{(1)} u_n\rangle a^*_m\,a_n+ \langle v_m, h^{(2)} v_n \rangle b^*_m\,b_n  + c_1V^{(1)}_{m00n}a^*_ma_n+ c_2V^{(2)}_{m00n}b^*_mb_n\\
&+{\textstyle\frac{1}{2}}\,c_2V^{(2)}_{mn00}b^*_mb^*_n+{\textstyle\frac{1}{2}}\,c_2\overline{V^{(2)}_{mn00}}b_mb_n+{\textstyle\frac{1}{2}}\,c_1V^{(1)}_{mn00}a^*_ma^*_n+{\textstyle\frac{1}{2}}\,c_1\overline{V^{(1)}_{mn00}}a_ma_n \\
&+\sqrt{c_1c_2}\,V^{(12)}_{0mn0}\,b^*_ma_n+\sqrt{c_1c_2}\,\overline{V^{(12)}_{0mn0}}\,a^*_mb_n \\
&+ \sqrt{c_1c_2}\,V^{(12)}_{mn00}\,a^*_mb^*_n +\sqrt{c_1c_2}\, \overline{V^{(12)}_{mn00}}\, a_mb_n\Big]\\
&-{\textstyle\frac{1}{2}}\,c_1V^{(1)}_{0000}-{\textstyle\frac{1}{2}}\,c_2V^{(2)}_{0000}
\end{split}
\end{equation}
which acts on the excited Fock space
\begin{equation}  \label{eq:def-cF+}
  \cF_+\;:=\;\cF^{(1)}_+\otimes\cF^{(2)}_+ \;=\;\bigoplus_{L=0}^\infty\;\Bigg(\bigoplus_{\substack{n,m\in\mathbb{N}_0 \\ n+m=L}}\big( \fh^{(1)}_+\big)^{\otimes_{\operatorname{sym}}n} \otimes \big(\fh^{(2)}_+\big)^{\otimes_{\operatorname{sym}}m}\Bigg)
\end{equation}
where 
	\begin{equation}
 \mathfrak{h}_+^{(1)}\;:= \{u_0 \}^{\bot} \subset L^2(\R^3)\,,\qquad \mathfrak{h}_+^{(2)}\;:= \{v_0 \}^{\bot} \subset L^2(\R^3).
\end{equation}
Note that $\bH$ is independent of the choice of $(u_m)_{m=1}^\infty$ and $(v_n)_{n=1}^\infty$, apart from the technical assumption that these functions belong to the domains $D(h^{(1)})$, $D(h^{(2)})$  of the self-adjoint extension of $h^{(1)}$, $h^{(2)}$, respectively. We can rigorously interpret $\bH$ as an operator with core domain
\begin{equation}  \label{eq:core-bH}
\bigcup_{M=0}^\infty \bigoplus_{L=0}^M \;\Bigg(\bigoplus_{\substack{n,m\in\mathbb{N}_0 \\ n+m=L}}\big( \fh^{(1)}_+ \cap D(h^{(1)})\big)^{\otimes_{\operatorname{sym}}n} \otimes \big(\fh^{(2)}_+ \cap D(h^{(2)})\big)^{\otimes_{\operatorname{sym}}m}\Bigg),
\end{equation}
which turns out to be bounded from below and can be extended to a self-adjoint operator by Friedrichs' method. 

In fact, the Bogoliubov Hamiltonian $\bH$ is nothing but the quantized form of (half) the Hessian of the Hartree functional at the minimiser. This $N$-independent Hamiltonian is expected to describe the fluctuations around the condensate. 

Our second main result is the next order correction to the ground state energy in the mean field regime and a norm convergence for the ground state.

\begin{theorem}[Bogoliubov correction to the mean-field limit]\label{thm:mf_correction}~

Let Assumptions ($A_1$) and ($A_2^{\mathrm{MF}}$) be satisfied.  Then the following statements hold true. 
\begin{itemize}
\item[(i)] There exists a unique minimiser $(u_0,v_0)$ (up to phases) for the variational problem
$$
e_{\operatorname{H}} := \inf_{\substack{ u,v\,\in\, H^1(\mathbb{R}^3) \\ \|u\|_2=\|v\|_2=1}}	\cE^{\operatorname{H}}[u,v].
$$
\item[(ii)] The Bogoliubov Hamiltonian $\bH$ in \eqref{eq:Bogoliubov_Hamiltonian} is bounded from below on $\cF_+$ with the core domain \eqref{eq:core-bH}. Moreover, its Friedrichs' self-adjoint extension,  still denoted by $\bH$, has a unique, non-degenerate ground state 
$\Phi^{\operatorname{gs}}=(\Phi^{\rm gs}_{m,n})_{m,n\geqslant 0} \in \cF_+.$

\item[(iii)] The ground state energy of $H_N^{\rm MF}$ satisfies
\begin{equation} \label{eq:mf_correction}
\lim_{N\to\infty}\big( E_{N}^{\operatorname{MF}}-Ne_{\operatorname{H}}\big)\;=\;\inf\sigma(\bH)\,.
\end{equation}	
\item[(iv)] The ground state $\psi_{N}^{\operatorname{gs}}$ of $H_{N}^{\operatorname{MF}}$ satisfies, up to a correct choice of phase, the norm approximation
\begin{equation}\label{eq:convergence_gs}
\lim_{N\to\infty} \left\| \psi_{N}^{\operatorname{gs}}- \sum_{0\leqslant m\leqslant N_1}\sum_{0\leqslant n\leqslant N_2} \frac{(a_0^*)^{N_1 - m}} {\sqrt {(N_1 - m)!}} \frac{(b_0^*)^{N_2 -n}} {\sqrt {(N_2 - n)!}} \Phi^{\operatorname{gs}}_{m,n} \right\|_{\cH_{N_1,N_2}} =0. 
\end{equation}
\end{itemize}
\end{theorem}

We remark that the ground state $\psi_N$ of $H_{N}^{\rm MF}$ is unique, up to complex phases. More precisely, in the case of no magnetic fields (as in Theorem \ref{thm:mf_correction}), the ground state of $H_{N}^{\rm MF}$ under the partial symmetric conditions (in the first and the second components) is the same with the absolute ground state (without any symmetry); see \cite[Section 3.2.4]{LieSei-10}. Thus the uniqueness and positivity of $\psi_N$ follow from the standard analysis of Schr\"odinger operators; see e.g. \cite[Chapter 11]{LieLos-01}. 

In the one-component case, Bogoliubov's second order correction for the ground state energy and the excitation spectrum in the mean-field regime has been first obtained by Seiringer \cite{Seiringer-11} for the homogeneous Bose gas (where particles are confined in a unit torus, without external potentials), and by Grech and Seiringer \cite{GreSei-13} for the non-homogeneous trapped gas. 
Then Lewin, Nam, Serfaty, Solovej \cite{LewNamSerSol-15} showed
that, under the assumption of Bose-Einstein condensation, one can get the next order
expansion in the energy by Bogoliubov theory in very general setting, and in particular
also for Coulomb-type interaction potentials.
Further extensions include a mixed mean-field large-volume limit by Derezi\'nski and Napiork\'owski \cite{DerNap-14}, collective excitations and multiple condensations by Nam and Seiringer \cite{NamSei-15}, 
a power expansion by Pizzo \cite{Pizzo-15} of the ground state in terms of a small coupling constant describing the interaction strength (which is $N^{-1}$ in our mean-field regime), 
and an infinitely-splitting double-well model by Rougerie and Spehner \cite{RouSpe-16}. In the present work, we will follow the general strategy in \cite{LewNamSerSol-15}  to justify Bogoliubov's approximation. Here we can also deal with the excitation spectrum of $H_N^{\rm MF}$, but we will focus only on the ground state to simplify the representation. 

In a very recent breakthrough \cite{BBCS-18}, Boccato, Brennecke, Cenatiempo, and Schlein were able to justify Bogoliubov's theory in the Gross-Pitaevskii limit for the homogeneous Bose gas. We expect that a similar result should hold for the multi-component case as well. 
 
 We will prove Theorem \ref{thm:GP} in Section \ref{sec:proof-GP} and prove Theorem \ref{thm:mf_correction} in Section \ref{sect:proof_correction}.

\section{Proof of Theorem \ref{thm:GP}} \label{sec:proof-GP}

\subsection{GP minimiser} Under the assumptions $(A_1)$ and $(A_2^{\mathrm{GP}})$, the existence of minimisers for the  Gross-Pitaevskii functional \eqref{eq:functionals-GP} follows easily from the standard direct method in the calculus of variations. Here we only focus on the uniqueness part. 

Let us define an auxiliary functional, with $f,g\geqslant 0$,
	\begin{equation}
	\cD^{\operatorname{GP}}[f,g]:=\cE^{\operatorname{GP}}[\sqrt{f},\sqrt{g}].
	\end{equation}
	The first step is to show convexity of $\cD^{\operatorname{GP}}$, namely
	\begin{equation} \label{eq:convexity}
	\cD^{\operatorname{GP}}\Big[\frac{f+r}{2},\frac{g+s}{2}\Big]\leqslant\frac{\cD^{\operatorname{GP}}[f,g]+\cD^{\operatorname{GP}}[r,s]}{2}.
	\end{equation}
	This is easily checked for the summands of $\cD^{\operatorname{GP}}$ that contain the kinetic operator (by \cite[Theorem 7.8]{LieLos-01}) and for those that contain the trapping potentials. For the terms containing the interaction potentials, let us consider, in self-explanatory notation, $\cD^{\operatorname{GP}}_{1}$, $\cD^{\operatorname{GP}}_{2}$, $\cD^{\operatorname{GP}}_{12}$ as the three summands of $\cD^{\operatorname{GP}}$ containing, respectively, $a_1$, $a_2$ and $a_{12}$. We have the identities 
	\begin{align} \label{eq:convex_GP1}
	\frac{\cD^{\operatorname{GP}}_1[f,g]+\cD^{\operatorname{GP}}_1[r,s]}{2}-\cD^{\operatorname{GP}}_1\Big[\frac{f+r}{2},\frac{g+s}{2} \Big]&=4\pi a_1 c_1^2\int \d k\Big|\frac{\widehat{f}(k)-\widehat{r}(k)}{2}\Big|^2,\\
	\frac{\cD^{\operatorname{GP}}_2[f,g]+\cD^{\operatorname{GP}}_2[r,s]}{2}-\cD^{\operatorname{GP}}_2\Big[\frac{f+r}{2},\frac{g+s}{2} \Big]&=4\pi a_2 c_2^2\int \d k\Big|\frac{\widehat{g}(k)-\widehat{s}(k)}{2}\Big|^2, \nn \\
	\frac{\cD^{\operatorname{GP}}_{12}[f,g]+\cD^{\operatorname{GP}}_{12}[r,s]}{2}-\cD^{\operatorname{GP}}_{12}\Big[\frac{f+r}{2},\frac{g+s}{2} \Big]&=8\pi a_{12}c_1c_2\int \d k\frac{\overline{\widehat{f}(k)-\widehat{r}(k)}}{2}\frac{\widehat{g}(k)-\widehat{s}(k)}{2}.\nn
	\end{align}
	By the Cauchy-Schwarz inequality together with \eqref{eq:inequality_a},
	\begin{align*}
	&4\pi a_{1}c_1^2\int \d k\Big|\frac{\widehat{f}(k)-\widehat{r}(k)}{2}\Big|^2 + 4\pi a_{2}c_2^2\int \d k\Big|\frac{\widehat{g}(k)-\widehat{s}(k)}{2}\Big|^2 \\
	& \geqslant 8\pi a_{12}c_1c_2\int \d k \Big| \frac{\widehat{f}(k)-\widehat{r}(k)}{2}\Big| \cdot \Big| \frac{\widehat{g}(k)-\widehat{s}(k)}{2} \Big|,
		\end{align*}
		and the convexity property \eqref{eq:convexity} follows. 
		
		Next, let us show that any Gross-Pitaevskii minimisers is positive, up to a complex phase. Indeed, let $(u_1,v_1)$ be a Gross-Pitaevskii minimiser. By the diamagnetic inequality \cite[Theorem 7.8]{LieLos-01}, $(|u_1|,|v_1|)$ is a Gross-Pitaevskii minimiser too, and we have 
		\begin{align}
		\label{eq:dia-equa}\int |\nabla u_1|^2 =\int |\nabla |u_1||^2, \quad \int |\nabla v_1|^2 =\int |\nabla |v_1||^2.
		\end{align} 
		Moreover, by a standard elliptic regularity from the Gross-Pitaevskii equation for $(|u_1|, |v_1|)$,  it follows that $|u_1|,|v_1|>0$ pointwise. Together with the equalities \eqref{eq:dia-equa}, we can conclude  from  \cite[Theorem 7.8]{LieLos-01} that $u_1=\theta_1 |u_1|$ and $v_1=\theta_1' |v_1|$ for complex constants $\theta_1, \theta_1'$. Thus up to complex phases, we can assume that $u_1,v_1>0$ pointwise. 
		
		Next, assume that $(u_1,v_1)$ and $(u_2,v_2)$ are two Gross-Pitaevskii minimisers, with $u_i,v_i> 0$ for $i=1,2$. Denote $f_i:=|u_i|^2$ and $g_i:=|v_i|^2$. Obviously,  $(f_1,g_1)$ and $(f_2,g_2)$ are minimisers for $\cD^{\operatorname{GP}}[f,g]$ with the constraint $\|f\|_{L^1}=\|g\|_{L^1}=1$. Combining with the convexity of $\cD^{\operatorname{GP}}$, we have  the following chain of inequalities
			\begin{equation*}
	0\geqslant \frac{\cD[f_1,g_1]+\cD[f_2,g_2]}{2}-\cD\Big[\frac{f_1+f_2}{2},\frac{g_1+g_2}{2} \Big]\geqslant 0.
	\end{equation*}
	This implies
	\begin{equation*}
	\frac{\cD[f_1,g_1]+\cD[f_2,g_2]}{2}=\cD\Big[\frac{f_1+f_2}{2},\frac{g_1+g_2}{2} \Big],
	\end{equation*}
	and in particular
	\begin{equation} \label{eq:uniqueness-key-equal}
	\frac{\langle \sqrt{f_1},-\Delta\sqrt{f_1}\rangle+\langle \sqrt{f_2},-\Delta\sqrt{f_2}\rangle}{2}=\Big\langle\sqrt{\frac{f_1+f_2}{2}},-\Delta\sqrt{\frac{f_1+f_2}{2}}\Big\rangle.
	\end{equation}
	By \cite[Theorem 7.8]{LieLos-01}, the equality \eqref{eq:uniqueness-key-equal} and the fact that $f_1,f_2>0$ imply that $f_1$ and $f_2$ are proportional. The normalization condition $\|f_1\|_{L^1}=\|f_2\|_{L^1}=1$ implies that $f_1=f_2$, and hence $u_1=u_2$. The same argument shows $v_1=v_2$.

\subsection{Energy upper bound} We will follow ideas in \cite{LieSeiYng-00}, with some modifications. First, let us recall the following result \cite[Appendix A.1]{ErdSchYau-04}. 

\begin{lemma} \label{lem:fNeumann} Let $0\leqslant V \in C_c^\infty(\R^3)$ be a radial function, with scattering length $a>0$. Then $N^2 V(N.)$ has scattering length $a/N$. Moreover, for every constant $\ell>0$, if $N$ is large enough so as to have $\supp V\subset \{|x|\leqslant N\ell\}$, then there exists a unique ground state $f \geqslant 0$ of the Neumann problem  
$$
-\Delta f +\frac{1}{2} N^2 V(N.) f= \lambda_N f
$$
on the ball $|x|\leqslant \ell$, with $f(x)=1$ on $|x|= \ell$. We can can extend $f$ to $\R^3$ by setting $f(x)=1$ if $|x|\geqslant \ell$, and hence
\begin{align} \label{eq:f-Neu-eq-R}
-\Delta f +\frac{1}{2}N^2 V(N.) f= \lambda_N f  \mathbf{1}_{B(0,\ell)}
\end{align}
on $\R^3$. Moreover,  
\begin{align} \label{eq:fNeu}
\lambda_N = \frac{3a}{N\ell^3}+  \frac{O(1)}{ N^{2}\ell^4}, \quad 0 \leqslant 1-f \leqslant \frac{C\mathbf{1}(|x|\leqslant \ell)}{N|x|}, \quad |\nabla f| \leqslant \frac{C\mathbf{1}(|x|\leqslant \ell)}{N|x|^2}.
\end{align}
\end{lemma}

\noindent
{\bf Trial function.} Let us introduce the notations 
\begin{align} \label{eq:notation-z}
(z_1,...,z_N) &:=(x_1,...,x_{N_1},y_1,...,y_{N_2}),\\
U_{i} &:=\begin{cases} U_{\rm trap}^{(1)} \quad &\text{if} \quad i \leqslant N_1,\\
U_{\rm trap}^{(2)}  \quad &\text{if}\quad i > N_1,\end{cases} \nn\\
(V_{ij},a_{i,j})&:=\begin{cases} (N^2V^{(1)}(N.),a_1) \quad &\text{if} \quad i,j \leqslant N_1,\\
(N^2V^{(2)}(N.),a_2) \quad &\text{if}\quad i,j > N_1,\\
(N^2V^{(12)}(N.),a_{12})\quad &\text{if}\quad i\leqslant N_1<j\quad \text{or} \quad j\leqslant N_1<i.\nn\end{cases}
\end{align}
Let $\ell>0$ be a small constant which is independent of $N$. For every $i\ne j$, let $(f_{ij},\lambda_{N,ij})$ be the pair $(f,\lambda_N)$ in Lemma \ref{lem:fNeumann} with $N^2 V(N.)$ replaced by $V_{ij}$. 

Thus we can write
\begin{align} \label{eq:HN-z}
H_N^{\rm GP}=\sum_{i=1}^N (-\Delta_{z_i} +U_i(z_i) )+\sum_{i<j}^N V_{ij}(z_i-z_j).
\end{align}

Take two functions $u,v \in C_c^\infty(\R^3)$, which are independent of $N$ and $\ell$, such that $\|u\|_{L^2}=\|v\|_{L^2}=1$. For every $i=1,2,...,N$ denote
\begin{align} \label{eq:trial-u-v}
u_{i}:=\begin{cases} u \quad &\text{if} \quad i \leqslant N_1,\\
v \quad &\text{if}\quad i> N_1.\end{cases}
\end{align}
Consider the trial function 
\begin{equation}\label{eq:trial-state-GP}
 \Psi(z_1,...,z_N) =\prod_{i=1}^{N} u_i (z_i) \prod_{j<k}^{N} f_{jk} (z_j-z_k).
\end{equation}
Note that we are using the full product $ \prod_{j<k}^{N} f_{jk} (z_j-z_k)$ to capture the short-range correlation, instead of using only a `nearest-neighbor induction' as in \cite{LieSeiYng-00} in the one-component case. We found that this strategy is more transparent and flexible, as it does not require bosonic symmetry between particles. 

The trial state $\Psi$ in \eqref{eq:trial-state-GP} is not normalized, but nevertheless we have
\begin{align}\label{eq:GP-upper-1}
E_N^{\rm GP} \leqslant \frac{\langle \Psi, H_N^{\rm GP} \Psi\rangle}{\langle \Psi, \Psi\rangle}.
\end{align}

\noindent
{\bf Norm estimates.} For every $i=1,2,...,N$, let us denote the $z_i$-independent function
$$\Psi_i:= \frac{\Psi}{u_i(z_i) \prod_{j\ne i}^N f_{ij}(z_i-z_j) }.$$
Since $0\leqslant f_{ij}\leqslant 1$, we have the pointwise estimate 
\begin{align} \label{eq:Psi-i-pw}
|\Psi| \leqslant |u_i(z_i)| |\Psi_i|.
\end{align}
On the other hand, using \eqref{eq:fNeu} we can estimate
$$
1- \prod_{j\ne i} f_{ij}^2(z_i-z_j) \leqslant \sum_{j\ne i} (1- f_{ij}^2(z_i-z_j)) \leqslant \sum_{j\ne i} \frac{C\mathbf{1}(z_i-z_j)}{N|z_i-z_j|}.
$$
Thus
\begin{align*}
0\leqslant |u_i(z_i)|^2|\Psi_i|^2- |\Psi|^2 &=  \Big( 1- \prod_{j\ne i} f_{ij}^2(z_i-z_j)  \Big) |u_i(z_i)|^2 |\Psi_i|^2 \\
&\leqslant \sum_{j\ne i} \frac{C\|u_i\|_{L^\infty}\mathbf{1}(|z_i-z_j|\leqslant \ell)}{N|z_i-z_j|} |\Psi_i|^2.
\end{align*}
Integrating this estimate leads to 
\begin{align} \label{eq:norm-Psi-i}
\|\Psi_i \|_{L^2}^2 \geqslant \|\Psi\|_{L^2}^2 \geqslant (1-C\ell^2)  \|\Psi_i \|_{L^2}^2. 
\end{align}
We will choose $\ell>0$ small such that $1-C\ell^2$ is close to $1$. 

Similarly, for every $i\ne j$, the $(z_i,z_j)$-independent function
$$\Psi_{ij}:= \frac{\Psi}{u_i(z_i) u_j(z_j) f_{ij}(z_i-z_j) \prod_{k\ne i,j}^N f_{ik}(z_i-z_k)f_{jk}(z_j-z_k) }$$
satisfies
\begin{align} \label{eq:norm-Psi-ij}
\|\Psi_{ij}\|_{L^2}^2 \geqslant \|\Psi\|_{L^2}^2 \geqslant (1-C\ell^2)  \|\Psi_{ij} \|_{L^2}^2. 
\end{align}

\noindent
{\bf Energy estimates.} In order to obtain an upper bound for the energy       convergence of Theorem 2.1, we show how to estimate all summands of the energy of the trial function (3.11) in terms of the GP functional, at the expense of negligible remainders.

First, we bound the one-body potential energy. For simplicity, let us assume $U_i \geqslant 0$ for all $i=1,2,...,N$ (this technical assumption will be removed at the end). Using \eqref{eq:Psi-i-pw} and \eqref{eq:norm-Psi-i} we can bound
\begin{align} \label{eq:GP-upper-1a}
\int U_i(z_i)|\Psi|^2 \leqslant \int U_i(z_i) |u_i(z_i)|^2 |\Psi_i|^2 &= \left( \int_{\R^3} U_i(z) |u_i(z)|^2 \d z \right) \|\Psi_i\|_{L^2}^2 \\
&\leqslant  \left( \int_{\R^3} U_i(z) |u_i(z)|^2 \d z \right) (1-C\ell^2)^{-1} \|\Psi\|_{L^2}^2. \nn
\end{align}
Here the identity follows from the fact that $\Psi_i$ is independent of $z_i$. 

Next, we consider the kinetic energy. For every $i=1,2,...,N$, we have 
\begin{align} \label{eq:GP-upper-kinetic-expanse}
\int |\nabla_{z_i} \Psi|^2 &= \int \Big| (\nabla_{z_i} u_i) \frac{\Psi}{u_i} + \sum_{j\ne i} (\nabla_{z_i} f_{ij}) \frac{\Psi}{f_{ij}}   \Big|^2 \nn\\
&=\int |\nabla_{z_i} u_i|^2 \frac{|\Psi|^2}{|u_i|^2} + \sum_{j\ne i} \int |\nabla_{z_i} f_{ij}|^2 \frac{|\Psi|^2}{|f_{ij}|^2} \\
&  + 2 \Re \sum_{j\ne i}  \int \overline{(\nabla_{z_i} u_i) \frac{\Psi}{u_i}}  (\nabla_{z_i} f_{ij}) \frac{\Psi}{f_{ij}} +  2 \Re \sum_{j\ne i \ne k \ne j}  \int \overline{(\nabla_{z_i} f_{ij}) \frac{\Psi}{f_{ij}} } (\nabla_{z_i} f_{ik}) \frac{\Psi}{f_{ik}}. \nn
\end{align} 
Let us show that the cross terms in \eqref{eq:GP-upper-kinetic-expanse} are small. In fact, for all $i\ne j$,  using \eqref{eq:fNeu}, \eqref{eq:Psi-i-pw} and \eqref{eq:norm-Psi-i}, we can estimate
\begin{align} \label{eq:GP-upper-1b}
&\Big| \int \overline{(\nabla_{z_i} u_i) \frac{\Psi}{u_i}}  (\nabla_{z_i} f_{ij}) \frac{\Psi}{f_{ij}} \Big| \leqslant   \int |\Psi_i|^2 |u_i\nabla_{z_i} u_i|  |\nabla_{z_i}f_{ij}| \nn\\
&\leqslant  \|u_i \nabla u_i\|_{L^\infty} \int |\Psi_i|^2  |\nabla_{z_i}f_{ij}(z_i-z_j)|  \nn\\
&= \|u_i \nabla u_i\|_{L^\infty}  \left( \int_{\R^3} |\nabla_z f_{ij}(z)| \d z  \right) \left( \int |\Psi_i|^2 \right) \nn\\
&\leqslant \frac{C\ell}{N} \|\Psi_i\|_{L^2}^2 \leqslant \frac{C\ell}{N (1-C\ell^2)} \|\Psi\|_{L^2}^2.
\end{align}
Here the identity follows by integrating w.r.t. $z_i$ fist (and using again that $\Psi_i$ is independent of $z_i$). The constant $C$ may depend on $u$ and $v$, but it is always independent of $N$ and $\ell$. Similarly, for all $i\ne j \ne k \ne i$, we have
\begin{align} \label{eq:GP-upper-1c}
&\Big|  \int \overline{(\nabla_{z_i} f_{ij}) \frac{\Psi}{f_{ij}} } (\nabla_{z_i} f_{ik}) \frac{\Psi}{f_{ik}} \Big| \leqslant   \int |\Psi_{ij}|^2 |u_i|^2 |u_j|^2 |\nabla_{z_i} f_{ij}|  |\nabla_{z_i}f_{ik}| \nn\\
&\leqslant \|u_i \|_{L^\infty}^2 \|u_j \|_{L^\infty}^2   \int |\Psi_{ij}|^2  |\nabla_{z_i} f_{ij}(z_i-z_j)|  |\nabla_{z_i}f_{ik}(z_i-z_k)|  \nn\\
&= \|u_i \|_{L^\infty}^2 \|u_j \|_{L^\infty}^2   \left( \int_{\R^3} |\nabla_z f_{ij}(z)| \d z  \right) ^2 \left( \int |\Psi_{ij}|^2 \right) \nn\\
&\leqslant \frac{C\ell^2}{N^2} \|\Psi_{ij}\|_{L^2}^2 \leqslant \frac{C\ell^2}{N^2(1-C\ell^2)} \|\Psi\|_{L^2}^2.
\end{align}
The identity follows by integrating w.r.t. $z_j$ fist, then integrating w.r.t. $z_i$, and using the fact that $\Psi_{ij}$ is independent of $(z_i,z_j)$.

Next we turn to the main terms in \eqref{eq:GP-upper-kinetic-expanse}. The first term can be estimated similarly to \eqref{eq:GP-upper-1a}:
\begin{align} \label{eq:GP-upper-1d}
\int |\nabla_{z_i} u_i|^2 \frac{|\Psi|^2}{|u_i|^2} \leqslant \int |\nabla_{z_i} u_i|^2 |\Psi_i|= \|\nabla u_i\|_{L^2}^2  \|\Psi_i\|_{L^2}^2 \leqslant \frac{1}{1-C\ell^2} \|\nabla u_i\|_{L^2}^2 \|\Psi\|_{L^2}^2. 
\end{align}
The second term in  \eqref{eq:GP-upper-kinetic-expanse} will be coupled with the interaction energy. We have
\begin{align*}
&\int |\nabla_{z_i} f_{ij}(z_i-z_j)|^2 \frac{|\Psi|^2}{|f_{ij}|^2} + \frac{1}{2} \int V_{ij}(z_i-z_j)|\Psi|^2 \\
&\leqslant \int \Big[ |\nabla_{z_i} f_{ij}(z_i-z_j)|^2 + \frac{1}{2}V_{ij}(z_i-z_j) |f_{ij}(z_i-z_j)|^2\Big] |u_i(z_i)|^2 |u_j(z_j)|^2  |\Psi_{ij}|^2 \\
&=\left( \int_{\R^3\times \R^3} \Phi_{ij} (x-y) |u_i(x)|^2 |u_j(y)|^2  \d x \d y \right) \|\Psi_{ij}\|_{L^2}^2\\
&\leqslant \left( \int_{\R^3\times \R^3} \Phi_{ij} (x-y) |u_i(x)|^2 |u_j(y)|^2  \d x \d y \right) \frac{1}{1-C\ell^2} \|\Psi\|_{L^2}^2 
\end{align*}
where
$$
\Phi_{ij}(z):= |\nabla_{z} f_{ij}(z)|^2 + \frac{1}{2}V_{ij}(z) |f_{ij}(z)|^2.  
$$
Since $\Phi_{ij}$ is supported on $|x|\leqslant \ell$, we can estimate
\begin{align*}
&\Big| \int_{\R^3\times \R^3} \Phi_{ij} (x-y) |u_i(x)|^2 |u_j(y)|^2  \d x \d y - \|\Phi_{ij}\|_{L^1} \int_{\R^3} |u_i(x)|^2 |u_j(x)|^2 \d x \Big|\\
&= \Big|  \int_{\R^3\times \R^3} \Phi_{ij} (x-y) |u_i(x)|^2 (|u_j(x)|^2- |u_j(y)|^2)  \d x \d y  \Big|\\
&\leqslant \sup_{|x-y|\leqslant \ell}  \Big| |u_j(x)|^2- |u_j(y)|^2 \Big| \int_{\R^3\times \R^3} \Phi_{ij} (x-y) |u_i(x)|^2   \d x \d y \\
&\leqslant \ell \|\nabla (|u_j|^2)\|_{L^\infty} \|\Phi_{ij}\|_{L^1}. 
\end{align*}
Moreover, using equation \eqref{eq:f-Neu-eq-R} for $f_{ij}$, the fact that $0\leqslant f_{ij}\leqslant 1$, and the estimate \eqref{eq:fNeu} for the eigenvalue $\lambda_N=\lambda_{N,ij}$ we find that
\begin{align*}
\|\Phi_{ij}\|_{L^1} &=\int_{\R^3} |\nabla_{z} f_{ij}(z)|^2 + \frac{1}{2}V_{ij}(z) |f_{ij}(z)|^2 \d z\\
 &= \lambda_{N,ij} \int_{\R^3} |f_{ij}(z)|^2 \mathbf{1}(|z|\leqslant \ell) \d z \\
 &\leqslant \Big(  \frac{3a_{i,j}}{N\ell^3}+  \frac{C}{N^2 \ell^4} \Big)  \int_{\R^3} \mathbf{1}(|z|\leqslant \ell) \d z \leqslant \frac{4\pi a}{N}+ \frac{C}{N^2 \ell}.  
\end{align*}
Thus
\begin{align*}
\int_{\R^3\times \R^3} \Phi_{ij} (x-y) |u_i(x)|^2 |u_j(y)|^2  \d x \d y &\leqslant (1+C\ell) \|\Phi_{ij}\|_{L^1} \int_{\R^3} |u_i(x)|^2 |u_j(x)|^2 \d x \\
& \leqslant (1+C\ell) \Big( \frac{4\pi a_{i,j}}{N}+ \frac{C}{N^2 \ell} \Big)\int_{\R^3} |u_i(x)|^2 |u_j(x)|^2 \d x
\end{align*}
and hence
\begin{align} \label{eq:GP-upper-1e}
&\int |\nabla_{z_i} f_{ij}(z_i-z_j)|^2 \frac{|\Psi|^2}{|f_{ij}|^2} + \frac{1}{2} \int V_{ij}(z_i-z_j)|\Psi|^2\nn\\
&\leqslant  \frac{1+C\ell}{1-C\ell^2} \Big( \frac{4\pi a_{i,j}}{N}+ \frac{C}{N^2 \ell} \Big) \Big(\int_{\R^3} |u_i(x)|^2 |u_j(x)|^2 \d x \Big) \|\Psi\|_{L^2}^2. 
\end{align}

\noindent
{\bf Conclusion of the upper bound.} In summary, putting \eqref{eq:GP-upper-1a}-\eqref{eq:GP-upper-1e}  together we obtain, for every $i=1,2,...,N$,
\begin{align}
&\left \langle \Psi, \Big( -\Delta_{z_i}+U_i(z_i) +\sum_{j\ne i}^N \frac{1}{2} V_{ij}(z_i-z_j)\Big) \Psi \right\rangle \|\Psi\|_{L^2}^{-2} \nn\\
&\leqslant \frac{1}{1-C\ell^2} \left( \|\nabla u_i\|_{L^2}^2+ \int_{\R^3} U_i(z) |u_i(z)|^2 \d z \right)  +  \frac{C\ell}{1-C\ell^2}+ \frac{C\ell^2}{1-C\ell^2}  \nn  \\
&+ \frac{1+C\ell}{1-C\ell^2} \sum_{j\ne i}^N \Big( \frac{4\pi a_{i,j}}{N}+ \frac{C}{N^2 \ell} \Big) \Big(\int_{\R^3} |u_i(x)|^2 |u_j(x)|^2 \d x \Big).
\end{align}
Summing over $i=1,2,...,N$  and using the choice \eqref{eq:trial-u-v}, we find that
\begin{align} \label{eq:GP-upper-final estimate}
\frac{E_N^{\rm GP}}{N}\leqslant \frac{\langle \Psi, H_N^{\rm GP} \Psi \rangle }{ N \|\Psi\|_{L^2}^2} &\leqslant  \frac{c_1}{1-C\ell^2} \left( \|\nabla u\|_{L^2}^2+ \int_{\R^3} U_{\rm trap}^{(1)}(z) |u(z)|^2 \d z \right) \nn  \\
& +  \frac{c_2}{1-C\ell^2} \left( \|\nabla v\|_{L^2}^2+ \int_{\R^3} U_{\rm trap}^{(2)}(z) |v(z)|^2 \d z \right)\nn\\
&+ c_1^2 \frac{1+C\ell}{1-C\ell^2} \Big(4\pi a_1+ \frac{C}{N \ell} \Big) \Big(\int_{\R^3} |u(x)|^4  \d x \Big)\nn\\
&+ c_2^2 \frac{1+C\ell}{1-C\ell^2} \Big(4\pi a_2+ \frac{C}{N \ell} \Big) \Big(\int_{\R^3} |v(x)|^4  \d x \Big)\nn\\
&+2c_1c_2 \frac{1+C\ell}{1-C\ell^2} \Big(4\pi a_{12}+ \frac{C}{N \ell} \Big) \Big(\int_{\R^3} |u(x)|^2|v(x)|^2  \d x \Big)\nn\\
& +  \frac{C\ell}{1-C\ell^2}+ \frac{C\ell^2}{1-C\ell^2} .
\end{align}
Taking $N\to +\infty$, and then taking $\ell\to 0$ in \eqref{eq:GP-upper-final estimate} lead to 
\begin{align} \label{eq:GP-upper-Euv}
\limsup_{N\to \infty}\frac{E_N^{\rm GP}}{N}\leqslant \cE^{\rm GP}[u,v].
\end{align}
So far, we have proved \eqref{eq:GP-upper-Euv} under the additional assumption that $U^{(\alpha)}_{\rm trap}\geqslant 0$, $\alpha\in \{1,2\}$. In general, if $U^{(\alpha)}_{\rm trap}$'s have negative parts, we can use \eqref{eq:GP-upper-1a} with $U_i$ replaced by $\max(U_i,-\eps^{-1})+\eps^{-1} \geqslant 0$, where $\eps>0$ is a small constant. This gives, instead of  \eqref{eq:GP-upper-Euv}, 
\begin{equation*}
	\limsup_{N\to\infty}\frac{E_{N,\varepsilon}^{\rm GP}}{N}\;\leqslant \;\mathcal{E}^{\rm GP}_\varepsilon [u,v].
\end{equation*}
where $E^{\mathrm{GP}}_{N,\varepsilon}$ and $\mathcal{E}^{\rm GP}_\varepsilon [u,v]$ are, respectively, the many-body ground state energy and the Gross-Pitaevskii functional with $U^{(\alpha)}_{\rm trap}$ replaced by $\max(U^{(\alpha)}_{\rm trap},-\varepsilon^{-1})$, $\alpha\in \{1,2\}$. We observe that a $\varepsilon^{-1}$ summand appears on both sides of the inequality, and hence exactly cancels. Since, by Lebesgue's monotone convergence theorem, $E^{\mathrm{GP}}_{N,\varepsilon}\to E^{\mathrm{GP}}_N$ and $\mathcal{E}^{\rm GP}_\varepsilon [u,v]\to \mathcal{E}^{\rm GP} [u,v]$ as $\varepsilon\to 0$, we conclude that \eqref{eq:GP-upper-Euv} holds true in general. 

Optimizing  \eqref{eq:GP-upper-Euv} over all $u,v\in C_c^\infty(\R^3)$ satisfying  $\|u\|_{L^2}=\|v\|_{L^2}=1$, we obtain  the desired upper bound
\begin{align} \label{eq:GP-upper-final}
\limsup_{N\to \infty}\frac{E_N^{\rm GP}}{N}\leqslant e_{\rm GP}.
\end{align}

\subsection{Dyson Lemma} Now we turn to the lower bound. We will follow the strategy in \cite{NamRouSei-16}. First, as in \cite{LieSeiYng-00,LieSeiSol-05,LieSei-06,NamRouSei-16}, we will replace the short-range potential $w_N$ by a longer-range potential with less singular scaling behavior. This idea goes back to Dyson \cite{Dyson-57}.

For every $R >0$ define
$$ \theta_{R}(x)= \theta\Big(\frac{x}{R}\Big), \quad U_R(x) = \frac{1}{R^3} U\Big(\frac{x}{R}\Big)$$
where $U \in C_c^\infty(\R^3)$ are radial functions satisfying
$$0\leqslant \theta\leqslant 1,\quad \theta(x)\equiv 0 \mbox{ for } |x|\leqslant 1, \quad \theta(x)\equiv 1 \mbox{ for } |x|\geqslant 2,$$
$$U\geqslant 0, \quad U(x) \equiv 0 \mbox{ for } |x| \notin [1/2,1], \quad \int_{\R ^3} U = 4\pi.$$
We will always denote by $p=-i\nabla$ the momentum variable.

The following result is taken from \cite{LieSeiSol-05}.

\begin{lemma}[Generalized Dyson lemma] \label{lem:Dyson-1bd} Let $v$ be a non-negative smooth function, supported on $|x|\leqslant R/2$ with scattering length $a$. Then for all $\eps,s>0$,
$$
p \theta_s(p) \mathbf{1}(|x|\leqslant R) \theta_s(p) p +\frac{1}{2}v(x) \geqslant (1-\eps) a U_R(x) -\frac{CaR^2}{\eps s^5}.
$$
\end{lemma}
\begin{proof} The bound follows from \cite[Lemma 4]{LieSeiSol-05} with $(U,\chi,s)$ replaced by $(U_R,\theta_s, s^{-1})$ and the first estimate in \cite[Eq. (52)]{LieSeiSol-05}.
\end{proof}

Next, we apply Lemma \ref{lem:Dyson-1bd} to derive a lower bound to the many-body Hamiltonian $H_N^{\rm GP}$. Under the notations \eqref{eq:notation-z}, we have

\begin{lemma}[Lower bound for many-body Hamiltonian] \label{lem:Dyson-Nbd} Let $\eps,s>0$ be independent of $N$ and let $N^{-1}\ll R \ll N^{-1/2}$. Then	\begin{align}
	H_{N}^{\operatorname{GP}} &\geqslant \sum_{i=1}^{N}\Big[-\Delta_{z_i}+ U_i(z_i)-(1-\eps)p_{z_i}^2\theta_s^2(p_{z_i}) \Big] \nn \\
	&\quad + \frac{(1-\eps)^2}{N} \sum_{j\ne i}^N a_{i,j} U_R(z_i-z_j) \prod_{k\ne i,j}^N \theta_{2R}(z_j-z_k)  +o(N).
	\end{align}
\end{lemma}

The purpose of Lemma \ref{lem:Dyson-Nbd} is to replace the short-range potentials $V^{(\alpha)}_N$ by the longer-range potential $U_R$, which essentially places us in the mean-field limit. This is done by using almost all of the high-momentum part $p^2 \theta_s(p)$ of the kinetic operator, and employing a many-body cut-off  $\prod_{k\ne i,j}^N \theta_{2R}(z_j-z_k)$ which rules out the event of having three particles close to each other.  This technical cut-off will be removed later.

\begin{proof} Note that $N^{-1}V^{(1)}_N=N^2 V^{(1)}(N.)$ is supported on $|x|\leqslant CN^{-1}$ and has scattering length $a_1N^{-1}$. Therefore, when  $N^{-1}\ll R \ll N^{-1/2}$ we can apply Lemma \ref{lem:Dyson-1bd} to obtain
\begin{align}\label{eq:Dyson-appy-1}
p_{z_i} \theta_s(p_{z_i}) \mathbf{1}(|z_i-z_j|\leqslant R) \theta_s(p_{z_i}) p_{z_i} +\frac{1}{2}V_{ij}(z_i-z_j) \geqslant (1-\eps) \frac{a_{i,j}}{N} U_R(z_i-z_j) +o(N^{-2}).
\end{align}

For every $i=1,2,...,N$, if every point in $\{z_j\}_{j\ne i}$ has a distance $\geqslant 2R$ to the others, then there is at most one of them has a distance $\leqslant R$ to $z_i$. In this case, 
$$
\sum_{j\ne i}^{N} \mathbf{1}(|z_i-z_j|\leqslant R)  \leqslant 1,
$$
and hence summing \eqref{eq:Dyson-appy-1} over $j$ leads to 
\begin{align*}
  p^2_{z_i} \theta_s^2 (p_{z_i})  + \sum_{j\ne i}^{N} \frac{1}{2}V_{ij}(z_i-z_j) \geqslant    (1-\eps) \sum_{j \ne i}^{N}  \frac{a_{i,j}}{N} U_R(z_i-z_j)  +o(N^{-2}).
\end{align*}
The latter estimate can be extended to all $\{z_j\}_{j\ne i}\subset \R^3$ as
\begin{align} \label{eq:DS-xi}
p^2_{z_i} \theta_s^2 (p_{z_i})  + \sum_{j\ne i}^{N} \frac{1}{2}V_{ij}(z_i-z_j)
\geqslant   \frac{(1-\eps)}{N} \sum_{j\ne i}^N a_{i,j} U_R(z_i-z_j) \prod_{k\ne i,j}^N \theta_{2R}(z_j-z_k)  +o(N^{-2}).
\end{align}
because the left side is always nonnegative. Multiplying both sides by $(1-\eps)$ leads to the desired estimate. 
\end{proof}

We use again the notation $(z_1,...,z_N):=(x_1,...,x_{N_1},y_1,...,y_{N_2})$ and introduce
\begin{align*}
h_i &: = \begin{cases} \widetilde{T}^{(1)}_{z_i} := -\Delta_{z_i} +U_{\rm trap}^{(1)}(z_i) - (1-\eps)p_{z_i}^2\theta_s^2(p_{z_i}) , & \mbox{if }i \leqslant N_1,\\
\widetilde{T}^{(2)}_{z_i}  := -\Delta_{z_i} +U_{\rm trap}^{(2)}(z_i)  - (1-\eps)p_{z_i}^2\theta_s^2(p_{z_i}) & \mbox{if } i>N_1, 
 \end{cases} \\
W_i &:= (1-\eps)^2 \sum_{j \ne i}^N a_{i,j} U_R(z_i-z_j) \prod_{k\ne i,j}^N \theta_{2R}(z_j-z_k), \\
\widetilde{H}_N &:= \sum_{i=1}^N \Big( h_i+\frac{1}{N}W_i\Big). 
\end{align*}
Lemma \ref{lem:Dyson-Nbd} can be rewritten as
\begin{align} \label{eq:GP-MF-1}
H_N^{\rm GP} \geqslant \widetilde{H}_N + o(N). 
\end{align}
Thus for the lower bound on $E_N^{\rm GP}$ it suffices to estimate the ground state energy of the modified Hamiltonian $\widetilde{H}_N$. 

By proceeding exactly as in \cite[Lemma 3.1, Lemma 3.4]{NamRouSei-16} (where the symmetries of $h_i$'s and $W_i$'s are not essential), we have the second moment estimate
\begin{equation}\label{eq:second-moment}
	(\widetilde{H}_{N})^2\geqslant\frac{1}{3}\Big( \sum_{i=1}^{N}h_i\Big)^2 - C_{\eps,s} N^2
	\end{equation}
and we can remove the cut-off $\prod_{k\ne i,j}^N \theta_{2R}(z_j-z_k)$: 
\begin{equation} \label{eq:3-body}
	\widetilde{H}_{N} \geqslant \sum_{i=1}^{N}h_i + \frac{(1-\eps)^2}{N} \sum_{j \ne i}^N a_{i,j} U_R(z_i-z_j) + o(N) (N^{-1}\widetilde{H}_{N})^4,
	\end{equation}	
provided that $\eps,s>0$ are independent of $N$ and $N^{-2/3}\ll R \ll N^{-1/2}$. 

Now let  $\widetilde{\Psi}_N$ be a ground state for $\widetilde{H}_N$ (which exists by a standard compactness argument). Taking the expectation of \eqref{eq:3-body} against  $\widetilde{\Psi}_N$, and using the equation $\widetilde{H}_N \widetilde{\Psi}_N = \widetilde{E}_N \widetilde{\Psi}_N$ with $\widetilde{E}_N=O(N)$, we find that
\begin{align} \label{eq:wEN}
\frac{E_N^{\rm GP}}{N} &\geqslant \frac{\widetilde E_N}{N} + o(1)_{N\to \infty} \nn\\
&\geqslant \left\langle \widetilde{\Psi}_N, \left( \frac{1}{N}\sum_{i=1}^{N}h_i + \frac{(1-\eps)^2}{N^2} \sum_{j \ne i}^N a_{i,j} U_R(z_i-z_j) \right) \widetilde{\Psi}_N \right\rangle + o(1)_{N\to \infty},
\end{align}
where the first inequality is due to \eqref{eq:GP-MF-1}. Thus, it remains to be bounded from below the right-side of \eqref{eq:wEN}.

\subsection{Energy lower bound} \label{sec:GP-lb} A further simplification on the right-hand side of \eqref{eq:wEN} is obtained by inserting a finite dimensional cut-off (similarly to what is done in \cite{NamRouSei-16}). We report here the argument.

We can find a constant $C_0>0$ (which might depend on $s,\eps$) such that
$$K:= \eps(-\Delta)+ \min(U^{(1)}_{\rm trap},U^{(2}_{\rm trap})+C_0\geqslant 1.$$
Then $K$ has compact resolvent because $\min(U^{(1)}_{\rm trap}(x),U^{(2}_{\rm trap}(x))\to +\infty$ as $|x|\to +\infty$. Therefore, for every fixed $L>0$, we have the finite-rank projection 
$$
P:= \1(K\leqslant L).
$$

Using the operator inequality (see, e.g. \cite[Lemma 3.2]{NamRouSei-16})
\begin{equation} \label{eq:oper-ineq}
U_R(z_i-z_j)\leqslant C\|U_R\|_{L^1}(1-\Delta_{z_i})^{1-\delta}  (1-\Delta_{z_j})^{1-\delta}, \quad \forall 1/4>\delta>0 
\end{equation}
and the fact that $1-\Delta$ is $K$-bounded, we have the simple Cauchy-Schwarz type inequality (c.f. \cite[Eq. before (4.10)]{NamRouSei-16})
$$
U_R(z_i-z_j) \geqslant P_{z_i} \otimes P_{z_j} U_R(z_i-z_j) P_{z_i} \otimes P_{z_j} - C_{\eps,s} L^{-1/10} K_{z_i}K_{z_j}, \quad \forall 1\leqslant i\ne j \leqslant N.
$$

From the second moment estimate \eqref{eq:second-moment}, we find that
\begin{equation}\label{eq:2nd-Psi}
\langle \widetilde{\Psi}_N, K_{z_i} K_{z_j} \widetilde{\Psi}_N\rangle \leqslant C_{\eps,s},\quad \forall 1\leqslant i \ne j   \leqslant N.
\end{equation}	
Thus \eqref{eq:wEN} reduces to 
\begin{align} \label{eq:wEN-b}
\frac{E_N^{\rm GP}}{N} &\geqslant \left\langle \widetilde{\Psi}_N, \left( \frac{1}{N}\sum_{i=1}^{N}h_i + \frac{(1-\eps)^2}{N^2} \sum_{j \ne i}^N a_{i,j} P_{z_i}\otimes P_{z_j} U_R(z_i-z_j) P_{z_i}\otimes P_{z_j} \right) \widetilde{\Psi}_N \right\rangle \nn\\
&\qquad + o(1)_{N\to \infty}+O(L^{-1/10})\nn\\
&= c_1 \Tr\Big[ \widetilde{T}^{(1)} \gamma_{\widetilde{\Psi}_{N}}^{(1,0)}\Big]  + c_2 \Tr\Big[ \widetilde{T}^{(2)} \gamma_{\widetilde{\Psi}_{N}}^{(0,1)}\Big]+o(1)_{N\to \infty}+O(L^{-1/10})\nn\\
&\qquad + (1-\eps)^2 c_1^2 a_1 \Tr\Big[ P_{x_1}\otimes P_{x_2}U_R(x_1-x_2) P_{x_1}\otimes P_{x_2} \gamma_{\widetilde{\Psi}_{N}}^{(2,0)}\Big]  \nn\\
&\qquad + (1-\eps)^2 c_2^2 a_2 \Tr\Big[ P_{y_1}\otimes P_{y_2}U_R(y_1-y_2) P_{y_1}\otimes P_{y_2} \gamma_{\widetilde{\Psi}_{N}}^{(0,2)}\Big] \nn\\
	&\qquad +  (1-\eps)^2 2c_1 c_2 a_{12} \Tr\Big[ P_{x}\otimes P_{y} U_R(x-y) P_{x}\otimes P_{y} \gamma_{\widetilde{\Psi}_{N}}^{(1,1)}\Big].
\end{align}
Here we simply use the definition of the reduced density matrices \eqref{eq:g1convergesH-GP-intro}.

	Our next tool is the  following abstract result.

\begin{theorem}[Quantum de Finetti theorem for 2-component Bose gas] \label{thm:deF} Let $\cK$ be a separable Hilbert space. Let $\Psi_{N}$ be a wave function in $\cK^{\otimes N_1}_{\rm sym} \otimes \cK^{\otimes N_2}_{\rm sym}$ and let $\gamma^{(k,\ell)}_{\Psi_{N}}$ be the reduced density matrices (defined similarly as in \eqref{eq:def_double_partial_trace-KERNEL-kl} with $L^2(\R^3)$ replaced by $\cK$). Then, up to a subsequence of $\{\Psi_{N}\}$, there exists a Borel probability measure $\mu$ supported on the set 
	$\{(u,v):u,v\in \cK, \|u\| \leqslant 1, \|v\|\leqslant 1\}$
	such that
	\begin{equation} \label{eq:qdeF-weak}
	 \gamma^{(k,\ell)}_{\Psi_{N}} \wto \int |u^{\otimes k} \otimes v^{\otimes \ell} \rangle \langle u^{\otimes k} \otimes v^{\otimes \ell}|  \d \mu (u,v),\quad \forall k,\ell=0,1,2,...
	\end{equation}
	weakly-* in trace class as $N\to \infty$. Moreover, if $\gamma^{(1,0)}_{\Psi_{N}}$ and $\gamma^{(0,1)}_{\Psi_{N}}$ converge strongly in trace class, then $\mu$ is supported on the set
	$\{(u,v):u,v\in \cK, \|u\|=\|v\|=1\}$ and the convergence in \eqref{eq:qdeF-weak} is strong in trace class for all $k,\ell\geqslant 0$.
\end{theorem}

This is the two-component analogue for the quantum de Finetti theorem, which was proved in~\cite{Stormer-69,HudMoo-75} for the case of strong convergence and in \cite{AmmNie-08,LewNamRou-14} for the case of weak convergence. Theorem \ref{thm:deF} can be proved by following the strategy in the one-component case in \cite{LewNamRou-14}.  We sketch a proof in Appendix \ref{sect:deF} for the reader's convenience. 

Now let us conclude the desired lower bound using Theorem \eqref{thm:deF}. Since 
$$
h_i \geqslant K_{z_i}-2C_{\eps,s}, \quad \forall i=1,...,N,
$$
we deduce from \eqref{eq:2nd-Psi} that $\Tr[ K\gamma_{\widetilde{\Psi}_{N}}^{(1,0)}]$ and $\Tr[ K \gamma_{\widetilde{\Psi}_{N}}^{(0,1)}]$ are bounded uniformly in $N$. Since $K$ has compact resolvent, up to a subsequence as $N\to \infty$, we obtain that $\gamma_{\widetilde{\Psi}_{N}}^{(1,0)}$ and $\gamma_{\widetilde{\Psi}_{N}}^{(0,1)}$ converge strongly in trace class. Thus up to a subsequence again, Theorem \ref{thm:deF} ensures the existence of a Borel probability measure $\nu$ supported on the set 
	$$\{(u,v):u,v\in L^2(\R^3), \|u\|_{L^2}=\|v\|_{L^2}=1\}$$
	such that
\begin{align} \label{eq:cv-quan-deF-H_1}
	\lim_{N\to \infty} \gamma^{(k,\ell )}_{\widetilde{\Psi}_{N}} = \int |u^{\otimes k} \otimes v^{\otimes \ell} \rangle \langle u^{\otimes k} \otimes v^{\otimes \ell}|  \d \nu (u,v),\quad \forall k,\ell=0,1,2,...
	\end{align}
	
	Now we take the limit $N\to \infty$, and then $L\to \infty$ later on the right side of \eqref{eq:wEN-b}. Since $\widetilde{T}^{(1)},\widetilde{T}^{(2)}$ are bounded from below, we can use \eqref{eq:cv-quan-deF-H_1} and Fatou's lemma to get
	\begin{align} \label{eq:Ft1}
	\liminf_{N\to \infty} \Tr\Big[ \widetilde{T}^{(1)} \gamma_{\widetilde{\Psi}_{N}}^{(1,0)}\Big] &\geqslant \int \langle u, \widetilde{T}^{(1)} u\rangle \d \nu(u,v),\\
	\liminf_{N\to \infty} \Tr\Big[ \widetilde{T}^{(2)} \gamma_{\widetilde{\Psi}_{N}}^{(0,1)}\Big] &\geqslant \int \langle v, \widetilde{T}^{(2)} v\rangle \d \nu(u,v).\label{eq:Ft2}
	\end{align}	
The operator inequality \eqref{eq:oper-ineq} and the fact that $(1-\Delta)$ is K-bounded ensure that $P_{x} \otimes P_{y} U_R(x-y) P_{x}\otimes P_{y}$  is uniformly bounded in $N$ as an operator. Therefore, the trace class convergence \eqref{eq:cv-quan-deF-H_1} implies that
\begin{align*}
&\Tr\Big[ P_{x}\otimes P_{y} U_R(x-y) P_{x}\otimes P_{y} \gamma_{\widetilde{\Psi}_{N}}^{(1,1)}\Big] \\
&= \int \langle u\otimes v, P_{x}\otimes P_{y} U_R(x-y) P_{x}\otimes P_{y} u\otimes v \rangle \d \nu(u,v)+o(1)_{N\to \infty}. 
\end{align*}
From the choice of $U_R$, we get
\begin{align*}
&\lim_{N\to \infty}\langle u\otimes v,  P_x\otimes P_y U_R(x-y)P_x\otimes P_y u\otimes v\rangle \\
&= \lim_{N\to \infty} \langle  |Pu|^2, U_R*|Pv|^2 \rangle = 4\pi \int_{\R^3} |Pu(x)|^2 |Pv(x)|^2 \d x.
\end{align*}
Next we take the limit $L\to \infty$ to remove the cut-off $P=\1(K\leqslant L)$. Since $\Tr[ K\gamma_{\widetilde{\Psi}_{N}}^{(1,0)}]$ and $\Tr[ K \gamma_{\widetilde{\Psi}_{N}}^{(0,1)}]$ are bounded, $\nu$ is supported on $Q(K)\times Q(K)$ where $Q(K)$ is the quadratic form domain of $K$. Consequently, for all $(u,v)$ in the support of $\nu$, we have $Pu\to u$ and $Pv\to v$ strongly in $Q(K)$ as $L\to \infty$. Moreover, since $(1-\Delta)$ is $K$-bounded, we have the continuous embeddings $Q(K)\subset H^1(\R^3)\subset L^4(\R^3)$. Therefore, 
\begin{align*}
\lim_{L\to \infty} \int_{\R^3} |Pu(x)|^2 |Pv(x)|^2 \d x=  \int_{\R^3} |u(x)|^2 |v(x)|^2 \d x,
\end{align*}
and hence
\begin{align*}
\lim_{L\to \infty} \lim_{N\to \infty}\langle u\otimes v,  P_x\otimes P_y U_R(x-y)P_x\otimes P_y u\otimes v\rangle  = 4\pi \int_{\R^3} |u(x)|^2 |v(x)|^2 \d x.
\end{align*}
Thus by Fatou's lemma, we find that 
\begin{align} \label{eq:Ft3}
&\liminf_{L\to \infty}\liminf_{N\to \infty}\Tr\Big[ P_{x}\otimes P_{y} U_R(x-y) P_{x}\otimes P_{y} \gamma_{\widetilde{\Psi}_{N}}^{(1,1)}\Big] \nn\\
&=\liminf_{L\to \infty}\liminf_{N\to \infty} \int \langle u\otimes v, P_{x}\otimes P_{y} U_R(x-y) P_{x}\otimes P_{y} u\otimes v \rangle \d \nu(u,v)\nn\\
&\geqslant \int \Big[4\pi \int_{\R^3} |u(x)|^2 |v(x)|^2 \d x\Big] \d \nu(u,v). 
\end{align}
Similarly, we also have
\begin{align} \label{eq:Ft4}
&\liminf_{L\to \infty} \liminf_{N\to \infty}\Tr\Big[ P_{x_1}\otimes P_{x_2} U_R(x_1-x_2) P_{x_1}\otimes P_{x_2} \gamma_{\widetilde{\Psi}_{N}}^{(2,0)}\Big]  \geqslant \int \Big[4\pi \int_{\R^3} |u(x)|^4  \d x\Big] \d \nu(u,v),\\
&\liminf_{L\to \infty} \liminf_{N\to \infty}\Tr\Big[ P_{y_1}\otimes P_{y_2} U_R(y_1-y_2) P_{y_1}\otimes P_{y_2} \gamma_{\widetilde{\Psi}_{N}}^{(0,2)}\Big] \geqslant \int \Big[4\pi \int_{\R^3} |v(y)|^4  \d x\Big] \d \nu(u,v). \label{eq:Ft5}
\end{align}
Inserting \eqref{eq:Ft1}-\eqref{eq:Ft5} into the right side of \eqref{eq:wEN-b}, we arrive at
\begin{align} \label{eq:wEN-c}
\liminf_{N\to \infty} \frac{E_N^{\rm GP}}{N}  \geqslant \int \widetilde{\cE}^{\rm GP}_{\eps,s}[u,v] \d \nu (u,v) \geqslant \inf_{\|u\|_{L^2}=\|v\|_{L^2}=1} \widetilde{\cE}^{\rm GP}_{\eps,s}[u,v] 
\end{align}
where
\begin{align}
 \widetilde{\cE}^{\rm GP}_{\eps,s}[u,v]  &:= c_1 \langle u, \widetilde{T}^{(1)} u\rangle + c_2 \langle v, \widetilde{T}^{(2)} v\rangle
  +(1-\eps)^2 4\pi a_1 c_1^2 \int_{\R^3} |u(x)|^4 \d x \\
  &\qquad + (1-\eps)^2 4\pi a_2 c_2^2 \int_{\R^3} |v(x)|^4 \d x + (1-\eps)^2 8 \pi a_{12}c_1c_2 \int_{\R^3} |u(x)|^2 |v(x)|^2 \d x. \nn
 \end{align}
 
Finally,  we take $s\to 0$, and then $\eps\to 0$. By a standard compactness argument as in \cite[after (103)]{LieSei-06}, we have
\begin{align} \label{eq:weGP-eGP}
\lim_{\eps\to 0} \lim_{s\to 0}\inf_{\|u\|_{L^2}=\|v\|_{L^2}=1} \widetilde{\cE}^{\rm GP}_{\eps,s}[u,v]  =e_{\rm GP}.
\end{align}
Thus \eqref{eq:wEN-c} leads to the desired lower bound
\begin{equation} \label{eq:lower-GP-conc}
\liminf_{N\to \infty}\frac{{E}_N^{\rm GP}}{N} \geqslant e_{\rm GP}.
\end{equation}
Strictly speaking, we have so far proved \eqref{eq:lower-GP-conc} for a subsequence as $N\to \infty$. However, since the limit $e_{\rm GP}$ is independent of the subsequence, we can obtain the estimate for the whole sequence by a standard contradiction argument. 

Combining with the energy upper bound \eqref{eq:GP-upper-final}, we conclude the proof of \eqref{eq:GP-energy-CV}:
$$
\lim_{N\to \infty} \frac{E_N^{\rm GP}}{N} = e_{\rm GP}.
$$

\subsection{Convergence of density matrices} Let $\psi_N$ be an approximate ground state for $H_N^{\rm GP}$. Since $\Tr[ (-\Delta +U_{\rm trap}^{(1)})\gamma_{{\psi}_{N}}^{(1,0)}]$ and $\Tr[ (-\Delta +U_{\rm trap}^{(2)})\gamma_{{\psi}_{N}}^{(0,1)}]$ are bounded uniformly in $N$, up to a subsequence as $N\to \infty$, $\gamma_{{\psi}_{N}}^{(1,0)}$ and $\gamma_{{\psi}_{N}}^{(0,1)}$ converge strongly in trace class. 
Thus, by Theorem \ref{thm:deF}, up to a subsequence again, there exists a Borel probability measure $\mu$ supported on the set 
	$$\{(u,v):u,v\in L^2(\R^3), \|u\|_{L^2}=\|v\|_{L^2}=1\}$$
	such that
\begin{align} \label{eq:deF-psiN-app}
	\lim_{N\to \infty} \gamma^{(k,\ell )}_{{\psi}_{N}} = \int |u^{\otimes k} \otimes v^{\otimes \ell} \rangle \langle u^{\otimes k} \otimes v^{\otimes \ell}|  \d \mu (u,v),\quad \forall k,\ell=0,1,2,...
\end{align}
strongly in trace class. 

Let us show that $\mu$ is supported on the set $\{(e^{i\theta_1} u_0, e^{\theta_2}v_0): \theta_1,\theta_2\in \R \}$, where $(u_0,v_0)$ is the unique Gross-Pitaevskii minimiser. This follows from the convergence of ground state energy and a standard Hellmann-Feynman type argument as in \cite{LieSei-06}.  To be precise, let us denote  
$$Q:= |u_0\otimes v_0\rangle \langle u_0\otimes v_0|$$ 
and for every fixed $\eta>0$, consider the perturbed Hamiltonian
$$
H^{\rm GP}_{N,\eta}:=H^{\rm GP}_{N}+\frac{\eta}{Nc_1c_2}\sum_{i=1}^{N_1} \sum_{j=1}^{N_2} Q_{x_i,y_j},
$$
where $Q_{x_i,y_j}$ indicates the projector $Q$ acting on the $i$-th variable of the first sector of $\mathcal{H}_{N_1,N_2,\mathrm{sym}}$ and on the $j$-th variable of the second sector.

Then by the same method as above, we obtain the analogue of \eqref{eq:GP-energy-CV} (a lower bound is sufficient for our purpose)
\begin{equation}\label{eq:H-GP-eta-lb}
\liminf_{N\to \infty} \frac{\inf \sigma(H^{\rm GP}_{N,\eta})}{N} \geqslant \inf_{\|u\|_{L^2}=\|v\|_{L^2}=1} \Big\{\cE^{\rm GP}[u,v]+ \eta |\langle u, u_0\rangle|^2 |\langle v,v_0\rangle|^2 \Big\}=: e_{\rm GP,\eta}.
\end{equation}
Next, we can write 
$$
\Tr [Q \gamma_{\psi_N}^{(1,1)}] = \frac{1}{N^2c_1c_2} \Big\langle \psi_N, \sum_{i=1}^{N_1} \sum_{j=1}^{N_2} Q_{x_i,y_j} \psi_N\Big\rangle = \frac{1}{N\eta} \Big[ \langle \psi_N, H_{N,\eta}^{\rm GP} \psi_N \rangle -  \langle \psi_N, H_N^{\rm GP} \psi_N \rangle \Big].
$$
Using the lower bound \eqref{eq:H-GP-eta-lb} and the assumption that $\psi_N$ is an approximate ground state for $H_N^{\rm GP}$, we find that
$$
\liminf_{N\to \infty}\Tr [Q \gamma_{\psi_N}^{(1,1)}] \geqslant \frac{1}{\eta} \Big[  e_{\rm GP,\eta}- e_{\rm GP}\Big] .
$$

Next, when $\eta\to 0$, the minimiser $(u_\eta, v_\eta)$ of $e_{\rm GP,\eta}$ becomes a minimising sequence for $e_{\rm GP}$, and it converges to the unique minimiser $(u_0,v_0)$ of $e_{\rm GP}$ by a standard compactness argument. Therefore, 
$$\liminf_{\eta\to 0}\frac{1}{\eta} \Big[  e_{\rm GP,\eta}- e_{\rm GP}\Big] \geqslant \liminf_{\eta\to 0} |\langle u_\eta, u_0\rangle|^2 |\langle v_\eta,v_0\rangle|^2 =1.$$
Thus in conclusion, we have with $Q=|u_0\otimes v_0\rangle \langle u_0\otimes v_0|$, 
$$
\liminf_{N\to \infty}\Tr [Q \gamma_{\psi_N}^{(1,1)}] \geqslant 1.
$$
Consequently, the convergence \eqref{eq:deF-psiN-app} implies that
$$
\int |\langle u, u_0\rangle|^2 |\langle v,v_0\rangle|^2 \d \mu(u,v)\geqslant 1.
$$
Thus $\mu$ is supported on the set $\{(e^{i\theta_1} u_0, e^{\theta_2}v_0): \theta_1,\theta_2\in \R \}$, and hence \eqref{eq:deF-psiN-app}  reduces to the desired convergence \eqref{eq:GP-1pdm-CV}:
$$
\lim_{N\to \infty} \gamma_{\psi_N}^{(k,\ell)} = |u_0^{\otimes k}\otimes v_0^{\otimes \ell}\rangle \langle u_0^{\otimes k}\otimes v_0^{\otimes \ell} |, \quad \forall k,\ell=0,1,2,...  
$$
in trace class. Again,  we have so far proved \eqref{eq:GP-1pdm-CV} for a subsequence as $N\to \infty$, but since the limit is unique, the convergence actually holds for the whole sequence.

\section{Proof of Theorem \ref{thm:mf_correction}} \label{sect:proof_correction}

\subsection{Leading order and Hartree theory}

In the mean-field regime, we have the following analogue of Theorem \eqref{thm:GP}.

\begin{theorem}[Leading order in the mean-field limit]\label{thm:leading-H}~ 
	
	Let Assumptions ($A_1$) and ($A_2^{\mathrm{MF}}$) be satisfied. 
\begin{itemize}
 \item[(i)] There exists a unique minimiser $(u_0,v_0)$ (up to phases) for the variational problem
$$
e_{\operatorname{H}} := \inf_{\substack{ u,v\,\in\, H^1(\mathbb{R}^3) \\ \|u\|_{L^2}=\|v\|_{L^2}=1}}	\cE^{\operatorname{H}}[u,v].
$$

\item [(ii)] The ground state energy of $H_{N}^{\rm MF}$ satisfies
 \begin{equation} \label{eq:energy-leading-H}
\lim_{N\to\infty}\frac{E^{\operatorname{MF}}_{N}}{N}=e_{\operatorname{H}}.
\end{equation}	
\item[(iii)] If $\psi_{N}$ is an approximate ground state of $H_{N}^{\rm MF}$, in the sense that
\begin{equation*}
\lim_{N\to\infty}\frac{\langle\psi_N,H_N^{\operatorname{MF}}\psi_N\rangle}{N}\;=\;e_{\operatorname{H}},
\end{equation*}
then it exhibits complete double-component Bose-Einstein condensation: 
\begin{equation} \label{eq:H-1pdm-CV}
\lim_{N\to +\infty}\gamma_{\psi_N}^{(k,\ell)}\;=\;\ket{u_0^{\otimes k}\otimes v_0^{\otimes \ell}}\bra{u_0^{\otimes k}\otimes v_0^{\otimes \ell}}, \quad \forall k,\ell=0,1,2,...
\end{equation}
in trace class. 
\end{itemize}
\end{theorem}

\begin{proof} The proof of Theorem \ref{thm:leading-H} is similar to (indeed easier than) the proof of Theorem \ref{thm:GP}. Let us quickly explain the necessary adaptation. 

(i) The existence of minimisers of $e_{\rm H}$ is standard. The uniqueness of miminizer (up to complex phases) follows a convexity argument as in the Gross-Pitaevskii regime. More precisely, if we denote $\cD^{\operatorname{H}}[f,g]=\cE^{\rm H}[\sqrt{f}, \sqrt{g}]$ for $f,g\geqslant 0$, then $\cD^{\operatorname{H}}$ is convex. Indeed,  by considering $\cD^{\operatorname{H}}_1$, $\cD^{\operatorname{H}}_2$ and $\cD^{\operatorname{H}}_{12}$ as the summands of the Hartree functional containing respectively $V^{(1)}$, $V^{(2)}$ and $V^{(12)}$, we obtain the following analogue
	\begin{align*} 
	\frac{\cD^{\operatorname{H}}_1[f,g]+\cD^{\operatorname{H}}_1[r,s]}{2}-\cD^{\operatorname{H}}_1\Big[\frac{f+r}{2},\frac{g+s}{2} \Big]&=\frac{c_1^2}{2}\int \d k\Big|\frac{\widehat{f}(k)-\widehat{r}(k)}{2}\Big|^2\widehat{V^{(1)}}(k),\\
	\frac{\cD^{\operatorname{H}}_2[f,g]+\cD^{\operatorname{H}}_2[r,s]}{2}-\cD^{\operatorname{H}}_2\Big[\frac{f+r}{2},\frac{g+s}{2} \Big]&=\frac{c_2^2}{2}\int \d k\Big|\frac{\widehat{g}(k)-\widehat{s}(k)}{2}\Big|^2\widehat{V^{(2)}}(k), \\
	\frac{\cD^{\operatorname{H}}_{12}[f,g]+\cD^{\operatorname{H}}_{12}[r,s]}{2}-\cD^{\operatorname{H}}_{12}\Big[\frac{f+r}{2},\frac{g+s}{2} \Big]&=c_1c_2\int \d k\frac{\overline{\widehat{f}(k)-\widehat{r}(k)}}{2}\frac{\widehat{g}(k)-\widehat{s}(k)}{2}\widehat{V^{(12)}}(k).
	\end{align*}
	Therefore, the convexity follows from the Cauchy-Schwarz inequality and Assumptions \eqref{eq:inequality_V}-\eqref{eq:inequality_V1}. The rest is exactly similar to the Gross-Pitaevskii case. 
	
(ii)-(iii) The convergence of energy and approximate ground states can be obtained by following the strategy in the one component case \cite{LewNamRou-14}. In fact, the energy upper bound 
	$$
	\limsup_{N\to \infty} \frac{E_N^{\rm MF}}{N} \leqslant e_{\rm H}
	$$
	follows immediately from choosing the trial state $u^{\otimes N_1} \otimes v^{\otimes N_2}$. Now let $\Psi_{N}$ be a wave function in $\cH_{N_1,N_2}$ such that
	\begin{align} \label{eq:appxo-H} 
	\langle \Psi_{N}, H_{N}^{\rm MF} \Psi_{N}\rangle \leqslant E_N^{\rm MF} + o(N).
	\end{align}
	We can write, in terms of reduced density matrices, 
		\begin{align}
	e_{\rm H}+o(1)_{N\to \infty} \geqslant \frac{\langle \Psi_{N}, H_{N}^{\rm MF} \Psi_{N}\rangle}{N} &=  c_1 \Tr\Big[ T^{(1)} \gamma_{\Psi_{N}}^{(1,0)}\Big]  + \frac{c_1^2}{2} \Tr\Big[ V^{(1)}(x_1-x_2) \gamma_{\Psi_{N}}^{(2,0)}\Big] \nn \\
	&\qquad + c_2 \Tr\Big[ T^{(2)} \gamma_{\Psi_{N}}^{(0,1)}\Big]  + \frac{c_2^2}{2} \Tr\Big[ V^{(2)}(y_1-y_2) \gamma_{\Psi_{N}}^{(0,2)}\Big] \nn \\
	&\qquad  + c_1 c_2 \Tr\Big[ V^{(12)}(x-y) \gamma_{\Psi_{N}}^{(1,1)}\Big]\label{eq:EN-DM} 
	\end{align}
	where $T^{(\alpha)}\;:=\; -\Delta+U^{(\alpha)}_{\operatorname{trap}}$. Using the assumptions $(A_1)$-$(A_2^{\rm MF})$, we obtain the operator inequalities 
	\begin{align} 
	\pm V^{(1)}(x_1-x_2) &\leqslant \eps T^{(1)}_{x_1} + C_\eps,\label{eq:low-opr-1}\\
	\pm V^{(2)}(y_1-y_2) &\leqslant  \eps T^{(2)}_{y_1} +C_\eps ,\label{eq:low-opr-2}\\
	\pm V^{(12)}(x-y) &\leqslant \eps( T^{(1)}_x + T^{(2)}_y) +C_\eps  \label{eq:low-opr-3}
	\end{align}
	for all $\eps>0$. 
	
	Consequently,  $
	\Tr[ T^{(1)}\gamma_{\Psi_{N}}^{(1,0)}]$ and $\Tr[ T^{(2)} \gamma_{\Psi_{N}}^{(0,1)}]$ are bounded uniformly in $N$. Since $T^{(1)}$ and $T^{(2)}$ have compact resolvents, up to a subsequence as $N\to \infty$, $\gamma_{\Psi_{N}}^{(1,0)}$ and $\gamma_{\Psi_{N}}^{(0,1)}$ converge strongly in trace class. Thus Theorem \ref{thm:deF} ensures that, up to a subsequence again, there exists a Borel probability measure $\mu$ supported on the set 
	$$\{(u,v):u,v\in L^2(\R^3), \|u\|=\|v\|=1\}$$
	such that
\begin{align} \label{eq:cv-quan-deF-H}
	\lim_{N\to \infty} \gamma^{(k,\ell )}_{\Psi_{N}} = \int |u^{\otimes k}\otimes v^{\otimes \ell}\rangle \langle u^{\otimes k} \otimes v^{\otimes \ell}| \d \mu (u,v), \quad \forall k,\ell=0,1,2,...
	\end{align}
	strongly in trace class.
	
	Next, thanks to the operator inequality \eqref{eq:low-opr-1}, 
	$$  \frac{c_1}{4} T^{(1)}_{x_1} + \frac{c_1}{4}T^{(1)}_{x_2} + \frac{c_1^2}{2} V^{(1)}(x_1-x_2) \geqslant -C.$$
	Therefore, from the convergence \eqref{eq:cv-quan-deF-H} and Fatou's lemma, it follows that
	\begin{align*}
	&\liminf_{N\to \infty}\Tr \Big[ \Big(  \frac{c_1}{4} T^{(1)}_{x_1} + \frac{c_1}{4}T^{(1)}_{x_2} + \frac{c_1^2}{2} V^{(1)}(x_1-x_2)\Big)  \gamma_{\Psi_{N}}^{(2,0)}\Big] \nn\\
	&\geqslant \Tr \Big[ \Big(  \frac{c_1}{4} T^{(1)}_{x_1} + \frac{c_1}{4}T^{(1)}_{x_2} + \frac{c_1^2}{2} V^{(1)}(x_1-x_2)\Big)  \int |u^{\otimes 2}\rangle \langle u^{\otimes 2}| \d \mu (u,v) \Big]\\
	&= \int \Big[ \frac{c_1}{2}\langle u, T^{(1)}u\rangle + \frac{c_1^2}{2}  \langle |u|^2, V^{(1)}*|u|^2\rangle \Big] \d \mu (u,v).
	\end{align*}
	Similarly, we have 
	\begin{align*}
	&\liminf_{N\to \infty}\Tr \Big[ \Big(  \frac{c_1}{4} T^{(2)}_{y_1} + \frac{c_2}{4}T^{(2)}_{y_2} + \frac{c_2^2}{2} V^{(2)}(y_1-y_2)\Big)  \gamma_{\Psi_{N}}^{(0,2)}\Big] \nn\\
	&\geqslant \int \Big[ \frac{c_2}{2}\langle v, T^{(2)}v\rangle + \frac{c_2^2}{2}  \langle |v|^2, V^{(2)}*|v|^2\rangle \Big] \d \mu (u,v)
	\end{align*}
	and
	\begin{align*}
	&\liminf_{N\to \infty}\Tr \Big[ \Big(  \frac{c_1}{2} T^{(1)}_{x} + \frac{c_2}{2}T^{(2)}_{y} + c_1c_2 V^{(12)}(x-y)\Big)  \gamma_{\Psi_{N}}^{(1,1)}\Big] \nn\\
	&\geqslant \int \Big[ \frac{c_1}{2}\langle u, T^{(1)}u\rangle + \frac{c_2}{2}\langle v, T^{(2)}v\rangle + c_1c_2  \langle uv, V^{(12)}*(uv)\rangle \Big] \d \mu (u,v).
	\end{align*}
	Summing these lower bounds, we can bound the right side of \eqref{eq:EN-DM} as
		\begin{align*}
	\liminf_{N\to \infty} \frac{\langle \Psi_N, H_N \Psi_N\rangle}{N} 
	\geqslant  \int \cE_{\rm H}[u,v] \d \mu(u,v) \geqslant e_{\rm H}.
	\end{align*}
	Combining with the upper bound in \eqref{eq:EN-DM}, we conclude \eqref{eq:energy-leading-H}:
	\begin{align*}
	\lim_{N\to \infty} \frac{\langle \Psi_N, H_N \Psi_N\rangle}{N} 
	= \int \cE_{\rm H}[u,v] \d \mu(u,v) = e_{\rm H}.
	\end{align*}
	The last equality means that $\mu$ is supported on the set of Hartree minimisers, i.e., $\{(e^{i\theta_1}u_0, e^{i\theta_2}v_0): \theta_1, \theta_2\in \R \}$, and hence \eqref{eq:cv-quan-deF-H} reduces to \eqref{eq:H-1pdm-CV}. Strictly speaking, we have proved  \eqref{eq:energy-leading-H} and \eqref{eq:H-1pdm-CV} for a subsequence as $N\to \infty$, but the convergence must hold for the whole sequence because the limits are unique. 
This completes the proof. 
\end{proof}

\subsection{Bogoliubov Hamiltonian} 

The aim of this section is to show that the Bogoliubov Hamiltonian $\bH$ defined in \eqref{eq:Bogoliubov_Hamiltonian} is precisely the same operator that arises from a suitable second quantization of the Hessian of the Hartree functional $\cE^{\operatorname{H}}$ evaluated at the minimiser $(u_0,v_0)$. We refer to \cite{LewNamSerSol-15} for discussions in one-component case. 
 
The main result of this section is Theorem \ref{thm:upper_lower_bound} below, which gives useful estimates on $\bH$. In order to formulate this result precisely, let us first recall the explicit (canonical) isomorphism that realizes $\cF_+$ in \eqref{eq:def-cF+} as a Fock space. 

We consider the Fock space with base space $\fh^{(1)}_+\oplus\fh^{(2)}_+$
\begin{equation}
\cG_+:=\bigoplus_{N=0}^{\infty}\big(\fh^{(1)}_+\oplus\fh^{(2)}_+ \big)^{\otimes_{\operatorname{sym}}N}.
\end{equation}
For a generic $f\oplus g\in\fh^{(1)}_+\oplus\fh^{(2)}_+$, let us denote the canonical creation and annihilation operators on $\cG_+$ as $Z^*(f\oplus g),Z(f\oplus g)$. The $N$-th sector of $\cG_+$ is interpreted as the space of states with \emph{exactly} $N$ total particles, regardless of which type they are. In fact (see, e.g.~\cite[Theorems 16 and 19]{Derezinski2006_CCR-CAR}) $\cG_+$ is isomorphic to $\cF_+$ through a natural isomorphism that preserves the CCR.
\begin{theorem} \label{thm:isomorphism}
	There exists a unitary operator $U:\cF_+\to\cG_+$ such that
	\begin{itemize}
		\item[(i)] $U(\Omega_{\cF_+})=\Omega_{\cG_+}$, where $\Omega_{\cF_+}$ is the vacuum of $\cF_+$ and $\Omega_{\cG_+}$ is the vacuum of $\cG_+$,
		\item[(ii)] for any $f\oplus g\in\fh^{(1)}_+\oplus\fh^{(2)}_+$
		\begin{align}
		Z^*(f\oplus g)U&=U\big(a^*(f)\otimes \mathbbm{1}+\mathbbm{1}\otimes b^*(g)\big) \nonumber\\
		Z(f\oplus g)U&=U\big(a(f)\otimes \mathbbm{1}+\mathbbm{1}\otimes b(g)\big)\nonumber.
 		\end{align}
	\end{itemize}
\end{theorem}

We define the second quantization of an operator $\cA$ on $\fh^{(1)}_+\oplus\fh^{(2)}_+$ by
\begin{equation} \label{eq:second_quantization_big}
\d\Gamma(\cA):=\sum_{m,n\geqslant 1}\langle f_m,\cA f_n\rangle Z^*(f_m)Z(f_n),
\end{equation}
where $(f_m)_{m=1}^\infty$ is an orthonormal basis of $\fh^{(1)}_+\oplus\fh^{(2)}_+$ belonging entirely to the domain of $\cA$, with an overall \emph{operator closure} being understood on the right side. Similarly, for  generic self-adjoint operators $A^{(1)}$ on $\fh^{(1)}_+$ and $A^{(2)}$ on $\fh^{(2)}_+$, we denote
\begin{equation} \label{eq:second_quantization}
\begin{split}
\d\Gamma^{(1)}(A^{(1)})\;:=&\;\sum_{m,n\geqslant 1}\langle u_m, A^{(1)}u_n\rangle a^*_m a_n\\
\d\Gamma^{(2)}(A^{(2)})\;:=&\;\sum_{m,n\geqslant 1}\langle v_m, A^{(2)}v_n\rangle b^*_m b_n\,,
\end{split}
\end{equation}
with $(u_m)_{m=1}^\infty$ an orthonormal basis of $\fh^{(1)}_+$ and $(v_n)_{n=1}^\infty$ an orthonormal basis of $\fh^{(2)}_+$.  In particular,
\begin{equation}
\cN_1\;:=\;\d\Gamma^{(1)}(\mathbbm{1})\,,\qquad  \cN_2\;:=\;\d\Gamma^{(2)}(\mathbbm{1})
\end{equation}
defines the number operators in each species' sectors, and 
\begin{equation}
\cN\;:=\;\cN_1+\cN_2
\end{equation}
defines the total number operator on $\cF_+$.

Within this formalism, it is natural to introduce the class of quadratic Hamiltonians in the Fock space $\cG_+$; through the isomorphism of Theorem \ref{thm:isomorphism}, such a class turns out to correspond to the class of Hamiltonians which are \emph{jointly} quadratic in $a$, $a^*$, $b$, and $b^*$, as is the case for $\bH$. Note that, already the operators defined by \eqref{eq:second_quantization_big}, which are quadratic in $Z$ and $Z^*$, are in general not separately quadratic in $a$, $a^*$ or $b$, $b^*$; this is true only if the operator $\cA$ is reduced with respect to the direct sum $\fh^{(1)}_+\oplus\fh^{(2)}_+$.

Let us consider two densely defined operators
\begin{equation*}
\begin{split}
&\cB_1:\cD(\cB_1)\subset\mathfrak{h}^{(1)}_+\oplus\mathfrak{h}^{(2)}_+\to\mathfrak{h}^{(1)}_+\oplus\mathfrak{h}^{(2)}_+\\
&\cB_2:\cD(\cB_2)\subset\big(\mathfrak{h}^{(1)}_+\big)^*\oplus\big(\mathfrak{h}^{(2)}_+\big)^*\to\mathfrak{h}^{(1)}_+\oplus\mathfrak{h}^{(2)}_+,
\end{split}
\end{equation*}
satisfying the properties
\begin{align*}
\cD(\cB_1)\subset J^*\cD(\cB_2), \quad  \cB_1^*=\cB_1, \quad J\cB_2J=\cB_2^*,
\end{align*}
where
\begin{equation}
J:\fh^{(1)}_+\oplus\fh^{(2)}_+\to\big(\fh^{(1)}_+\big)^*\oplus\big(\fh^{(2)}_+\big)^*,\qquad  J(f\oplus g):=\langle f\oplus g,\,\cdot\,\rangle_{\fh^{(1)}_+\oplus\fh^{(2)}_+}
\end{equation}
is the operator mapping a vector to the corresponding form. Let us form the operator
\begin{equation} \label{eq:def_B}
\cB:=\begin{pmatrix}
\cB_{1} & \cB_{2}\\
\cB_2^*& J\cB_{1}J^*
\end{pmatrix}
\end{equation}
acting on the space 
\begin{equation}
\fh:=\mathfrak{h}^{(1)}_+\oplus\mathfrak{h}^{(2)}_+\oplus\big(\mathfrak{h}^{(1)}_+\big)^*\oplus\big(\mathfrak{h}^{(2)}_+\big)^*.
\end{equation}
We define
\begin{equation} \label{eq:def_HB}
\bH_{\cB}:=\d\Gamma(\cB_1)+\frac{1}{2}\sum_{m,n\geqslant 1}\Big( \langle f_m,\cB_2 J f_n\rangle Z(f_m) Z(f_n)+  \overline{\langle f_m,\cB_2 J f_n\rangle}Z^*(f_m)Z^*(f_n)\Big)
\end{equation}
on the space
\begin{equation*}
\bigoplus_{n=0}^\infty \cD(\cB_1)^{\otimes_{\operatorname{sym}}n}.
\end{equation*}
In turns out that many properties of the quadratic Hamiltonian $\bH_\cB$ depend crucially on their analogues for the corresponding classical operator $\cB$. The following Lemma, which is a consequence of \cite[Theorem 2]{Nam-Napiorkowski-Solovej-2016}, collects some of them.
 
\begin{lemma} \label{lemma:diagonalization}
	Assume that $\cB_1>0$, $\cB>0$ and that $\cB_2$ is a Hilbert-Schmidt operator. Assume further that $\|\cB_1^{-1/2}\cB_{2}J\cB_1^{-1/2}\|<1$. Then:
	\begin{itemize}
	\item[(i)](Self-adjointness) Formula \eqref{eq:def_HB} defines a self-adjoint operator.
	\item[(ii)](Uniqueness of the ground state) $\bH_\cB$ has a unique ground state $\Phi^{\operatorname{gs}}_\cB$.
	\item[(iii)](Spectral gap) If, in addition, $\cB\geqslant\tau>0$ for some $\tau>0$, then
	\begin{equation} \label{eq:diagonalization_gap}
	\inf\sigma(\bH_{|\{\Phi^{\operatorname{gs}}_\cB\}^\perp})>\lambda(\bH_\cB),
	\end{equation} 
	where $\lambda(\bH_\cB)$ is the ground state energy of $\bH_\cB$.
	\end{itemize} 
    In particular, $\bH_\cB$ is bounded from below, namely there exists a constant $C_\cB>0$ such that
    \begin{equation} 
    \bH_\cB\geqslant -C_\cB.
    \end{equation}
\end{lemma}

\begin{proof}
	All the claims follow directly from Theorem 2 in \cite{Nam-Napiorkowski-Solovej-2016}: by such result there exists a unitary operator $\bU$ on $\cG_+$ such that
	\begin{equation} \label{eq:diagonalization}
	\bU \bH_\cB \bU^*=\d\Gamma(\xi)+\inf\sigma(\bH_{\mathcal{B}}),
	\end{equation}
	for a positive operator $\xi$ on $\fh^{(1)}_+\oplus\fh^{(2)}_+$. This proves the self-adjointness and implies that $\bU\Omega_{\cG_+}$ is the unique ground state of $\bH_\cB$. If, in addition, $\cB\geqslant \tau>0$, then $\xi\geqslant \tau>0$, and this implies \eqref{eq:diagonalization_gap}.
\end{proof}

Notice that in Lemma \ref{lemma:diagonalization} we require $\cB_2$ to be Hilbert-Schmidt, an assumption which is fulfilled in the application we are interested in, and which ensures the weaker hypotheses in \cite{Nam-Napiorkowski-Solovej-2016} to be satisfied.

Our interest in operators of the form $\bH_\cB$ is due to the fact that the Bogoliubov Hamiltonian \eqref{eq:Bogoliubov_Hamiltonian} can be realized as a quadratic Hamiltonian in the sense of \eqref{eq:def_HB}. More precisely,
\begin{equation}
\bH=U^*\bH_{\operatorname{Hess}\cE^{\operatorname{H}}[u_0,v_0]}U,
\end{equation}
where $\operatorname{Hess}\cE^{\operatorname{H}}[u_0,v_0]$ is the Hessian of the Hartree functional evaluated at the minimiser and $U$ is given by Theorem \ref{thm:isomorphism}. In the present context the Hessian of the Hartree functional is defined by the second term of a Taylor expansion around the minimiser $(u_0,v_0)$, that is,
\begin{equation} \label{eq:Taylor_expansion}
\begin{split}
\cE^{\operatorname{H}}&[u,v]=\cE^{\operatorname{H}}[u_0,v_0]\\
&+\frac{1}{2}\big\langle \sqrt{c_1}(u-u_0)\oplus\sqrt{c_2}(v-v_0),\operatorname{Hess}\cE^{\operatorname{H}}[u_0,v_0]\sqrt{c_1}(u-u_0)\oplus\sqrt{c_2}(v-v_0)\big\rangle\\
&+o\big(\|u-u_0\|,\|v-v_0\|\big).
\end{split}
\end{equation}
In \eqref{eq:Taylor_expansion} we are considering variations that are \emph{weighted} according to the relative populations of the two species.

In order to explicitly write the expression of $\hess$, let us introduce the following three integral operators $K^{(\alpha)}$, $\alpha\in\{1,2,12\}$, together with their kernels:
\begin{align}
&K^{(1)}:\mathfrak{h}^{(1)}_+\to\mathfrak{h}^{(1)},\qquad \label{eq:K_1} &&K^{(1)}(x,y):=V^{(1)}(x-y)u_0(x)u_0(y)\\
&K^{(2)}:\mathfrak{h}^{(2)}_+\to\mathfrak{h}^{(2)},\qquad\label{eq:K_2} &&K^{(2)}(x,y):=V^{(2)}(x-y)v_0(x)v_0(y)\\
&K^{(12)}:\mathfrak{h}^{(2)}_+\to\mathfrak{h}^{(1)},\qquad \label{eq:K_3} &&K^{(12)}(x,y):=V^{(12)}(x-y)u_0(x)v_0(y).
\end{align}
With the quantities introduced in \eqref{eq:def-Vmnpq} we can write
\begin{equation*}
\begin{split}
\langle u_m, K^{(1)}u_n\rangle&=V^{(1)}_{m00n}\\ 
\langle v_m, K^{(2)}v_n\rangle&=V^{(2)}_{m00n}\\ 
\langle u_m, K^{(12)}v_n\rangle&=V^{(12)}_{m00n}
\end{split}\qquad\qquad
\begin{split}
\langle u_m, K^{(1)}  \overline{u_n}\rangle &= V^{(1)}_{mn00}\\
\langle  v_m, K^{(2)} \overline{v_n}\rangle &= V^{(2)}_{mn00}\\
\langle u_m, K^{(12)}  \overline{v_n}\rangle &= V^{(12)}_{mn00}.
\end{split}
\end{equation*}
Moreover, as a straightforward consequence of Assumption ($A_2^{\mathrm{MF}}$), each such operator is Hilbert-Schmidt: indeed,
\begin{equation} \label{eq:K_HS}
\|K^{(1)}\|_{\operatorname{HS}}^2=\int \d x\d y |K^{(1)}(x,y)|^2\leqslant C^{(1)}+C^{(1)}\|\nabla u_0\|^2_2+C^{(1)}\|\nabla v_0\|_2^2< +\infty,
\end{equation}
and the same holds for $K^{(2)}$ and $K^{(12)}$.

In terms of the $K$'s, and of $h^{(1)}$ and $h^{(2)}$ defined in \eqref{eq:h's}, the Hessian of the Hartree functional reads
\begin{equation} \label{eq:Hess}
\begin{split}
&\qquad\hess=\\
&\begin{pmatrix}
h^{(1)}+c_1K^{(1)}& \sqrt{c_1c_2}K^{(12)}& c_1K^{(1)} J^* & \sqrt{c_1c_2} K^{(12)}J^* \\
\sqrt{c_1c_2}K^{(12)*}  & h^{(2)}+c_2K^{(2)} & \sqrt{c_1c_2}K^{(12)*}J^*  & c_2K_2^{(2)}J^* \\
c_1JK^{(1)}  &  \sqrt{c_1c_2}JK^{(12)}  &  Jh^{(1)}J^*+c_1JK^{(1)}J^*&\sqrt{c_1c_2}JK_1^{(12)}J^*  \\
\sqrt{c_1c_2}JK^{(12)*}  &  c_2JK^{(2)}   & \sqrt{c_1c_2}JK^{(12)*}J^* & Jh^{(2)}J^*+c_2JK^{(2)}J^*
\end{pmatrix}
\end{split}
\end{equation}
as a matrix-valued operator acting on $\mathfrak{h}^{(1)}_+\oplus \mathfrak{h}^{(2)}_+\oplus\big(\mathfrak{h}^{(1)}_+\big)^*\oplus\big(\mathfrak{h}^{(2)}_+\big)^*$.

The main result of this section is the following 

\begin{theorem}[Bounds on Bogoliubov Hamiltonian] \label{thm:upper_lower_bound}
	Under the same hypotheses of Theorem \ref{thm:mf_correction}, one has
	\begin{equation} \label{eq:upper_lower_bound}
	\begin{split}
	\frac{1}{C}\big(\d\Gamma^{(1)}(h^{(1)})+&\d\Gamma^{(2)}(h^{(2)})+\cN_1+\cN_2\big) -C \leqslant\bH\\
	&\leqslant\d\Gamma^{(1)}(h^{(1)})+\d\Gamma^{(2)}(h^{(2)})+C\cN_1+C\cN_2+C,
	\end{split}
	\end{equation}
	for some constant $C>0$. Consequently, $\bH$ has a self-adjoint extension by Friedrichs' method, still denoted by $\bH$, with the same form domain of $\d\Gamma^{(1)}(h^{(1)}+1)+\d\Gamma^{(2)}(h^{(2)}+1)$. Moreover, $\bH$ has a unique, non-degenerate  ground state $\Phi^{\operatorname{gs}}$:
\begin{equation} \label{eq:gap}
\inf\sigma(\bH_{|\{\Phi^{\operatorname{gs}}\}^\perp})>\langle\Phi^{\operatorname{gs}}, \bH\Phi^{\operatorname{gs}}\rangle.
\end{equation}
\end{theorem}

As a preparatory result towards the proof of Theorem \ref{thm:upper_lower_bound}, we show that $\operatorname{Hess}\cE^{\operatorname{H}}[u_0,v_0]$ has strictly positive bottom.

\begin{lemma}\label{lemma:positivity_Hessian}
	There exists a constant $\eta>0$ such that
	\begin{equation} \label{eq:positivity_Hessian}
	\operatorname{Hess}\cE^{\operatorname{H}}[u_0,v_0]\geqslant \eta.
	\end{equation}
\end{lemma}
This is clearly a non-degeneracy result for the minimiser $(u_0,v_0)$ of the Hartree functional.
\begin{proof}
	 We consider the decomposition
	 \begin{equation*}
	 \hess=\operatorname{Hess}_{h}+\operatorname{Hess}_{K},
	 \end{equation*}
	 where
	\begin{equation}  \label{eq:expansion_Hessian}
	\begin{split}
	&\operatorname{Hess}_{h}:=\begin{pmatrix}
	h^{(1)}&0&0 &0  \\
	0& h^{(2)}& 0 & 0 \\
	0& 0  &  J{h^{(1)}}J^*& 0 \\
	0 &  0& 0 & J{h^{(2)}}J^*
	\end{pmatrix}\\
	&\operatorname{Hess}_{K}:=
	\begin{pmatrix}
	c_1K^{(1)}& \sqrt{c_1c_2}K^{(12)}& c_1K^{(1)}J^*  & \sqrt{c_1c_2} K^{(12)}J^*  \\
	\sqrt{c_1c_2}K^{(12)*}  & c_2K^{(2)} & \sqrt{c_1c_2}K^{(12)*}J^*   & c_2K^{(2)}J^*  \\
	c_1JK^{(1)}   &  \sqrt{c_1c_2}JK^{(12)}  &  c_1JK^{(1)}J^*&\sqrt{c_1c_2}JK^{(12)}J^*  \\
	\sqrt{c_1c_2}JK^{(12)*}  &  c_2JK^{(2)}   & \sqrt{c_1c_2}JK^{(12)*}J^* & c_2JK^{(2)}J^*
	\end{pmatrix}.
	\end{split}
	\end{equation}
	
	First, we argue that $\operatorname{Hess}_h$ must be bounded away from zero. Indeed, since $(u_0,v_0)$ is the unique minimiser of the Hartree functional, one has $h^{(1)}>0$ on $\fh^{(1)}_+$ and $h^{(2)}>0$ on $\fh^{(2)}_+$. Since Assumptions ($A_1$) and ($A_2^{\mathrm{MF}}$) imply that $h^{(1)}$ and $h^{(2)}$ have compact resolvent, their spectra cannot accumulate to zero, and this implies the existence of some $\eta>0$ such that
	\begin{equation}\label{eq:positivity_h}
	\operatorname{Hess}_h\geqslant\eta.
	\end{equation}
	
	Concerning $\operatorname{Hess}_K$, we observe that it is a matrix-valued operator with structure
	\begin{equation*}
	\operatorname{Hess}_K=\begin{pmatrix}
	\cA& \cA J^*\\J\cA&J\cA J
	\end{pmatrix}
	\end{equation*}
	where
	\begin{equation}
	\cA:=\begin{pmatrix}
	c_1 K^{(1)}&\sqrt{c_1c_2} K^{(12)}\\ \sqrt{c_1c_2} K^{(12)}&c_2 K^{(2)}
	\end{pmatrix}.
	\end{equation}
	Since for any $f\oplus g\in\mathfrak{h}^{(1)}_+\oplus \mathfrak{h}^{(2)}_+$ one has
	\begin{equation*}
	\langle f\oplus g,\cA \,f\oplus g\rangle =c_1\langle f, K^{(1)}f\rangle + c_2\langle g,K^{(2)}g\rangle +2\,\sqrt{c_1c_2}\,\Re\,\langle f,K^{(12)}g\rangle,
	\end{equation*}
	it is straightforward to see that, by Cauchy-Schwarz,  Assumption ($A_2^{\mathrm{MF}}$) implies $\cA\geqslant 0$. Hence, $\operatorname{Hess}_K\geqslant 0$ follows. This result, together with \eqref{eq:positivity_h}, implies $\operatorname{Hess}\cE^{\operatorname{H}}[u_0,v_0]\geqslant\eta>0$.
\end{proof}

We can finally prove Theorem \ref{thm:upper_lower_bound}.

\begin{proof}[Proof of Theorem \ref{thm:upper_lower_bound}]
	We recognize that $\bH=U^*\bH_{\cB}U$ with $\cB=\operatorname{Hess}\cE^{\operatorname{H}}[u_0,v_0]$, and, comparing \eqref{eq:Hess} with \eqref{eq:def_B}, we deduce that
	\begin{equation}
	\cB_2=\begin{pmatrix}
	c_1K^{(1)} J^* & \sqrt{c_1c_2} K^{(12)}J^*\\
	\sqrt{c_1c_2} K^{(12)*}J^* & c_2K^{(2)} J^*
	\end{pmatrix}
	\end{equation}
	and
	\begin{equation*}
	\cB_1=\begin{pmatrix}
	h^{(1)}&0\\0&h^{(2)}
	\end{pmatrix}+\cB_2J.
	\end{equation*}
	Since $\cB_1>0$, $\hess>0$ by Lemma \ref{lemma:positivity_Hessian}, $\cB_2$ is Hilbert-Schmidt, and $\|\cB_1^{-1/2}\cB_{2}J\cB^{-1/2}_1\|<1$, we can apply Lemma \ref{lemma:diagonalization}. As a direct consequence we have that $\bH$ is bounded from below.
	
	We now show that the argument can be re-done so as to get the more refined lower bound \eqref{eq:upper_lower_bound}. Indeed, it is easy to see that, for $\eps>0$ small enough, the operator
	\begin{equation*}
	\begin{split}
	\cB_\eps:=&\operatorname{Hess}\cE^{\operatorname{H}}[u_0,v_0]\\
	&-\eps\begin{pmatrix}
	h^{(1)}+\mathbbm{1}&0&0&0\\0&h^{(2)}+\mathbbm{1}&0&0\\0&0&Jh^{(1)}J^*+\mathbbm{1}&0\\0&0&0&Jh^{(2)}J^*+\mathbbm{1}
	\end{pmatrix}
	\end{split}
	\end{equation*}
	is positive. Hence, for $\cB_\eps$ too we can apply Lemma \ref{lemma:diagonalization} and obtain the existence of a positive constant $C_{\cB_\eps}$ such that
	\begin{equation*}
	\bH_{\cB_\eps}\ge-C_{\cB_\eps}.
	\end{equation*}
	By \eqref{eq:def_HB}, last inequality is equivalent to
	\begin{equation*}
	\bH\geqslant \eps(\d\Gamma^{(1)}(h^{(1)})+\d\Gamma^{(2)}(h^{(1)})+\cN_1+\cN_2) -C_{\cB_\eps},
	\end{equation*}
	which is the first inequality we want to prove.
	
	To prove the second part of \eqref{eq:upper_lower_bound}, we remark that, for any $\widetilde{C}>0$,
	\begin{equation*}
	\d\Gamma^{(1)}(h^{(1)})+\d\Gamma^{(2)}(h^{(2)})+\widetilde{C}\cN_1+\widetilde{C}\cN_2-\bH=U^*\bH_{\widetilde{C}\mathbbm{1}-\operatorname{Hess}_K} U,
	\end{equation*}
	with $\operatorname{Hess}_K$ defined in \eqref{eq:expansion_Hessian}.
	Since all the $K^{(j)}$'s are bounded operators, $\operatorname{Hess}_K$ is bounded as well. Hence, for $\widetilde{C}$ large enough, $\widetilde{C}\mathbbm{1}-\operatorname{Hess}_K>0$. We can then apply Lemma \ref{lemma:diagonalization}, which ensures the existence of $C_{K}>0$ such that
	\begin{equation*}
	\bH_{\widetilde{C}\mathbbm{1}-\operatorname{Hess}_K}>-C_{K}.
	\end{equation*}
	Equivalently,
	\begin{equation}
	\bH\leqslant \d\Gamma^{(1)}(h^{(1)})+\d\Gamma^{(2)}(h^{(2)})+\widetilde{C}\cN_1+\widetilde{C}\cN_2+C_{K}.
	\end{equation}
	Thus \eqref{eq:upper_lower_bound} is proven by choosing $C:=\max\{\eps^{-1},C_{\cB_\eps},\widetilde{C},C_{K}\}$.
	
From the above proof, we already recognized that $\bH=U^*\bH_{\hess}U$, and all the hypotheses of Lemma \ref{lemma:diagonalization} are fulfilled if $\cB=\hess$. Hence, a direct application of Lemma \ref{lemma:diagonalization} shows that $\bH$ can be extended to a self-adjoint operator which has a unique ground state $\Phi^{\operatorname{gs}}$ and satisfies \eqref{eq:gap}:
$$
\inf\sigma(\bH_{|\{\Phi^{\operatorname{gs}}\}^\perp})>\langle\Phi^{\operatorname{gs}}, \bH\Phi^{\operatorname{gs}}\rangle.
$$
Moreover, the bounds  \eqref{eq:upper_lower_bound} implies that $\bH$ has the same  form domain of $\d\Gamma^{(1)}(h^{(1)}+1)+\d\Gamma^{(2)}(h^{(2)}+1)$.
\end{proof}

The estimate \eqref{eq:upper_lower_bound} will play an important role in Section \ref{sect:proof_correction} in the proof of Theorem \ref{thm:mf_correction}.

\begin{remark}After the identification of $\bH$ as the second quantization of the Hessian in the sense of \eqref{eq:def_HB}, the key point towards the proof of \eqref{eq:upper_lower_bound} was Lemma \ref{lemma:positivity_Hessian}. This is for us a mere consequence of Assumption ($A_2^{\mathrm{MF}}$) in which we require the positivity condition \eqref{eq:inequality_V} and the `miscibility' condition \eqref{eq:inequality_V1}. One could relax Assumption ($A_2^{\mathrm{MF}}$) by requiring the Hessian to be bounded away from zero in the first place; observe that when this is the case one should additionally require the uniqueness of the minimiser of the Hartree functional, while for us this is  another direct consequence of Assumption ($A_2^{\mathrm{MF}}$).\end{remark}

\begin{remark}
 A direct application of the diagonalization result of \cite[Theorem 2]{Nam-Napiorkowski-Solovej-2016} would allow to bound $\mathbb{H}$ from below in terms of an operator that is surely quadratic in $Z$, $Z^*$, but not separately in $a$, $a^*$ or $b$, $b^*$, thus preventing from obtaining the inequality \eqref{eq:upper_lower_bound} that is needed in the proof of Theorem \ref{thm:mf_correction}. We can fix this issue by further recognising (an observation that has no analogue for the one-component case) that the operator $\xi$ arising in the identity \eqref{eq:diagonalization} can be actually chosen to be \emph{reduced} with respect to $\fh^{(1)}_+\oplus\fh^{(2)}_+$, an additional feature that allows to estimate $\mathbb{H}$ from below by means of the two number operators. Such arguments are not needed for our main argument once assumption ($A_2^{\mathrm{MF}}$) is taken.
\end{remark}

\subsection{Estimate in the truncated two-component Fock space} \label{sect:truncated}
	

The claim of Theorem \ref{thm:mf_correction} is that the ground state energy of $\bH$ provides the second order correction to the ground state energy of $H_N^{\operatorname{MF}}$. Since $\bH$ and $H_N^{\operatorname{MF}}$ act on two different spaces, respectively, $\cF_+$ and $\mathcal{H}_{N_1,N_2}$, we rather compare $\bH$ with the operator $U_NH_N^{\operatorname{MF}}U_N^*$ on $\cF_+$ with a suitable unitary transformation $U_N$. This will lead to Theorem \ref{thm:truncated_bound} below, the main result of this section. 

The unitary operator $U_N$ is defined using ideas in \cite{LewNamSerSol-15}. More precisely, for arbitrary
\begin{equation*} 
\phi\in \big(\fh^{(1)}\big)^{\otimes_{\operatorname{sym}}j}\otimes \big(\fh^{(2)}\big)^{\otimes_{\operatorname{sym}}k}\qquad\textrm{and}\qquad \chi\in\big(\fh^{(1)}\big)^{\otimes_{\operatorname{sym}}\ell}\otimes \big(\fh^{(2)}\big)^{\otimes_{\operatorname{sym}}r}
\end{equation*}
we define $\phi \boxtimes\chi$ to be the function in $(h^{(1)})^{\otimes_{\operatorname{sym}}(j+\ell)}\otimes (h^{(2)})^{\otimes_{\operatorname{sym}}(k+r)}$ given by
\begin{equation}
\begin{split}
&(\phi \boxtimes\chi)(x_1,\dots, x_{j+\ell};y_1,\dots,y_{k+r})\;:=\;\frac{1}{\sqrt{j!\ell!(j+\ell!)}\sqrt{k!r!(k+r)!}}\;\times\\
&\times\sum_{\substack{\sigma\in \Sigma_{j+\ell} \\ \pi\in \Sigma_{k+r}}}\phi(x_{\sigma_1},\dots,x_{\sigma_j};y_{\pi_1},\dots,y_{\pi_k})\, \chi(x_{\sigma_{j+1}},\dots,x_{\sigma_{j+\ell}};y_{\pi_{k+1}},\dots,y_{\pi_{k+r}})\,,
\end{split}
\end{equation}
where $\Sigma_p$ is the symmetric group of $p$ elements.
A function  $\psi_N\in \cH_{N_1,N_2} $ decomposes uniquely as
\begin{equation}\label{eq:decomposition}
 \psi_N\;=\;\sum_{j=0}^{N_1}\sum_{k=0}^{N_2}\;\chi_{jk}\,\boxtimes\,\big(u_0^{\otimes(N_1-j)}\otimes v_0^{\otimes(N_2-k)} \big)
\end{equation}
for some $\chi_{jk}\in(\fh^{(1)}_+)^{\otimes_{\operatorname{sym}}j}\otimes (\fh^{(2)}_+)^{\otimes_{\operatorname{sym}}k}$, where for \emph{each summand} of the r.h.s.~of \eqref{eq:decomposition} it is understood that
\[
 \begin{split}
  \chi_{jk}\;&\equiv\;\chi_{jk}(x_1,\dots,x_j;y_1,\dots,y_k) \\
  u_0^{\otimes(N_1-j)}\;&\equiv\;u_0(x_{j+1})\cdots u_0(x_{N_1}) \\
  v_0^{\otimes(N_2-k)}\;&\equiv\;v_0(y_{k+1})\dots v_0(y_{N_2}) \,.
 \end{split}
\]

Thanks to the orthogonality relations
\begin{equation}
\begin{split}
 \Big\langle \:\chi_{jk}\,\boxtimes\,\big(u_0^{\otimes(N_1-j)}\otimes v_0^{\otimes(N_2-k)}&\big)\;,\;\chi_{\ell r}\,\boxtimes\,\big(u_0^{\otimes(N_1-\ell)}\otimes v_0^{\otimes(N_2-r)}\big)\;\Big\rangle\;= \\
 &=\|\chi_{jk} \|_2^2\;\delta_{j\ell}\,\delta_{kr}\,,
\end{split}
\end{equation}
it is easy to check that
\begin{equation} \label{eq:def_U}
U_{N}:\cH_{N_1,N_2}\longrightarrow \cF_+^{\,\leqslant N}\,,\qquad
U_{N}\psi_N\;:=\;(\chi_{jk})_{j+k\leqslant N}
\end{equation}
defines a unitary operator between Hilbert spaces, where
\begin{equation}
	\mathcal{F}_+^{\,\leqslant N}:=\;\bigoplus_{L=0}^N\;\Bigg(\bigoplus_{\substack{n+m=L \\ n\leqslant N_1,\,m\leqslant N_2}}\big( \mathfrak{h}^{(1)}\big)^{\otimes_{\operatorname{sym}}n} \otimes \big(\mathfrak{h}^{(2)}\big)^{\otimes_{\operatorname{sym}}m}\Bigg).
\end{equation}

The following is an analogue of the one-component result in \cite[Proposition 4.2]{LewNamSerSol-15}, whose proof  is merely algebraic. 

\begin{proposition}
	The action of the operator $U_N:\cH_{N_1,N_2}\to\cF^{\leqslant N}_+$ defined in \eqref{eq:def_U} can be written as
	\begin{equation} \label{eq:action_U_N}
	\big(U_N\psi_N\big)_{jk}=\Big(\big(Q^{(1)}\big)^{\otimes j}\otimes \big(Q^{(2)}\big)^{\otimes k}\frac{a^{N_1-j}_0\,\, b^{N_2-k}_0}{\sqrt{(N_1-j)!(N_2-k)!}}\Psi_N\Big)_{jk},
	\end{equation}
	where $Q^{(1)}=\mathbbm{1}-\ket{u_0}\bra{u_0}$, $Q^{(2)}=\mathbbm{1}-\ket{v_0}\bra{v_0}$, and $\Psi_N\in\cF_+^{\leqslant N}$ is the vector whose only one non-zero component coincides with $\psi_N$. For $\Phi\in\cF_+^{\leqslant N}$, the adjoint of $U_N$ acts as
	\begin{equation}
	U_N^*\;\Phi=\sum_{j=0}^{N_1}\sum_{k=0}^{N_2}\frac{1}{\sqrt{(N_1-j)!(N_2-k)!}}\Big((a^*_0)^{N_1-j}\,(b_0^*)^{N_2-k}\Phi\Big)_{N_1N_2}.
	\end{equation}
	Moreover, for all non-zero $m,n\in \bN$, the following identities hold true
\begin{equation} \label{eq:relations}
\begin{split}
U_Na^*_0 a_0U_N^*&=N_1-\cN_1,\\
U_Na_0^*a_m U_N^*&=\sqrt{N_1-\cN_1}a_m,\\
U_Na^*_m a_0U_N^*&=a^*_m\sqrt{N_1-\cN_1},\\
U_Na^*_ma_nU^*_N&=a^*_ma_n,
\end{split}\qquad 
\begin{split}
U_Nb^*_0 b_0U_N^*&=N_2-\cN_2,\\
U_Nb_0^*b_m U_N^*&=\sqrt{N_2-\cN_2}b_m,\\
U_Nb^*_m b_0U_N^*&=b^*_m\sqrt{N_2-\cN_2},\\
U_Nb^*_mb_nU^*_N&=b^*_mb_n.
\end{split}
\end{equation}
\end{proposition}

Notice that, as customary, all terms in the left side of \eqref{eq:relations}, are tacitly understood as $U_N\mathfrak{I}_{N_1,N_2}^* a^*_0 a_0 \mathfrak{I}_{N_1,N_2}U_N^*$ and the like, where $\mathfrak{I}_{N_1,N_2}$ is the lifting map from $\cH_{N_1,N_2}$ to the Fock slice with $N_1,N_2$ particles.

Thanks to \eqref{eq:relations} we can explicitly conjugate the many-body Hamiltonian  \eqref{eq:many_body_Hamiltonian_second_quantized} with $U_N$. The main result of this Subsection is the following Proposition, which provides a preliminary estimate valid on the space $\cF_+^{\leqslant M}$. The integer $M$ satisfies the property $M\leqslant N$, and we shall suitably fix it at the end of the proof. Here and henceforth it is understood that, eventually as $M$ and $N$ tend to infinity, $M$ must be chosen so as \emph{both} $M\leqslant N_1$ and $M\leqslant N_2$.

\begin{theorem}[Estimate on truncated Fock space] \label{thm:truncated_bound}
	Under the same hypotheses of Theorem \ref{thm:mf_correction}, given $M\leqslant N$, for any $\Phi\in \cF_+^{\leqslant M}\cap \cD[\bH]$, one has
	\begin{equation} \label{eq:truncated_bound}
	\big|\langle U_NH_N^{\operatorname{MF}} U_N^*\rangle_{\Phi}-Ne_{\operatorname{H}}-\langle \bH\rangle_\Phi\big|\leqslant C\sqrt{\frac{M}{N}}\langle \bH+C\rangle_\Phi
	\end{equation}
	for a positive constant $C$ (independent of $N$ and $M$).
\end{theorem}

We refer to \cite[Proposition 5.1]{LewNamSerSol-15} for the analogue in the one-component case. 

Let us remark that the condition $\Phi\in\cF_+^{\leqslant M}\cap\cD[\bH]$ implies, by Theorem \ref{thm:upper_lower_bound}, that $\Phi$ belongs to $\cF_+^{\leqslant M}\cap\cD[\d\Gamma^{(1)}(h^{(1)})+\d\Gamma^{(2)}(h^{(2)})]$. Using Assumption ($A_2^{\mathrm{MF}}$), one easily sees that this implies $U_N^*\Phi\in\cD[H_N^{\operatorname{MF}}]$, and hence \eqref{eq:truncated_bound} is well-defined for $\Phi\in\cF_+^{\leqslant M}\cap\cD[\bH]$.

Now we turn to the proof of  Theorem \ref{thm:truncated_bound}. We first compute exactly $U_NH_N^{\operatorname{MF}}U^*_N$, which will be done in Lemma \ref{lemma:splitting}. Then, we isolate from $U_NH_N^{\operatorname{MF}}U^*_N$ the leading contribution $Ne_{\operatorname{H}}$ and the second order correction $\bH$; this will be done in Lemma \ref{lemma:isolation}. Finally, we will show that all the remaining non-relevant terms can be estimated by the right-hand side of \eqref{eq:truncated_bound}.

\begin{lemma} \label{lemma:splitting}
	Let us define the following five operators on the domain $\cF^{\leqslant M}\cap\cD[\bH]$.
	\begin{equation} \label{eq:def_M_0}
	\begin{split}
	M_0:=&T^{(1)}_{00}(N_1-\cN_1)+T^{(2)}_{00}(N_2-\cN_2)+\frac{1}{2N}V^{(1)}_{0000}(N_1-\cN_1)(N_1-\cN_1-1)\\
	&+\frac{1}{2N}V^{(2)}_{0000}(N_2-\cN_2)(N_2-\cN_2-1)+\frac{1}{N}V^{(12)}_{0000}(N_1-\cN_1)(N_2-\cN_{2})\\
	&+\mu_1\cN_1+\mu_2\cN_2+\frac{c_1}{2}V^{(1)}_{0000}+\frac{c_2}{2}V^{(2)}_{0000}.
	\end{split}
	\end{equation}
	\begin{equation}\label{eq:def_M_1}
	\begin{split}
	M_1:=\sum_{m\geqslant 1}\Big[&  a^*_m\sqrt{N_1-\cN_1}\Big(  T^{(1)}_{m0} + V^{(1)}_{m000} \frac{N_1-\cN_1-1}{N}+V^{(12)}_{m000}\frac{N_2-\cN_2}{N} \Big)\\
	&+b^*_m\sqrt{N_2-\cN_2}\Big(T^{(2)}_{m0}+V^{(2)}_{m000}\frac{N_2-\cN_2-1}{N}+V^{(12)}_{m000}\frac{N_1-\cN_1}{N}\Big)\\
	& +\Big(  T^{(1)}_{0m} + V^{(1)}_{00m0} \frac{N_1-\cN_1-1}{N}+V^{(12)}_{00m0}\frac{N_2-\cN_2}{N} \Big)\sqrt{N_1-\cN_1}a_m \\
	&+\Big(T^{(2)}_{0m}+V^{(2)}_{00m0}\frac{N_2-\cN_2-1}{N}+V^{(12)}_{00m0}\frac{N_1-\cN_1}{N}\Big)\sqrt{N_2-\cN_2}b_m\Big].
	\end{split}
	\end{equation}
	\begin{equation}\label{eq:def_M_2}
	\begin{split}
	M_2:=\sum_{m,n\geqslant 1}\Big[& T^{(1)}_{mn}a^*_ma_n+T^{(2)}_{mn}b^*_mb_n-\mu_1\cN_1-\mu_2\cN_2-\frac{c_1}{2}V^{(1)}_{0000}-\frac{c_2}{2}V^{(2)}_{0000}\\
	&+\frac{1}{2N}V^{(1)}_{mn00}a^*_ma^*_n\sqrt{N_1-\cN_1}\sqrt{N_1-\cN_1-1}\\
	&+\frac{1}{2N}V^{(1)}_{00mn}\sqrt{N_1-\cN_1-1}\sqrt{N_1-\cN_1}a_ma_n\\
	&+\frac{1}{N}V^{(1)}_{m0n0}a^*_ma_n(N_1-\cN_1)+\frac{1}{N}V^{(1)}_{m00n}a^*_ma_n(N_1-\cN_1)\\
	&+\frac{1}{2N}V^{(2)}_{00mn}b^*_mb^*_n\sqrt{N_2-\cN_2}\sqrt{N_2-\cN_2-1}\\
	&+\frac{1}{2N}V^{(2)}_{00mn}\sqrt{N_2-\cN_2-1}\sqrt{N_2-\cN_2}b_mb_n\\
	&+\frac{1}{N}V^{(2)}_{m0n0}b^*_mb_n(N_2-\cN_2)+\frac{1}{N}V^{(2)}_{m00n}b^*_mb_n(N_2-\cN_2)\\
	&+\frac{1}{N}V^{(12)}_{mn00}a^*_mb^*_n\sqrt{N_1-\cN_1}\sqrt{N_2-\cN_2}\\
	&+\frac{1}{N}V^{(12)}_{00mn}\sqrt{N_1-\cN_1}\sqrt{N_2-\cN_2}a_mb_n\\
	&+\frac{1}{N}V^{(12)}_{m0n0}a^*_ma_n(N_2-\cN_2)+\frac{1}{N}V^{(12)}_{m00n}b^*_mb_n(N_1-\cN_1)\\
	&+\frac{1}{N}V^{(12)}_{m00n}a^*_mb_n\sqrt{N_1-\cN_1}\sqrt{N_2-\cN_2}\\
	&+\frac{1}{N}V^{(12)}_{m00n}\sqrt{N_1-\cN_1}\sqrt{N_2-\cN_2}a_mb^*_n
	\Big].
	\end{split}
	\end{equation}
	\begin{equation}\label{eq:def_M_3}
	\begin{split}
	M_3:=\frac{1}{N}\sum_{m,n,q\geqslant 1}\Big[&V^{(1)}_{mnp0}a^*_ma^*_na_p\sqrt{N_1-\cN_1}+V^{(2)}_{mnp0}b^*_mb^*_nb_p\sqrt{N_2-\cN_2}\\
	& +V^{(12)}_{mnp0}a^*_ma_pb^*_n\sqrt{N_2-\cN_2} +V^{(12)}_{m0np}a^*_mb^*_nb_p\sqrt{N_1-\cN_1}\\
	&+V^{(1)}_{p0mn}\sqrt{N_1-\cN_1}a^*_pa_ma_n + V^{(2)}_{p0mn}\sqrt{N_2-\cN_2} b^*_pb_mb_n\\
	& +V^{(12)}_{p0mn}\sqrt{N_2-\cN_2} a^*_pa_mb_n +V^{(12)}_{npm0}\sqrt{N_1-\cN_1} a_mb^*_pb_n\Big].
	\end{split}
	\end{equation}
	
	\begin{equation}\label{eq:def_M_4}
	\begin{split}
	M_4:=\sum_{m,n,p,q\geqslant 1}\Big[&\frac{1}{2N}V^{(1)}_{mnpq}a^*_ma^*_na_pa_q+\frac{1}{2N}V^{(2)}_{mnpq}b^*_mb^*_nb_pb_q\\
	&+\frac{1}{N}V^{(12)}_{mnpq}a^*_ma_pb^*_nb_q\Big].
	\end{split}
	\end{equation}
	Then, 
	\begin{equation}
	U_NH_N^{\operatorname{MF}}U^*_N=\sum_{j=0}^4 M_j.
	\end{equation}
\end{lemma}
\begin{proof}
	The proof is obtained by means of a direct computation that systematically uses the relations \eqref{eq:relations}. Notice that the term
	\begin{equation*}
	\mu_1\cN_1+\mu_2\cN_2+\frac{c_1}{2}V^{(1)}_{0000}+\frac{c_2}{2}V^{(2)}_{0000}
	\end{equation*}
	has been added in the last line of $M_0$ and subtracted in the first of $M_2$.
\end{proof}

We now show that the relevant terms can be isolated from $M_0$ and $M_2$ and that there is an exact cancellation in $M_1$, due to the fact that $(u_0,v_0)$ is the minimiser of the Hartree functional.

\begin{lemma} \label{lemma:isolation}
	For $M_0$, $M_1$, $M_2$ defined in \eqref{eq:def_M_0}-\eqref{eq:def_M_4}, one has the following re-arrangements.
	\begin{itemize}
		\item[(i)] (Isolation of the leading term from $M_0$) 
		\begin{equation} \label{eq:isolation_0}
		\begin{split}
		M_0=& Ne_H+\frac{1}{2N}V^{(1)}_{0000}\,\cN_1(\cN_1+1)+ \frac{1}{2N}V^{(2)}_{0000} \,\cN_2(\cN_2+1)\\
		&+\frac{1}{N}V^{(12)}_{0000}\,\cN_1\cN_2.
		\end{split}
		\end{equation}
		\item[(ii)] (Cancellation of the linear contribution to $M_1$)
		\begin{equation} \label{eq:cancellation}
		\begin{split}
		M_1=\frac{1}{N}\sum_{m\geqslant 1}\Big[&-V^{(1)}_{m000}a^*_m\sqrt{N_1-\cN_1}\,(\cN_1+1)-V^{(2)}_{m000}b^*_m\sqrt{N_2-\cN_2}\,(\cN_2+1) \\
		&-V^{12}_{m000}a^*_m\sqrt{N_1-\cN_1}\,\cN_2-V^{(12)}_{0m00}b^*_m\sqrt{N_2-\cN_2}\cN_1\\
		&-V^{(1)}_{00m0}(\cN_1+1)\sqrt{N_1-\cN_1}a_m-V^{(2)}_{00m0}\,(\cN_2+1)\sqrt{N_2-\cN_2} b_m\\
		&-V^{12}_{00m0}\,\sqrt{N_1-\cN_1}\,\cN_2\,a_m-V^{(12)}_{000m}\,\cN_1\sqrt{N_2-\cN_2}\,b_m\Big].
		\end{split}
		\end{equation}
		\item[(iii)] (Isolation of the Bogoliubov Hamiltonian from $M_2$)
		\begin{equation} \label{eq:isolation_2}
		\begin{split}
		M_2=\bH+\sum_{m,n\geqslant 1}\Big(\;&\frac{1}{2}V^{(1)}_{mn00}a^*_ma^*_n\frac{\sqrt{N_1-\cN_1}\sqrt{N_1-\cN_1-1}-N_1}{N}\\
		&\frac{1}{2}V^{(1)}_{00mn}\frac{\sqrt{N_1-\cN_1}\sqrt{N_1-\cN_1-1}-N_1}{N}a_ma_n\\
		&+\frac{1}{N}\big(V^{(1)}_{mn0}+V^{(1)}_{m00n}\big)a^*_ma_n(1-\cN_1)\\
		&+\frac{1}{2}V^{(2)}_{mn00}b^*_mb^*_n\frac{\sqrt{N_2-\cN_2}\sqrt{N_2-\cN_2-1}-N_2}{N}\\
		&+\frac{1}{2}V^{(2)}_{00mn}\frac{\sqrt{N_2-\cN_2}\sqrt{N_2-\cN_2-1}-N_2}{N}b_mb_n\\
		&+\frac{1}{N}\big(V^{(2)}_{m0n0}+V^{(2)}_{m00n}\big)b^*_mb_n(1-\cN_2)\\
		&+V^{(12)}_{mn00}a^*_mb^*_n\frac{\sqrt{N_1-\cN_1}\sqrt{N_2-\cN_2}-\sqrt{N_1N_2}}{N}\\
		&+V^{(12)}_{00mn}\frac{\sqrt{N_1-\cN_1}\sqrt{N_2-\cN_2}-\sqrt{N_1N_2}}{N}a_mb_n\\
		&-\frac{1}{N}V^{(12)}_{m0n0}a^*_ma_n\cN_2-\frac{1}{N}V^{(12)}_{0m0n}b^*_mb_n\cN_1\\
		&+V^{(12)}_{m00n}a^*_mb_n\frac{\sqrt{N_1-\cN_1}\sqrt{N_2-\cN_2}-\sqrt{N_1N_2}}{N}\\
		&+V^{(12)}_{0nm0}\frac{\sqrt{N_1-\cN_1}\sqrt{N_2-\cN_2}-\sqrt{N_1N_2}}{N}a_mb^*_n\;\Big).
		\end{split}
		\end{equation}
	\end{itemize}
\end{lemma}

\begin{proof}
	We recall that the minimum of the Hartree functional is
	\begin{equation}
	e_H=c_1T^{(1)}_{00}+c_2T^{(2)}_{00}+\frac{c_1^2}{2}V^{(1)}_{0000}+\frac{c_2^2}{2}V^{(2)}_{0000}+c_1c_2V^{(12)}_{0000}.
	\end{equation}
	A direct computation then yields \eqref{eq:isolation_0}.
	
	To prove \eqref{eq:cancellation}, we note that, since $(u_0,v_0)$ minimizes the Hartree functional, we have the identities
	\begin{equation}
	\begin{split}
	T^{(1)}_{m0}+c_1V^{(1)}_{m000}+c_2V^{(12)}_{m000}&=0\\
	T^{(2)}_{m0}+V^{(2)}_{m000}+c_1V^{(12)}_{m000}&=0.
	\end{split}
	\end{equation}
	Since $c_i=N_i/N$, last two identities yield an exact cancellation in $M_1$: for example, the contribution coming from the first line of the r.h.s. of \eqref{eq:def_M_1} reduces to
	\begin{equation*}
	\frac{1}{N}\sum_{m\geqslant 1}\Big[-V^{(1)}_{m000}a^*_m\sqrt{N_1-\cN_1}\,(\cN_1+1)-V^{(2)}_{m000}b^*_m\sqrt{N_2-\cN_2}\,(\cN_2+1)\Big].
	\end{equation*} 
	This allows us to bring $M_1$ to the form \eqref{eq:cancellation}.
	
	Finally, \eqref{eq:isolation_2} is obtained by a mere regrouping of terms. For example, the contribution from the second line of \eqref{eq:def_M_2} can be rewritten as
	\begin{equation*}
	\frac{c_1}{2}\sum_{m,n\geqslant1}V^{(1)}_{mn00}a^*_ma^*_n+\frac{1}{2}\sum_{m,n\geqslant1}V^{(1)}_{mn00}a^*_ma^*_n\frac{\sqrt{N_1-\cN_1}\sqrt{N_1-\cN_1-1}-N_1}{N},
	\end{equation*}
	having isolated the $N$-independent contribution. The same is done for all the other summands of \eqref{eq:def_M_2}. Recalling the definition \eqref{eq:Bogoliubov_Hamiltonian} of $\bH$, the outcome is \eqref{eq:isolation_2}.
\end{proof}

The final step in order to prove Theorem \ref{thm:truncated_bound} is the following Lemma, that provides the appropriate estimate for all the remainders $M_0-Ne_{\operatorname{H}}$, $M_1$, $M_2-\bH$, $M_3$, and $M_4$.

\begin{lemma}\label{lemma:estimates}
	There exists a constant $C>0$ such that, for any $\Phi\in \cF_+^{\leqslant M}\cap\cD[\bH]$
	\begin{align} \label{eq:estimate_M_0}
	\big|\langle M_0\rangle_\Phi-Ne_{\operatorname{H}}\big|&\leqslant C\frac{M}{N}\langle \cN\rangle_\Phi\\
	\label{eq:estimate_M_1}
	|\langle M_1\rangle_\Phi|&\leqslant C\sqrt{\frac{M}{N}}\langle \cN\rangle_\Phi\\
	\big|\langle M_2\rangle_\Phi-\langle\bH\rangle_\Phi\big|& \leqslant C\frac{M}{N}\langle\cN\rangle_{\Phi}  \label{eq:estimate_M_2}\\
	|\langle M_3\rangle_\Phi|&\leqslant C\sqrt{\frac{M}{N}}\Big(\langle \bH\rangle_\Phi+\langle\cN\rangle_\Phi+C\Big)  \label{eq:estimate_M_3}\\
	|\langle M_4\rangle_\Phi|&\leqslant C\frac{M}{N}\Big(\langle \bH\rangle_\Phi+\langle\cN\rangle_\Phi+C\Big)  \label{eq:estimate_M_4}
	\end{align}
\end{lemma}

Let us postpose the proof of Lemma \ref{lemma:estimates} and now conclude.
\begin{proof}[Proof of Theorem \ref{thm:truncated_bound}]
	Let us fix $\Phi\in\cF_+^{\leqslant M}\cap \cD[\bH]$. By Lemma \ref{lemma:splitting} we get 
	\begin{equation*}
	\langle U_NH_NU^*_N\rangle_\Phi= N e_{\operatorname{H}}+\langle\bH\rangle_\Phi+\langle M_0-Ne_{\operatorname{H}}\rangle_\Phi+\langle M_1\rangle_\Phi+\langle M_2-\bH\rangle_\Phi+\langle M_3\rangle_\Phi+\langle M_4\rangle_\Phi,
	\end{equation*}
	and hence, applying Lemma \ref{lemma:estimates}, we find 
	\begin{equation}
	\begin{split}
	\big|\langle U_N H_N^{\operatorname{MF}} U_N^*\rangle_\Phi-Ne_{\operatorname{H}}-\langle\bH\rangle_\Phi\big| \leqslant C\sqrt{\frac{M}{N}}\big(2\langle \bH\rangle_\Phi+5\langle \cN\rangle_\Phi+2C\big),
	\end{split}
	\end{equation}
	having used $M/N\leqslant \sqrt{M/N}$. Using positivity of $h^{(1)}$ and $h^{(2)}$ and Theorem \ref{thm:upper_lower_bound}, one finds
	\begin{equation*}
	\cN\leqslant \cN+\d\Gamma^{(1)}(h^{(1)})+\d\Gamma^{(2)}(h^{(2)})\leqslant C\,\bH + C^2.
	\end{equation*}
	This yields the bound
	\begin{equation}
	\big|\langle U_N H_N^{\operatorname{MF}} U_N^*\rangle_\Phi-Ne_{\operatorname{H}}-\langle\bH\rangle_\Phi\big| \leqslant C\sqrt{\frac{M}{N}}\big(\langle \bH\rangle_\Phi+C\big),
	\end{equation}
	for a suitable constant $C$.
\end{proof}

It remains to prove Lemma \ref{lemma:estimates}. We first state a technical Lemma that we will use through the proof.
\begin{lemma} \label{lemma:technical}
	Under the same assumptions of Theorem \ref{thm:mf_correction},
	\begin{align}
	&\d\Gamma^{(i)}(T^{(j)})\leqslant \alpha\bH+\alpha\cN+\alpha,\qquad \text{for } j\in\{1,2\} \label{eq:technical1}\\
	&\d\Gamma^{(1)}\Big(|V^{(12)}|*|v_0|^2\Big)\leqslant \beta\cN \label{eq:technical2}\\
	&\d\Gamma^{(2)}\Big(|V^{(12)}|*|u_0|^2\Big)\leqslant \gamma\cN\label{eq:technical3}
	\end{align}
	for positive constants $\alpha,\beta,\gamma$.
\end{lemma}
\begin{proof}
	Let us prove \eqref{eq:technical1} for the case $j=1$. By assumption ($A_2^{\mathrm{MF}}$) we know that, for every $\eps>0$, $V^{(1)}\geqslant -\eps\, C^{(1)}(1-\Delta)-\eps^{-1}$ and $V^{(12)}\geqslant -\eps\, C^{(12)}(1-\Delta)-\eps^{-1}$. Hence, we deduce	
	\begin{equation}
	V^{(1)}*|u_0|^2\geqslant -\eps C^{(1)}-\eps C^{(1)}(-\Delta)-\eps C^{(1)}\|u_0\|_{H^{1}}^2 - \eps^{-1}
	\end{equation}
	and
	\begin{equation} \label{eq:technical_intermediate}
	V^{(12)}*|v_0|^2\geqslant -\eps C^{(12)}-\eps C^{(12)}(-\Delta)-\eps C^{(12)}\|v_0\|_{H^{1}}^2 -\eps^{-1}.
	\end{equation}
	By recalling the definition of $h^{(1)}$ from \eqref{eq:h's}, the last two estimates imply the existence of a constant $\widetilde\alpha>0$ large enough such that
	\begin{equation}
	T^{(1)}\leqslant \widetilde\alpha h^{(1)}+\widetilde\alpha.
	\end{equation}
	Taking the second quantization $\d\Gamma^{(1)}(\cdot)$ of both sides and using \eqref{eq:upper_lower_bound} we obtain \eqref{eq:technical1} for $\alpha>0$ big enough. The same holds for $T^{(2)}$.
	
	To prove \eqref{eq:technical2} it is enough to note that, by Assumption ($A_2^{\mathrm{MF}}$), the multiplication operator $V^{(12)}*|v_0|^2$ is bounded. The desired inequality is hence trivial, since the second quantization of a bounded positive operator is always estimated by a multiple of the number operator. An analogous proof holds for \eqref{eq:technical3}.
	
\end{proof}

\begin{proof}[Proof of Lemma \ref{lemma:estimates}]
	Let us write
	\begin{align}
	\langle M_0-Ne_{\operatorname{H}}\rangle_\Phi&=M_0^{(1)}+M_0^{(2)}+M_0^{(12)}\\
	\langle M_1\rangle_\Phi&=M_1^{(1)}+M_1^{(2)}+M_1^{(12)}\\
	\langle M_2-\bH\rangle_\Phi&=M_2^{(1)}+M_2^{(2)}+M_2^{(12)}\\
	\langle M_3\rangle_\Phi&=M_3^{(1)}+M_3^{(2)}+M_3^{(12)}\\
	\langle M_4\rangle_\Phi&=M_4^{(1)}+M_4^{(2)}+M_4^{(12)},
	\end{align}
	where, in self-explanatory notation, each summand with label ${(\alpha)}$ contains all the terms depending on the interaction potential $V^{(\alpha)}$. We will estimate the $M^{(12)}_k$'s; all the other terms do not involve interactions between particles of different type, and hence, they are on the same footing as the terms estimated in \cite[Proposition 5.2]{LewNamSerSol-15}.
	
	Let us consider $M_0^{(12)}$. Since $\Phi\in\cF_+^{\leqslant M}$, we have $\langle \cN_i\rangle_\Phi\leqslant \langle\cN\rangle_\Phi\leqslant M$, and hence
	\begin{equation} \label{eq:final_M0}
	\big|M_0^{(12)}\big|=\Big|V^{(12)}_{0000}\langle\frac{\cN_1\cN_2}{N}\rangle_\Phi\Big|\leqslant  \cK_0\frac{M}{N}\langle \cN\rangle_{\Phi},
	\end{equation}
	for $\cK_0=|V^{(12)}_{0000}|$.
	
	Let us now consider $M_1^{(12)}$. By a Cauchy-Schwarz inequality we can write
	\begin{equation*}
	\begin{split}
	\Big|\sum_{m\geqslant 1}V^{(12)}_{m000}\frac{1}{N}\langle a^*_m\sqrt{N_1-\cN_1}\cN_2\rangle_\Phi +\operatorname{h.c.}\Big|&\le	\frac{2}{N}\Big(\sum_{m\geqslant 1}|V^{(12)}_{m000}|^2\Big)^{1/2}\\
	&\qquad\times\Big(\sum_{m\geqslant 1}\langle a^*_ma_m\rangle_\Phi\langle (N_1-\cN_1)\cN_2^2\rangle_\Phi\Big)^{1/2}\\
	&\leqslant\frac{2\,\|K^{(12)}\|_{\operatorname{HS}}}{N}\Big(N\langle\cN\rangle_\Phi\langle\cN^2\rangle_\Phi\Big)^{1/2}\\
	&\leqslant 2\,\|K^{(12)}\|_{\operatorname{HS}}\,\sqrt{\frac{M}{N}}\langle \cN\rangle_\Phi.
	\end{split}
	\end{equation*}
	In the second and third step we have used $N_1\leqslant N$, positivity of $\cN_1$, the inequality $\langle\cN\rangle_\Phi\leqslant M$, together with the property
	\begin{equation*}
	\begin{split}
	\sum_{m\geqslant 1}|V^{(12)}_{m000}|^2=\sum_{m\geqslant 1}|\langle u_m, K^{(12)} u_0\rangle|^2&\leqslant\sum_{m\geqslant 1}\langle u_m, K^{(12)}\ket{u_0}\bra{u_0} K^{(12)} u_m\rangle\\
	&\leqslant \|K^{(12)}\|_{\operatorname{HS}}^2<+\infty.
	\end{split}
	\end{equation*}
	There is another summand in $M_1^{(12)}$, but it differs from the one we just estimated only by the interchange of the two components; for this reason, we omit the details of its estimate. Thus,
	\begin{equation} \label{eq:final_M1}
	|M_1^{(12)}|\leqslant 4 \|K^{(12)}\|^2_{\operatorname{HS}}\sqrt{\frac{M}{N}}\langle \cN\rangle_\Phi.
	\end{equation}
	
	Let us consider $M_2^{(12)}$, whose expression is
	\begin{align}
	M_2^{(12)}=&\sum_{m,n\geqslant 1}\Big[\frac{1}{N}V^{(12)}_{mn00}\langle a^*_mb^*_n\big(\sqrt{N_1-\cN_1}\sqrt{N_2-\cN_2}-\sqrt{N_1N_2}\big)\rangle_\Phi+\operatorname{h.c.}\Big] \label{eq:M_2_V12_1}\\
	&-\sum_{m,n\geqslant 1}\Big[\frac{1}{N}V^{(12)}_{m0n0}\langle a^*_ma_n\cN_2\rangle_\Phi\Big] \label{eq:M_2_V12_2}\\
	&-\sum_{m,n\geqslant 1}\Big[\frac{1}{N}V^{(12)}_{0m0n}\langle b^*_mb_n\cN_1\rangle_\Phi\Big]\label{eq:M_2_V12_3}\\
	&+\sum_{m,n\geqslant 1}\Big[\frac{1}{N}V^{(12)}_{m00n}\langle a^*_mb_n\big(\sqrt{N_1-\cN_1}\sqrt{N_2-\cN_2}-\sqrt{N_1N_2}\big)\rangle_\Phi+\operatorname{h.c.}\Big], \label{eq:M_2_V12_4}
	\end{align}
	and let us treat the four summands one by one. First let us define the operator
	\begin{equation*}
	X:=\sqrt{\frac{N_1-\cN_1}{N_1}}\sqrt{\frac{N_2-\cN_2}{N_2}}.
	\end{equation*}
	By a Cauchy-Schwarz inequality we get
	\begin{equation*}
	\begin{split}
	|\eqref{eq:M_2_V12_1}|\;\leqslant& \;2\,\sqrt{\frac{N_1N_2}{N^2}}\Big(\sum_{m,n\geqslant 1}|\langle u_m,K^{(12)} \overline{v_n}\rangle|^2\Big)^{1/2}\Big(\sum_{m,n\geqslant 1}\langle a^*_mb^*_na_mb_n\rangle_\Phi\Big)^{1/2}\langle\, (X-\mathbbm{1})^2\,\rangle_\Phi^{1/2}\\
	\;\leqslant&\;2\,\sqrt{\frac{N_1N_2}{N^2}}\Big(\sum_{m\geqslant 1}\langle u_m, \big|K^{(12)}\big|^2 u_m\rangle\Big)^{1/2} \langle \cN_1\cN_2\rangle_\Phi^{1/2}\Big\langle \Big(-\frac{\cN_1}{N_1}-\frac{\cN_2}{N_2}+\frac{\cN_1\cN_2}{N_1N_2}\Big)^2\Big\rangle_\Phi^{1/2},
	\end{split}
	\end{equation*}
	having used the estimate $(X-\mathbbm{1})^2\leqslant (X^2-\mathbbm{1})^2$. Now, using $N_1N_2\leqslant N^2$, $\langle\cN_i\rangle_\Phi\leqslant\langle\cN\rangle_\Phi\leqslant M$, and the fact that $K^{(12)}$ is Hilbert-Schmidt, we obtain
	\begin{equation*}
	\begin{split}
	|\eqref{eq:M_2_V12_1}|\;\leqslant& \;2 \|K^{(12)}\|_{\operatorname{HS}}\,\frac{M^{1/2}}{N_1N_2}\langle\cN\rangle_\Phi^{1/2} \big\langle\big( N_2\cN_1+N_1\cN_2-\cN_1\cN_2\big)^2\big\rangle^{1/2}.
	\end{split}
	\end{equation*}
	Since $\cN_1\cN_2\leqslant N_2\cN_1+N_1\cN_2$ on $\cF_+^{\leqslant M}$, we finally get
	\begin{equation} \label{eq:M_2_V12_1final}
	\begin{split}
	|\eqref{eq:M_2_V12_1}|\;\leqslant&\; 4 \|K^{(12)}\|_{\operatorname{HS}}\,\frac{M^{1/2}}{N_1N_2}\langle\cN\rangle_\Phi^{1/2} \big\langle\big( N_2\cN_1+N_1\cN_2\big)^2\big\rangle^{1/2}   \leqslant\widetilde{\cK_2}\frac{M}{N}\langle\cN\rangle_\Phi,
	\end{split}
	\end{equation}
	for some $\widetilde{\cK_2}>0$ .

	To estimate \eqref{eq:M_2_V12_2} we note that
	\begin{equation*}
	\eqref{eq:M_2_V12_2}=-\Big\langle\d\Gamma^{(1)}(V^{(12)}*|v_0|^2)\frac{\cN_2}{N}\Big\rangle_\Phi,
	\end{equation*}
	and, by \eqref{eq:technical2},
	\begin{equation} \label{eq:M_2_V12_2final}
	|\eqref{eq:M_2_V12_2}|\leqslant \beta\frac{M}{N}\langle \cN\rangle_\Phi.
	\end{equation}
	Analogously,
	\begin{equation}  \label{eq:M_2_V12_3final}
	|\eqref{eq:M_2_V12_3}|\leqslant \gamma\frac{M}{N}\langle\cN\rangle_\Phi
	\end{equation}
	for $\gamma$ given by \eqref{eq:technical3}.
    
    To estimate \eqref{eq:M_2_V12_4} we write
	\begin{equation*}
	\eqref{eq:M_2_V12_4}= \frac{1}{N}\sum_{\substack{j\geqslant 1,\,\,k\geqslant 0\\j+k\leqslant M}} j(k+1)\langle \Phi_{j,k},K^{(12)}_{1,1}\Phi_{j-1,k+1}\rangle+\operatorname{h.c.},
	\end{equation*}
	where $\Phi=(\Phi_{j,k})_{jk}\in\cF^{\leqslant M}_+$ and $K^{(12)}_{1,1}$ is the integral operator $K^{(12)}$ defined in \eqref{eq:K_3} and taken with kernel $K^{(12)}(x_1,y_1)$. By using Cauchy-Schwarz we get
	\begin{equation*}
	\begin{split}
	|\eqref{eq:M_2_V12_4}|\leqslant& \frac{2}{N} \sum_{\substack{j\geqslant 1,\,\,k\geqslant 0\\j+k\leqslant M}} j(k+1) \Big(\langle\Phi_{j,k},\Phi_{j,k}\rangle+\|K^{(12)}\|_{\operatorname{op}}^2 \langle \Phi_{j-1,k+1},\Phi_{j-1,k+1}\rangle\Big)\\
	=& \frac{2}{N}(1+\|K^{(12)}\|_{\operatorname{op}})\langle\Phi,\cN_1 (\cN_2+1)\Phi\rangle,
	\end{split}
	\end{equation*}
	and hence, the inequality $\cN_i\leqslant\cN\leqslant M$ valid on $\cF_+^{\leqslant M}$ yields
	\begin{equation} \label{eq:M_2_V12_4final}
	|\eqref{eq:M_2_V12_4}|\leqslant \cK_2'\frac{M}{N}\langle\cN\rangle_\Phi,
	\end{equation}
	for some $\cK_2'>0$.
	Putting together \eqref{eq:M_2_V12_1final}, \eqref{eq:M_2_V12_2final}, \eqref{eq:M_2_V12_3final}, and \eqref{eq:M_2_V12_4final} we conclude
	\begin{equation} \label{eq:final_M2}
	\big|M_2^{(12)}\big|\leqslant \cK_2\frac{M}{N}\langle\cN\rangle_\Phi,
	\end{equation}
	with $\cK_2:=\widetilde{\cK_2}+\beta+\gamma+\cK_2'$. 
	
	Let us consider $M_3^{(12)}$, whose expression is
	\begin{align}
	M_3^{(12)}=&\sum_{m,n,p\geqslant 1}\Big[\frac{1}{N}V^{(12)}_{mnp0} \langle a^*_ma_pb^*_n\sqrt{N_2-\cN_2}\rangle_\Phi+\operatorname{h.c.}\Big] \label{eq:M_3_V12_1}\\
	&+\sum_{m,n,p\geqslant 1}\Big[\frac{1}{N}V^{(12)}_{m0np} \langle a^*_mb^*_nb_p\sqrt{N_1-\cN_1} \rangle_\Phi +\operatorname{h.c.}\Big] \label{eq:M_3_V12_2}.
	\end{align}
	First, we notice that
	\begin{equation} \label{eq:M_3_V12_1a}
	\eqref{eq:M_3_V12_1}=\langle \Phi,U_N\sum_{j=1}^{N_1}\sum_{k=1}^{N_2}Q^{(1)}_{x_j}\otimes Q^{(2)}_{y_k}V^{(12)}(x_j-y_k)Q^{(1)}_{x_j}\otimes P^{(2)}_{y_k} U_N^*\,\Phi\rangle+\operatorname{h.c.}.
	\end{equation}
	Now, by splitting $V^{(12)}$ into positive and negative part and using Cauchy-Schwarz, one obtains the inequality
	\begin{equation}
	\begin{split}
	Q^{(1)}_{x}\otimes& Q^{(2)}_{y}V^{(12)}(x-y)Q^{(1)}_{x}\otimes P^{(2)}_{y}+ Q^{(1)}_{x}\otimes P^{(2)}_{y}V^{(12)}(x-y)Q^{(1)}_{x}\otimes Q^{(2)}_{y}\\
	\leqslant&\, \eps^{-1}Q^{(1)}_{x}\otimes Q^{(2)}_{y}\big|V^{(12)}(x-y)\big|Q^{(1)}_{x}\otimes Q^{(2)}_{y}\\
	&+\eps \,Q^{(1)}_{x}\otimes P^{(2)}_{y}\big|V^{(12)}(x-y)\big|Q^{(1)}_{x}\otimes P^{(2)}_{y},
	\end{split}
	\end{equation}
	whence, substituting into \eqref{eq:M_3_V12_1a}, one gets
	\begin{equation} \label{eq:M_3_V12_1b}
	\begin{split}
	\eqref{eq:M_3_V12_1}\leqslant& \frac{1}{\eps N} \sum_{m,n,p,q\geqslant 1} \big|V^{(12)}\big|_{mnpq}\langle a^*_mb^*_na_pb_q \rangle_\Phi \\
	&+\frac{\eps}{N} \big\langle  \d\Gamma^{(1)} \Big(\big|V^{(12)}\big|*|v_0|^2\Big) (N_2-\cN_2)\big\rangle_\Phi
	\end{split}
	\end{equation}
	for some $\eps>0$ that we are going to specify in a moment.
	We first focus on the first summand of the r.h.s of \eqref{eq:M_3_V12_1b}. Using the expression in components $\Phi=(\Phi_{j,k})_{jk}\in\cF_+^{\leqslant M}$, we can write
	\begin{equation*}
	\begin{split}
	\frac{1}{\eps N}\sum_{m,n,p,q\geqslant 1}\big|V^{(12)}_{mnpq}\big|\langle\Phi,a^*_ma_pb^*_nb_q\Phi&\rangle=\;\frac{1}{\eps N}\sum_{\substack{j,k\geqslant 1,\\j+k\leqslant M}}jk\langle\Phi_{j,k}, \big|V(x_1-y_1)\big|\Phi_{j,k}\rangle\\
	\leqslant&\;\frac{C^{(12)}}{\eps N}\sum_{\substack{j,k\geqslant 1,\\j+k\leqslant M}}jk\langle\Phi_{j,k},\big(1-\Delta_{x_1}-\Delta_{y_1}\big) \Phi_{j,k}\rangle\\
	=&\;\frac{C^{(12)}}{\eps N}\langle\Phi,\big (\cN_1\cN_2+\d\Gamma^{(1)}(T^{(1)})\cN_2+\cN_1\d\Gamma^{(2)}(T^{(2)})\big)\Phi\rangle,
	\end{split}
	\end{equation*}
	having used Assumption ($A_2^{\mathrm{MF}}$) in the second step. Thanks to the inequality $\cN_i\leqslant \cN\leqslant M$ valid on $\cF^{\leqslant M}_+$ and Lemma \ref{lemma:technical}, we obtain
	\begin{equation*}
	\frac{1}{\eps N} \sum_{mnpq} |V^{(12)}_{mnpq}|\langle\Phi,a^*_mb^*_na_pb_q\Phi\rangle \leqslant \widetilde{\cK_3}\frac{M}{\eps N}\big(\langle \cN\rangle_\Phi+\langle\bH\rangle_\Phi+\widetilde{\cK_3}\big),
	\end{equation*} 
	for some $\widetilde{\cK_3}>0$. The second summand in the r.h.s of \eqref{eq:M_3_V12_1b}, in turn, is estimated using \eqref{eq:technical2} as
	\begin{equation*}
	\frac{\eps}{N} \big\langle\Phi,\d\Gamma^{(1)}\Big(\big|V^{(12)}\big|*|v_0|^2\Big)(N_2-\cN_2)\big\rangle_\Phi \leqslant  \alpha\eps \big(\langle \cN\rangle_\Phi+\langle\bH\rangle_\Phi\big).
	\end{equation*}
	By finally choosing $\eps=(M/N)^{1/2}$, the last two inequalities yield
	\begin{equation*}
	|\eqref{eq:M_3_V12_1}|\leqslant (\widetilde{\cK_3}+\alpha)\sqrt{\frac{M}{N}}\big(\langle \cN\rangle_\Phi+\langle\bH\rangle_\Phi+\widetilde{\cK_3}\big).
	\end{equation*} 
	The term \eqref{eq:M_3_V12_2} differs from \eqref{eq:M_3_V12_1} only by the exchange of the two components, and hence, it is treated analogously. This yields
	\begin{equation} \label{eq:final_M3}
	\big|M_3^{(12)}\big|\leqslant \cK_3\sqrt{\frac{M}{N}} \big(\langle \cN\rangle_\Phi+\langle\bH\rangle_\Phi+\cK_3\big),
	\end{equation}
	for a positive constant $\cK_3$ big enough.
	
	Let us finally consider $M_4^{(12)}$. Using $\Phi=(\Phi_{j,k})_{jk}$, we can write
	\begin{equation*}
	\begin{split}
	M_4^{(12)}=&\;\frac{1}{N}\sum_{\substack{j,k\geqslant 1\\j+k\leqslant M}}jk\langle\Phi_{j,k}, V(x_1-y_1)\Phi_{j,k}\rangle\\
	\leqslant&\;\frac{C^{(12)}}{N}\sum_{\substack{j,k\geqslant 1\\j+k\leqslant M}}jk\langle\Phi_{j,k},\big(1-\Delta_{x_1}-\Delta_{y_1}\big) \Phi_{j,k}\rangle\\
	\leqslant&\;\frac{C^{(12)}}{N}\langle\Phi,\big(\cN_1\cN_2+\d\Gamma^{(1)}(T^{(1)})\cN_2+\cN_1\d\Gamma^{(2)}(T^{(2)})\big)\Phi\rangle,
	\end{split}
	\end{equation*}
	having used Assumption ($A_2^{\mathrm{MF}}$) in the second step. By the inequality $\cN_i\leqslant \cN\leqslant M$, valid on $\cF^{\leqslant M}_+$ and \eqref{eq:technical1}, we get
	\begin{equation} \label{eq:final_M4}
	M_4^{(12)} \leqslant \;\cK_4\frac{M}{ N}\big(\langle \cN\rangle_\Phi+\langle\bH\rangle_\Phi+\cK_4\big),
	\end{equation} 
	for some constant $\cK_4>0$.
	
	Eqs. \eqref{eq:final_M0}, \eqref{eq:final_M1}, \eqref{eq:final_M2}, \eqref{eq:final_M3}, and \eqref{eq:final_M4}, together with their analogous for the terms depending on $V^{(1)}$ and $V^{(2)}$ yield the desired claim, provided that the overall constant $C$ is chosen large enough.
\end{proof}

\subsection{Localization in Fock space}
Theorem \ref{thm:truncated_bound} provides an estimate for the expectation value the difference $U_N H_N^{\operatorname{MF}}U_N^*-N e_{\operatorname{H}}-\bH$ in the truncated space $\cF_+^{\leqslant M}$. In what follows we recall a result that allows us to localize the energy of a state in the space $\cF_+^{\leqslant M}$. As we shall see, at the end of the proof we will be able to choose $M\ll N$ in such a way that the localization produces only negligible remainders. This idea goes back to \cite[Theorem A.1]{LieSol-01} and we will follow the simplified representation in \cite[Proposition 6.1]{LewNamSerSol-15}. 

Consider two smooth, real functions $f$ and $g$ such that $0\leqslant f,g\leqslant1$, $f^2+g^2=1$, $f(x)=1$ for $|x|\le1/2$, and $f(x)=0$ for $|x|\geqslant 1$. By spectral calculus, we define
\begin{equation}
\begin{split}
f_M:=&f(\cN/M)\\
g_M:=&g(\cN/M).
\end{split}
\end{equation}
Let us also define the orthogonal projection $\cP_L$ onto the sector of $\cF_+$ with exactly $L$ particles, namely the subspace
\begin{equation*}
\bigoplus_{\substack{j\geqslant 0,\,\,k\geqslant 0\\ j+k= L}}\big(\fh^{(1)}_+\big)^{\otimes_{\operatorname{sym}}j}\otimes \big(\fh^{(2)}_+\big)^{\otimes_{\operatorname{sym}}k}= U^*\Big( \big(\fh^{(1)}_+\oplus\fh^{(2)}_+\big)^{\otimes_{\operatorname{sym}}L}\Big).
\end{equation*}
Here $U:\cF_+\to\cG_+$ is the unitary operator given by Theorem \ref{thm:isomorphism}.
\begin{proposition} \label{prop:localization}
	Let $A$ be a non-negative operator on $\cF_+$ such that $\cP_L \cD(A)\subset\cD(A)$ for any $L\in\bN$. Suppose moreover that there exists $\sigma\geqslant 0$ such that $\cP_LA\cP_{L'}=0$ when $|L-L'|>\sigma$. Then
		\begin{equation} \label{eq:localization}
		\pm A-f_MAf_M-g_MAg_M\leqslant \frac{C_f\sigma^3}{M^2}A_0,
		\end{equation}
	where $A_0:=\sum_{L\in\bN} \cP_LA\cP_L$ and $C_f$ is a positive constant depending only on $f$.
\end{proposition}
This result is a variant of the IMS formula, that can be found in \cite[Proposition 6.1]{LewNamSerSol-15}, which is, in turn, an adaptation of \cite[Theorem A.1]{LieSol-01}. As a direct consequence of Proposition \ref{prop:localization}, Lemma \ref{lemma:localization} provides the precise estimates that enable us to localize the energy of a state in $\cF_+^{\leqslant N}$ into the subspace $\cF^{\leqslant M}_+$.

Let us define the operator
\begin{equation}
\widetilde{H}_{N}:=U_{N}H_{N}^{\operatorname{MF}}U^*_{N}-Ne_H.
\end{equation} 
If $\psi_N^{\operatorname{MF}}$ is a ground state of $H_N^{\operatorname{MF}}$, then, by unitarity,
\begin{equation} \label{eq:gs_tilde}
\Phi_N:=U_N\psi_N^{\operatorname{MF}}
\end{equation}
is a ground state of $\widetilde{H}_N$. We will use the following notations for the ground state energies of $\widetilde{H}_{N}$ and $\bH$
\begin{equation}
\begin{split}
\lambda(\widetilde{H}_N):=&\langle \Phi_N,\widetilde{H}_N\Phi_N\rangle=E^{\operatorname{MF}}_N-Ne_{\operatorname{H}}\\
\lambda(\bH):=&\langle \Phi^{\operatorname{gs}},\bH\Phi^{\operatorname{gs}}\rangle.
\end{split}
\end{equation}

\begin{lemma} \label{lemma:localization} There exist positive constants $\kappa_1,\kappa_2$ such that, for any $M\leqslant N$,
	\begin{equation} \label{eq:localization_bogoliubov}
	\pm\big(\bH-f_M\bH f_M-g_M\bH g_M\big)\leqslant \frac{\kappa_1}{M^2}(\bH+\kappa_1)
	\end{equation}
	and
	\begin{equation} \label{eq:localization_many_body}
	\begin{split}
	\pm\big( \widetilde{H}_{N}-&f_M\widetilde{H}_{N} f_M-g_M\widetilde{H}_{N}g_M\big) \leqslant \frac{\kappa_2}{M^2}\big(\widetilde{H}_{N}+\kappa_2 N\big).
	\end{split}
	\end{equation}
\end{lemma}
\begin{proof}
	To prove \eqref{eq:localization_bogoliubov}, we apply Proposition \ref{prop:localization} with $A=\bH-\lambda(\bH)\geqslant 0$ and $\sigma=2$. All is needed is the computation of the corresponding $A_0$. We notice that
	\begin{equation*}
	\sum_{L=0}^{\infty}\cP_L\big(\bH-\lambda(\bH)\big) \cP_L=U^*\d\Gamma(\cB_1)U-\lambda(\bH),
	\end{equation*}
	where
	\begin{equation*}
	\cB_1=\begin{pmatrix}
	h^{(1)}+c_1K^{(1)}  & \sqrt{c_1c_2} K^{(12)}\\
	\sqrt{c_1c_2} K^{(12)*} & h^{(2)}+c_2K^{(2)} 
	\end{pmatrix},
	\end{equation*}
	$\d\Gamma(\cdot)$ is defined in \eqref{eq:second_quantization_big} and $U$ is given by Theorem \ref{thm:isomorphism}. Since the $K^{(j)}$'s are bounded, there exists $\widetilde{\kappa_1}>0$ (depending on $\lambda(\bH)$) such that
	\begin{equation*}
	\sum_{L=0}^\infty \cP_L\big(\bH-\lambda(\bH)\big) \cP_L\leqslant \d\Gamma^{(1)}(h^{(1)})+\d\Gamma^{(2)}(h^{(2)})+\widetilde{\kappa_1}\cN,
	\end{equation*}
	and hence, by Theorem \ref{thm:upper_lower_bound}, there exists ${\kappa_1}>0$ such that
	\begin{equation*}
	\sum_{L=0}^\infty \cP_L\big(\bH-\lambda(\bH)\big) \cP_L\leqslant {\kappa_1}(\bH+{\kappa_1}).
	\end{equation*}
	The claim is then proven thanks to \eqref{eq:localization}.
	
	To prove \eqref{eq:localization_many_body} we apply Proposition \eqref{prop:localization} with $A=\widetilde{H}_{N}-\lambda(\widetilde{H}_{N})$ and $\sigma=2$. To compute the $A_0$ corresponding to $\widetilde{H}_N-\lambda(\widetilde{H}_N)$, we first note that Theorem \ref{thm:truncated_bound} implies the existence of ${\kappa_3}>0$ such that $\widetilde{H}_N\leqslant {\kappa_3}(\bH+{\kappa_3})$ on $\cF_+^{\leqslant M}$ with $M\leqslant N$. Hence,
	\begin{equation}  \label{eq:diagonal}
	\begin{split}
	\sum_{L=0}^N\cP_L \widetilde{H}_{N}\cP_L&\leqslant {\kappa_3}\sum_{L=0}^N\cP_L(\bH+{\kappa_3})\cP_L\\
	&\leqslant \kappa_{3}'\big( \d\Gamma^{(1)}(h^{(1)})\restriction_{\cF_+^{\,\leqslant N}}+\d\Gamma^{(2)}(h^{(2)})\restriction_{\cF_+^{\,\leqslant N}}+\kappa_3'\big),
	\end{split}
	\end{equation}
	where the second inequality is due to \eqref{eq:upper_lower_bound}. Now, as a consequence of Assumptions ($A_1$) and ($A_2^{\mathrm{MF}}$), there exist constants $\eta,\tau>0$ such that the following stability inequality holds
	\begin{equation*}
	H_N^{\operatorname{MF}}\geqslant \eta \Big( \sum_{i=1}^{N_1}h^{(1)}+\sum_{j=1}^{N_2}h^{(2)}\Big)-\tau N.
	\end{equation*}
	Through a conjugation by $U_N$, last estimate can be rewritten as
	\begin{equation} \label{eq:stability}
	\d\Gamma^{(1)}(h^{(1)})+\d\Gamma^{(2)}(h^{(2)})\leqslant \eta^{-1}\widetilde{H}_N+\eta^{-1}\tau N.
	\end{equation}
	Combining \eqref{eq:stability} with \eqref{eq:diagonal}, we obtain that there exists $\kappa_2'>0$ such that
	\begin{equation} \label{eq:diagonal1}
	\sum_{L=0}^N \cP_L \widetilde{H}_N\cP_L\leqslant \kappa_2'\big( \widetilde{H}_N+\kappa_2' N\big).
	\end{equation}
	Moreover, \eqref{eq:stability} also implies the estimate
	\begin{equation} \label{eq:diagonal2}
	\lambda(\widetilde{H}_N)\geqslant -\tau N
	\end{equation}
	The claim then follows from \eqref{eq:localization}, because \eqref{eq:diagonal1} and \eqref{eq:diagonal2} imply
	\begin{equation*}
	\sum_{L=0}^N \cP_L\big( \widetilde{H}_N-\lambda(\widetilde{H}_N)\big)\cP_L\leqslant \kappa_2\big( \widetilde{H}_N+\kappa_2 N\big)
	\end{equation*}
	for a suitable constant $\kappa_2$.
\end{proof}

\subsection{Validity of Bogoliubov correction}

Now we are ready to conclude the proof of Theorem \ref{thm:mf_correction}. The proof of parts (i)-(ii) has been provided in previous sections (see Theorems \ref{thm:leading-H} (i) and Theorem \ref{thm:upper_lower_bound}). Now we concentrate on parts (iii)-(iv).

\bigskip

\noindent
{\bf Energy upper bound.} We start by proving an upper bound for the ground state energy, namely $\lambda(\widetilde{H}_N)\leqslant \lambda(\bH)+o(1)$. Using Lemma \ref{eq:localization_bogoliubov}, Theorem \ref{thm:truncated_bound}, and the trivial estimate $\bH\geqslant\lambda(\bH)$, we get the following inequality, valid on the space $\cD(\bH)\cap\cF_+^{\,\leqslant N}$, for $1\leqslant M\leqslant N$:
\begin{equation} \label{eq:partial_upper_bound}
\bH\geqslant f_M\Big[\Big(   1+C\sqrt{\frac{M}{N}}\Big)^{-1}\widetilde{H}_N-C\sqrt{\frac{M}{N}}\Big]f_M+\lambda(\bH)g^2_M-\frac{\kappa_1}{M^2}(\bH+\kappa_1).
\end{equation}
Since, by construction, the function $g$ satisfies $g^2(x)\leqslant 2x$, we have
\begin{equation*}
 g^2_M\leqslant\frac{2\cN}{M},
\end{equation*}
and, using the estimate $\cN\leqslant C(\bH+C)$ which follows from Theorem \ref{thm:upper_lower_bound}, we get 
\begin{equation*}
\langle g^2_M\rangle_{\Phi^{\operatorname{gs}}}\leqslant\frac{C}{M}
\end{equation*}
for some constant $C$ depending on $\lambda(\bH)$. This implies that, eventually for $M$ and $N$ large enough, $\langle f^2_M\rangle_{\Phi^{\operatorname{gs}}}>0$. Hence, after taking the expectation value of \eqref{eq:partial_upper_bound} on $\Phi^{\operatorname{gs}}$, we are allowed to divide both sides of the outcome by $\langle f^2_M\rangle_{\Phi^{\operatorname{gs}}}$ and rearrange terms using $f^2+g^2=1$; what we get is
\begin{equation}
\lambda(\bH)\geqslant \Big(1+C\sqrt{\frac{M}{N}}\Big)^{-1}\frac{\langle f_M\Phi^{\operatorname{gs}},\widetilde{H}_Nf_M\Phi^{\operatorname{gs}}\rangle}{\langle f^2_M\rangle_{\Phi^{\operatorname{gs}}}}-C\sqrt{\frac{M}{N}}-\frac{\kappa_1}{M^2\langle f^2_M\rangle_{\Phi^{\operatorname{gs}}}}\big(\lambda(\bH)+\kappa_1\big).
\end{equation}
Now, in the first summand in the r.h.s. we can exploit the fact that the energy of $f_M\Phi^{\operatorname{gs}}$ is certainly bigger than the ground state energy of $\widetilde{H}_N$. For the third summand in the right, in turn, we can use the estimate $1-C/M\leqslant\langle f^2_M\rangle_{\Phi^{\operatorname{gs}}}\leqslant 1$. What we obtain is
\begin{equation*}
\lambda(\bH)\geqslant \Big(1+C\sqrt{\frac{M}{N}}\Big)^{-1} \lambda(\widetilde{H}_N)- C\sqrt{\frac{M}{N}} -\frac{C}{M^2},
\end{equation*}
for a large enough $C>0$. We can optimize last inequality by choosing $M=N^{1/5}$, and this yields the upper bound
\begin{equation}
\lambda(\widetilde{H}_N)\leqslant \lambda(\bH)+CN^{-2/5}.
\end{equation}

\bigskip

\noindent
{\bf Energy lower bound.} We now prove the lower bound $\lambda(\widetilde{H}_N)\geqslant \lambda(\bH)-o(1)$. Using \eqref{eq:localization_many_body}, Theorem \ref{thm:truncated_bound}, and the trivial estimate $\widetilde{H}_N\geqslant \lambda(\widetilde{H}_N)$, we get the inequality
\begin{equation} \label{eq:initial_lower_bound}
\begin{split}
\widetilde{H}_N\geqslant f_M\Big[\Big(1-C\sqrt{\frac{M}{N}}\Big)\bH-C\sqrt{\frac{M}{N}}\Big]f_M+\lambda(\widetilde{H}_N)g^2_M-\frac{\kappa_2}{M^2}\big( \widetilde{H}_N+\kappa_2 N\big).
\end{split}
\end{equation}
We are going to take the expectation value of last inequality on $\Phi_N$, which is a ground state of $\widetilde{H}_N$ defined in \eqref{eq:gs_tilde}. Hence, by definition, we will have $\langle \widetilde{H}_N\rangle_{\Phi_N}=\lambda(\widetilde{H}_N)$.

Moreover, by Theorem \ref{thm:leading-H}, since $\psi_N^{\operatorname{MF}}=U_N^*\Phi_N$ is a ground state of $H_N^{\operatorname{MF}}$, it exhibits condensation in the sense of \eqref{eq:H-1pdm-CV}. Such property directly implies that
\begin{equation*}
\lim_{N\to\infty}\frac{\langle \cN\rangle_{\Phi_N}}{N}=0.
\end{equation*}
This, together with the fact that the function $g$ satisfies $g^2(x)\leqslant 2x$, yields
\begin{equation*}
\langle g_M^2\rangle_{\Phi_N}\leqslant \frac{ 2 \langle\cN\rangle_{\Phi_N}}{M}\mathop{\longrightarrow}_{N\to\infty}0,
\end{equation*}
provided $M$ is chosen such that $\langle\cN\rangle_{\Phi_N}\ll M\ll N$. Last formula implies in particular that, for $N$ large enough and for $M$ in the chosen regime,
\begin{equation*}
\langle f^2_M\rangle_{\Phi_N}>0.
	\end{equation*}
Hence, after taking the expectation value of \eqref{eq:initial_lower_bound} on $\Phi_N$, we are allowed to divide by $\langle f^2_M\rangle_{\Phi_N}$ and rearrange terms using $f^2+g^2=1$. The result is
\begin{equation} \label{eq:partial_lower_bound}
\begin{split}
\lambda(\widetilde{H}_N)\geqslant &\;\Big(1-C\sqrt{\frac{M}{N}}\Big)\frac{\langle f_M\Phi_N,\bH f_M\Phi_N\rangle}{\langle f^2_M\rangle_{\Phi_N}}\\
&\:-C\sqrt{\frac{M}{N}}-\frac{\kappa_2}{M^2\langle f^2_M\rangle_{\Phi_N}}\big(\lambda(\widetilde{H}_N)+\kappa_2 N\big).
\end{split}
\end{equation}
Now, $\langle g^2_M\rangle_{\Phi_N}\to0$ implies $\langle f^2_M\rangle_{\Phi_N}\to1$, and hence the second summand on the l.h.s. of \eqref{eq:partial_lower_bound} converges to zero. In the first summand in the r.h.s. we can certainly estimate from below the energy of $f_M\Phi_N$ with the ground state energy $\lambda(\bH)$. Finally, thanks to \eqref{eq:diagonal2} and $\langle f^2_M\rangle_{\Phi_N}\to1$, the third summand on the right converges to zero if $M$ is chosen such that 
\begin{equation*}
\max\{\sqrt{N},\langle \cN\rangle_{\Phi_N}\}\ll M \ll N.
\end{equation*}
By the last three remarks, \eqref{eq:partial_lower_bound} produces
\begin{equation}
\lambda(\widetilde{H}_N)\geqslant \lambda(\bH)-\delta_N,
\end{equation}
with $\lim_{N\to\infty}\delta_N=0$.

\bigskip

\noindent
{\bf Ground state convergence.} 	As in the proof of the lower bound, let us consider a ground state $\Phi_N$ of $\widetilde{H}_N$. Then, by the estimate $g^2(x)\leqslant 2x$ and the condensation result \eqref{eq:H-1pdm-CV}, we have
	\begin{equation} \label{eq:vanishing_g}
	\lim_{N\to\infty}g_M\Phi_N=0,
	\end{equation}
	provided we choose $M$ such that $\max\{\sqrt{N},\langle \cN\rangle_{\Phi_N}\}\ll M\ll N$. The convergence we want to prove is
	\begin{equation}
	\lim_{N\to\infty} \Phi_N=\Phi^{\operatorname{gs}},
	\end{equation}
	and, thanks to \eqref{eq:vanishing_g}, it is proven if we show
	\begin{equation}
	\lim_{N\to\infty}f_M\Phi^{}_N=\Phi^{\operatorname{gs}}
	\end{equation}
	with $M$ in the regime we already fixed.
	
	First, due to the upper and lower bounds proven above, we have
	\begin{equation*}
	\lambda(\bH)\leqslant\frac{\langle f_M\Phi_N^{\operatorname{}},\bH f_M\Phi_N^{\operatorname{}}\rangle}{\langle f^2_M\rangle_{\Phi_N^{\operatorname{}}}}\leqslant \lambda(\widetilde{H}_N)+\delta_N\leqslant \lambda(\bH)+\delta_N+CN^{-2/5},
	\end{equation*}
	which, together with $\langle f^2_M \rangle_{\Phi_N^{\operatorname{}}}\to 1$, implies	
	\begin{equation} \label{eq:convengence_lower_eigenvalue}
	\lim_{N\to\infty}\langle f_M\Phi_N^{\operatorname{}},\bH f_M\Phi_N^{\operatorname{}}\rangle=\lambda(\bH).
	\end{equation}
	
	Now, let us decompose $f_M\Phi_N^{\operatorname{}}$ into the component along $\Phi^{\operatorname{gs}}$ and the component along its orthogonal complement, namely
	\begin{equation*}
	f_M\Phi_N= a_N \Phi^{\operatorname{gs}}+\Phi^{\perp}_N,
	\end{equation*}
	for a coefficient $a_N\in\mathbb{C}$ and a vector $\Phi_N^\perp\in\cF_+$ such that $\Phi^{\perp}_N\perp\Phi^{\operatorname{gs}}$. Then, since $\Phi^{\operatorname{gs}}$ is an eigenvector of $\bH$, we obtain
	\begin{equation}
	\begin{split}
	\langle f_M\Phi_N^{\operatorname{}},\bH f_M\Phi_N^{\operatorname{}}\rangle=&\; \;|a_N|^2\langle \Phi^{\operatorname{gs}},\bH \Phi^{\operatorname{gs}}\rangle+\langle \Phi^\perp_N,\bH \Phi^\perp_N\rangle\\
	\geqslant &\; \lambda(\bH)|a_N|^2+\inf\sigma(\bH_{|\{\Phi^{\operatorname{gs}}\}^\perp})\|\Phi^\perp\|^2\\
	=&\;\|f_M\Phi_N\|^2\lambda(\bH)+\big(\inf\sigma(\bH_{|\{\Phi^{\operatorname{gs}}\}^\perp})-\lambda(\bH)\big)\|\Phi^\perp\|^2.
	\end{split}
	\end{equation}
	Due to \eqref{eq:convengence_lower_eigenvalue} and \eqref{eq:gap}, we conclude $\lim_{N\to\infty}\Phi^{\perp}=0$, which is equivalent to
	\begin{equation*}
	\lim_{N\to\infty}\|f_M \Phi_N-a_N\Phi^{\operatorname{gs}}\|=0.
	\end{equation*}
	The latter convergence in Fock space norm is equivalent to the convergence in the $j,k$-th sector, with $j+k\leqslant M$
	\begin{equation} \label{eq:L2_convergence}
	\lim_{N\to\infty}\Big\|a_N^{-1}\big(f_MU_N\psi_N^{\operatorname{MF}}\big)_{jk}-\Phi^{\operatorname{gs}}_{jk}\Big\|_{L^2(\mathbb{R}^3)^{\otimes_{\operatorname{sym}}j}\otimes L^2(\mathbb{R}^3)^{\otimes_{\operatorname{sym}}k}}=0.
	\end{equation}
	Here we have used $|a_N|\to1$.
	
	Now, since $\psi_N^{\operatorname{MF}}$ is the ground state of a Schr\"odinger operator, thanks to the diamagnetic inequality we can fix its phase so as to have $\psi_N^{\operatorname{MF}}\geqslant 0$ pointwise almost everywhere. Hence, the function $(f_MU_N\psi_{N}^{\operatorname{MF}})_{jk}$ is non-negative as well, because it is obtained by integrating $\psi_N^{\operatorname{MF}}$ against the positive functions $u_0$ and $v_0$ ($f_M$ contributes only by a non-negative multiplicative factor). Since the $L^2$-convergence in \eqref{eq:L2_convergence} implies pointwise convergence a.e., we deduce that $a_N$ must have a limit $e^{\mathrm{i}\theta}$. If we include this global phase factor inside $\Phi^{\operatorname{gs}}$, we deduce that
	\begin{equation}
	\lim_{N\to\infty}f_M\Phi_N=\Phi^{\operatorname{gs}},
	\end{equation}
	which, thanks to \eqref{eq:vanishing_g}, implies
	\begin{equation}
	\lim_{N\to\infty}U_N\psi_N^{\operatorname{MF}}=\Phi^{\operatorname{gs}}.
	\end{equation}
	By the definition of $U_N$, the latter is equivalent to the desired convergence \eqref{eq:convergence_gs}. \qed


\appendix

\section{Quantum de Finetti Theorem \ref{thm:deF}} \label{sect:deF}

In this Appendix we sketch a proof of Theorem \ref{thm:deF}, following the strategy in \cite{LewNamRou-14} in the one-component case. 
First, let us start with a finite dimensional version. 

\begin{lemma}[Quantum de Finetti theorem in finite dimensions] \label{thm:deF-finite-dim}Let $\cK$ be a Hilbert space with $\dim \cK=d<\infty$. Let $\Psi_{N}$ be a wave function in $\cK^{\otimes N_1}_{\rm sym} \otimes \cK^{\otimes N_2}_{\rm sym}$. Then there exists a Borel probability measure $\mu_{N}$ supported on the set 
	$$\{(u,v):u,v\in \cK, \|u\|=\|v\|=1\}$$
	such that
	$$
	\Tr \left| \gamma^{(k,\ell)}_{\Psi_{N}} - \int |u^{\otimes k}\rangle \langle u^{\otimes k}| \otimes |v^{\otimes \ell}\rangle \langle v^{\otimes \ell}| \d \mu_{N} (u,v) \right| \leqslant Cd\Big( \frac{k_1}{N_1}+ \frac{k_2}{N_2}\Big) 
	$$ 
	for all $k_1\in \{0,1,2,...,N_1\}, k_2\in \{0,1,...,N_2\}.$
\end{lemma}

This is the two-component analogue of the quantitive quantum de Finetti theorem in \cite{ChrKonMitRen-07} (see also \cite{Gottlieb-05,LewNamRou-15} and the references therein for related results).

\begin{proof} Recall the Schur formula  
	$$
	\int_{\|u\|=1} |u^{\otimes N_1}\rangle \langle u^{\otimes N_1}|  \d u= c_{N_1}\1_{\cK^{\otimes N_1}_{\rm sym}}
	$$
	where $\d u$ is the (normalized) Haar measure on the unit sphere $\{u\in \cK,\|u\|=1\}$ and
	$$ c_{N_1}:=\dim \Big( \cK^{\otimes N_1}_{\rm sym} \Big) = {N_1+d-1 \choose d-1}.$$
	From this and a similar identity for $\cK^{\otimes N_1}_{\rm sym}$, we can write 
	\begin{align*}
	\iint_{\|u\|=1,\|v\|=1} |u^{\otimes N_1}\rangle \langle u^{\otimes N_1}| \otimes |v^{\otimes N_2}\rangle \langle v^{\otimes N_2}| \d u \d v = c_{N_1}c_{N_2} \1_{\cK^{\otimes N_1}_{\rm sym}} \otimes \1_{\cK^{\otimes N_2}_{\rm sym}}.
	\end{align*} 
	The latter representation suggests that a natural candidate for $\mu_{N}$ is the Husimi measure
	$$
	\mu_{N}(u,v)= c_{N_1}c_{N_2} \left|\langle u^{\otimes N_1} \otimes v^{\otimes N_2}, \Psi_{N} \rangle\right|^2 \d u \d v.
	$$
	The rest is similar to the proof of the one-component case in \cite{ChrKonMitRen-07} or \cite{LewNamRou-15}. 
\end{proof}

Now we are ready to give
\begin{proof}[Proof of Theorem \ref{thm:deF}] {\bf Step 1: Finite dimensional case.} First we consider the case in which $\cK$ is finite dimensional. By Lemma \ref{thm:deF-finite-dim}, from the wave function $\Psi_{N}$ we can construct a Borel probability measure $\mu_{N}$ supported on the set 
	$$\{(u,v):u,v\in \cK, \|u\|=\|v\|=1\}$$
	such that
	$$
	\lim_{N\to \infty} \Tr \left| \gamma^{(k,\ell)}_{\Psi_{N}} - \int |u^{\otimes k}\rangle \langle u^{\otimes k}| \otimes |v^{\otimes \ell}\rangle \langle v^{\otimes \ell}| \d \mu_{N} (u,v) \right| =0, \quad \forall k,\ell=0,1,2,...
	$$
	On the other hand, since $\{\mu_{N}\}$ is a sequence of Borel probability measures supported on a compact set,  up to a subsequence, $\mu_N$ converges  to a Borel probability measure $\mu$ on $\{(u,v):u,v\in \cK, \|u\|=\|v\|=1\}$. This ensures that
	$$
	\lim_{N\to \infty} \Tr \left| \gamma^{(k,\ell)}_{\Psi_{N}} - \int |u^{\otimes k}\rangle \langle u^{\otimes k}| \otimes |v^{\otimes \ell}\rangle \langle v^{\otimes \ell}| \d \mu (u,v) \right| =0, \quad \forall k,\ell=0,1,2,...
	$$
	
	\noindent
	{\bf Step 2: Infinite dimensional case.} Let $\{\varphi_n\}_{n=1}^\infty$ be an orthonormal basis of $\cK$. Let $P_n$ be the projection onto the subspace $W_n =  {\rm span}(\varphi_1,...,\varphi_n)$. 
	
	Since the operators $\gamma_{\Psi_{N}}^{(k,\ell)}$ is bounded in trace class uniformly in $N$, up to a subsequence,  we have
	$$
	\gamma_{\Psi_{N}}^{(k,\ell)} \wto \gamma^{(k,\ell)}
	$$
	weakly-* in trace class for all $k,\ell\geqslant 0$.  Consequently, for every $n\in \mathbb{N}$ fixed, we have the strong convergence
	\bq \label{eq:PGammaN-PGamma}
	P_n^{\otimes k+\ell}  \gamma_{\Psi_{N}}^{(k,\ell)} P_n^{\otimes k+\ell} \to P_n^{\otimes k+\ell}  \gamma^{(k,\ell)} P_n^{\otimes k+\ell},\quad \forall k,\ell \geqslant 0.
	\eq

	Now using the geometric localization method in Fock space of \cite{Lewin-11}, we can find a state $\Gamma_{N,n}$ in the Fock space $\cF(W_n) \otimes \cF(W_n)$, located in the sectors of $\leqslant N$ particles, whose reduced density matrices are
	$$
	\Gamma_{N,n}^{(k,\ell)}= P_n^{\otimes k+\ell}  \gamma_{\Psi_{N}}^{(k,\ell)} P_n^{\otimes k+\ell}, \quad \forall \, 0\leqslant k\leqslant N_1, 0\leqslant \ell \leqslant N_2.
	$$
	Since $W_n$ is finite dimensional, we can argue as in Step 1 for the state $\Gamma_{N_,n}$ to find a Borel probability measure $\mu_{n}$ supported on the set 
	$$\{(u,v):u,v\in W_n, \|u\|=\|v\|=1\}$$ 	such that
		\begin{align}\label{eq:PGammaN-dmuN}
	&\lim_{N\to \infty} \Tr \left| P_n^{\otimes k+\ell}  \gamma_{\Psi_{N}}^{(k,\ell)} P_n^{\otimes k+\ell} \right.\nn\\
	&- \left.  \Big( \Tr \Big[ P_n^{\otimes k+\ell}  \gamma_{\Psi_{N}}^{(k,\ell)} P_n^{\otimes k+\ell} \Big] \Big) \int |u^{\otimes k}\rangle \langle u^{\otimes k}| \otimes |v^{\otimes \ell}\rangle \langle v^{\otimes \ell}| \d \mu_{n} (u,v)  \right|=0.
	\end{align}
	From \eqref{eq:PGammaN-PGamma} and \eqref{eq:PGammaN-dmuN}, we deduce that
		\begin{align} \label{eq:PGamma-dmu}
	P_n^{\otimes k+\ell}  \gamma^{(k,\ell)} P_n^{\otimes k+\ell} =  C_{k,\ell,n}\int |u^{\otimes k}\rangle \langle u^{\otimes k}| \otimes |v^{\otimes \ell }\rangle \langle v^{\otimes \ell}| \d \mu_{n} (u,v) , \quad \forall k,\ell\geqslant 0
	\end{align}
	where 
	$$
	C_{k,\ell,n} = \Tr\Big[P_n^{\otimes k+\ell}  \gamma^{(k,\ell)} P_n^{\otimes k+\ell}\Big].
	$$
	
	Next, note that if $m\geqslant n$, then the measure $\mu_{n}$ is the cylindrical projection $(\mu_{m})_{|W_n \oplus W_n}$. Therefore, according to \cite[Lemma 1]{Skorokhod-74}, there exists a Borel probability measure $\mu$ supported on 
	$$\{(u,v):u,v\in \cK, \|u\| \leqslant 1, \|v\| \leqslant 1\}$$
	such that for all $n=1,2,...$, the measure $\mu_n$ coincides with the cylindrical projection $\mu_{|W_n \oplus W_n}$.  Consequently, \eqref{eq:PGamma-dmu} can be rewritten as
		\begin{align*} 
	 P_n^{\otimes k+\ell}  \gamma^{(k,\ell)} P_n^{\otimes k+\ell} &=  \int |(P_nu)^{\otimes k}\rangle \langle (P_n u)^{\otimes k}| \otimes |(P_n v)^{\otimes \ell}\rangle \langle (P_n v)^{\otimes \ell}| \d \mu (u,v) \\
	& =  P_n^{\otimes k+\ell} \Big(  \int |u^{\otimes k}\rangle \langle u^{\otimes k}| \otimes |v^{\otimes \ell}\rangle \langle v^{\otimes \ell}| \d \mu(u,v) \Big) P_n^{\otimes k+\ell}.
	\end{align*}
	Since $P_n\to \1_{\cK}$ as $n\to \infty$, we deduce that 
	\begin{align} \label{eq:Gamma-dmu-final}
	\gamma^{(k,\ell)} =  \int |u^{\otimes k}\rangle \langle u^{\otimes k}| \otimes |v^{\otimes \ell}\rangle \langle v^{\otimes \ell}| \d \mu(u,v) , \quad \forall k,\ell\geqslant 0.
	\end{align} 
	
	\noindent 
	{\bf Step 3: Strong convergence.} If we assume further that $\gamma_{N}^{(1,0)}$ and $\gamma_{N}^{(0,1)}$ converge strongly in trace class, then $\Tr \gamma^{(1,0)}= \Tr \gamma^{(0,1)}=1.$ 	By taking the trace of \eqref{eq:Gamma-dmu-final}, we get
	$$
	\int \|u\|^2 \d \mu(u,v) = \int \|v\|^2 \d \mu(u,v) =1.
	$$
	Thus we can conclude that $\mu$ is supported on 
	$$\{(u,v):u,v\in \cK, \|u\| = \|v\| =1\}.$$
	Moreover, by \eqref{eq:Gamma-dmu-final} again we have $
	\Tr \gamma^{(k,\ell)}=1$, and hence $\gamma_{N}^{(k,\ell)}$ converges  to $\gamma^{(k,\ell)}$ strongly in trace class for all $k,\ell \geqslant 0.$
\end{proof}

\def\cprime{$'$}

\end{document}